\documentclass[10pt,a4paper]{article}
\usepackage{a4wide}

\usepackage[affil-it]{authblk}

\usepackage{amssymb, amsfonts, amsbsy, latexsym}
\usepackage{graphicx}
\usepackage{tikz}
\usepackage{color}
\usepackage{colortbl}
\usepackage[english]{babel}
\usepackage{url}
\usepackage{amsmath}

\usepackage{bm}

\usepackage{scalerel}

\usepackage{relsize}

\usepackage{amsthm}

\usepackage{blkarray}

\newtheorem{theorem}{Theorem}
\newtheorem{definition}[theorem]{Definition}
\newtheorem{proposition}[theorem]{Proposition}
\newtheorem{claim}[theorem]{Claim}
\newtheorem{remark}[theorem]{Remark}
\newtheorem{lemma}[theorem]{Lemma}
\newtheorem{example}[theorem]{Example}
\newtheorem{assumption}[theorem]{Assumption}

\newtheorem{approximation problem}[theorem]{Approximation problem}

\newtheorem{corollary}[theorem]{Corollary}

\newcommand{\ip}[2]{\left\langle#1,#2\right\rangle}
\newcommand{\abs}[1]{\left|#1\right|}
\newcommand{\norm}[1]{\left\|#1\right\|}

\def\bT{\breve{T}}

\def\bP{\breve{P}}
\def\bp{\breve{\pi}}
\def\bQ{\breve{Q}}

\def\hf{\hat{f}}

\def\hg{\hat{g}}

\def\Re{{\rm Re}}
\def\Im{{\rm Im}}

\def\w{\omega}
\def\d{\delta}
\def\e{\epsilon}

\def\k{\kappa}

\def\s{\sigma}

\def\NN{\mathbb{N}}
\def\cA{\mathcal{A}}

\def\cF{\mathcal{F}}

\def\cH{\mathcal{H}}

\def\cX{\mathcal{X}}
\def\cY{\mathcal{Y}}

\def\CC{\mathbb{C}}
\def\NN{\mathbb{N}}
\def\ZZ{\mathbb{Z}}

\def\RR{\mathbb{R}}
\def\SS{\mathbb{S}}

\begin{document}

\title{Uncertainty principles \\ and optimally sparse wavelet transforms}

\author{Ron Levie \ \ \ \ \ \ \ \ \ \ \ \ \ \ \ \ \ \ \ \ \     Nir Sochen \\
{\small ronlevie@post.tau.ac.il \ \ \ \ \ \ \ \ \ \ \ \ \ \ \ \ \ sochen@post.tau.ac.il}}
\affil{Tel Aviv University}
\date{\vspace{-5ex}}

\maketitle

\begin{abstract}
 
In this paper we introduce a new localization framework for wavelet transforms, such as the 1D wavelet transform and the Shearlet transform.
Our goal is to design nonadaptive window functions that promote sparsity in some sense. For that, we introduce a framework for analyzing localization aspects of window functions. Our localization theory diverges from the conventional theory in two ways. First, we distinguish between the group generators, and the operators that measure localization (called observables). Second, we define the uncertainty of a signal transform 
 as a whole, instead of defining the uncertainty of an individual window.
We show that the uncertainty of a window function, in the signal space, is closely related to the localization of the reproducing kernel of the wavelet transform, in phase space. 
As a result, we show that using uncertainty minimizing window functions, results in representations which are optimally sparse in some sense. 

\end{abstract}

\section{Introduction}
In this paper we consider ``generalized wavelet transforms'', namely signal transforms based on taking the inner product of the input signal with a set of transformations of a window function. Such a transform is defined as follows. Let the Hilbert space $\cH$ be the space of signals to be transformed. Let $G$ be a manifold, called phase space, and let $\pi:G\rightarrow {\cal U}(\cH)$ be a a strongly continuous mapping from $G$ to unitary operators in $\cH$. Consider a Radon measure of $G$, and take $L^2(G)$ as the output signal space of the generalized wavelet transform.
A generalized wavelet transform is defined by
\begin{equation}
V_{f}:\cH\rightarrow L^2(G) \quad ,\quad \big[V_f[s]\big](g)= \ip{s}{\pi(g)f}.
\label{eq:1}
\end{equation}
Here, $s\in\cH$ denotes the signal we input to the transform. The vector $f\in \cH$, called a window function, is a part of the definition of the transform. The independent variable of the output signal $V_f(s)\in L^2(G)$ is denoted by $g$. 
Some examples are the short time Fourier transform (STFT) \cite{Time_freq}, the continuous wavelet transform \cite{Cont_wavelet_original}\cite{Ten_lectures}, the Shearlet transform \cite{Shearlet}, the Curvelet transform \cite{Curvelet}, and the dyadic wavelet transform \cite{Ten_lectures}. 
The first three examples are based on a square integrable representation of a group $G$. Such representations are sometimes called continuous wavelet transforms, but in this paper we reserve this name to the classical 1D continuous wavelet transform. Signal transforms based on square integrable representations where extensively studied, see e.g the classical paper \cite{gmp} and the more recent book \cite{Fuhr_wavelet}. The latter two examples, namely the Curvelet transform and the discrete wavelet transform, are not based on a group $G$.
We give additional restrictions on $V_f$ in our framework in Assumptions \ref{ass_voice} and \ref{ass_voice2}. 

Many generalized wavelet transforms provide a sparse or optimal representation for their respective
classes of signals, in the following sense. Consider a discretization of a generalized wavelet transform. Namely, assume that there is some sampling $\{g_n\}_{n\in\NN}\subset G$, such that the mapping
\[d_{f}:\cH\rightarrow l^{2}\quad ,\quad [d_f(s)]_n= \ip{s}{\pi(g_n)f}=\big[V_f[s]\big](g_n)\] 
has the inversion formula
\[s=\sum_{n\in\NN}[d_f(s)]_n \pi(g_n)y\]
where $y$ is some other window function.
Let $s_N$ be the approximation of the signal $s$ using only the $N$ largest wavelet coefficients $[d_f(s)]_n$ of $s$. For the optimality statement of the 1D wavelet transform, consider the class of piecewise smooth functions. The asymptotic behavior of the $N$-best approximation is  
\[\norm{s-s_N}_2^2\ \asymp\  N^{-2} \quad , \quad N\rightarrow\infty\]
which is the best approximation rate possible when approximating piecewise smooth functions with dictionaries \cite{Nonlinear_approximation}. Similar optimality properties were proved for the Curvelet transform \cite{optimal_Curvelet} and the Shearlet transform \cite{optimal_Shearlet} in a class of 2D piecewise-smooth signals called ``cartoon-like images''.

Such optimality properties
are independent of the specific choice of the window function. Yet, choosing different window
functions for a given generalized wavelet transform may lead to signal representations with different properties.
In this paper we address the question of how to design and analyze window functions. 
Our goal is to design window functions that promote sparsity, or ``localized'' representations, in some non-asymptotic sense. We do this by defining uncertainty principles.
{ Our localization driven optimality concept complements the above approximation rate optimality property, rather than compete with it. Indeed, the choice of the window in the approximation rate optimality property is a degree of freedom, and thus adding a loss function is required in order to obtain a unique optimal wavelet transform. The loss functions in our theory are new uncertainty measures. }
 In Section \ref{A one parameter localization framework} we define an uncertainty of windows, and in Section \ref{The global localization framework} we strengthen the definition to obtain an uncertainty of a generalized wavelet transform as a whole. { For this, we set down a theory for defining wavelet transforms through measurements of physical quantities, or signal attributes. Many examples of general wavelet transforms are intuitively interpreted as procedures of measuring physical quantities. For example, the STFT measures the content of signals at different times and frequencies, and the 1D wavelet transform measures the content at different times and scales. The idea in our construction is to systematically define the physical quantities underlying a general wavelet transform.}

{
Physical quantities are defined as ``simple'' Lie groups of complex numbers, and $G$ is assumed to be a set of tuples of physical quantities. For each physical quantity we define an operator that measures this quantity, namely an observable. The observables are the link connecting the structure of $G$ and its representation $\pi$ with the uncertainty in measuring the physical quantities.
The systematic approach for defining the observables of the physical quantities underlying general wavelet transforms lends itself to a definition of uncertainty which is compatible with the structure of $G$. 
Taking an uncertainty minimizing vector as the window of the generalized wavelet transform, leads to signal transforms that map to a function space in phase space with optimally localized reproducing kernels. This is discussed in Section \ref{Uncertainty minimizers as sparsifying windows}}. Moreover, we explain how this window choice leads to the ``sparsest'' signal representation possible for the respective generalized wavelet transform. 

{The notion of optimal sparsity in our context is not related to the standard $N$-best approximation rate, and is explained next. In general, a wavelet transform based on a square integrable representation of a group $G$, is an isometry which is not onto $L^2(G)$. Thus, for each signal $s$, there are many functions $F\in L^2(G)$ that synthesize $s$ via $V_f^*[F]$. We consider a class of signals that can be synthesized by a delta train in phase space, $\sum_{j=1}^Jc_j\delta\big(g_j^{-1}(\cdot)\big)$, namely signals of the form $\sum_{j=1}^Jc_j\pi(g_j)f$. The wavelet transform of such a signal is not its corresponding delta train, but rather a blurring of this delta train, obtained by the convolution of the delta train with some filter kernel. Our optimality notion corresponds to windows that result in as little blurring as possible, by which they preserve as much as possible the separation of the peaks of the delta train. Such a property is useful when one wants to recover the delta train in phase space from the wavelet transform of $s$.}

{ We note that conventional variance based uncertainty principles of general wavelet transforms do not have properties relating them to localization in phase space. Our theory unifies variance based uncertainties, defined in the signal domain, with ambiguity function localization, defined in phase space. We see this as a validation that our newly defined uncertainties indeed quantify meaningful notions of localization.}

\subsection{Motivation for defining new uncertainty principles}
\label{Motivation for defining new uncertainty principles}

In the STFT there is a well understood framework for analyzing the localization of window functions. 
The STFT measures the content of the signal at different times and frequencies. This is done by translating the time and modulating the frequency of the window function and taking the inner product of the transformed window with the signal. Namely, in (\ref{eq:1}) $\cH=L^2(\RR)$, $G=\RR^2$ and $\pi(g_1,g_2)f(t)=e^{ ig_2 t}f(t-g_1)$. Thus a prevailing approach for localization, is to analyze window functions in term of their time and frequency variances. The mean time of a normalized $f$ is defined as $e_1=\int_{-\infty}^{\infty}t\left|f(t)\right|^2 dt$, the mean frequency of $f$ is defined as $e_2=\int_{-\infty}^{\infty}\omega\left|\hat{f}(\omega)\right|^2 d\omega$ (where $\hat{f}$ is the Fourier transform of $f$). The spreads of $f$ around its mean time and mean frequency are defined as the variances $\s_1=\int_{-\infty}^{\infty}(t-e_1)^2\left|f(t)\right| dt$ and $\s_2=\int_{-\infty}^{\infty}(\omega-e_2)^2\left|\hat{f}(\omega)\right|^2 d\omega$ respectively.
A good window function in this approach is one that has small spreads both in time and frequency. Quantitatively, the uncertainty of a window function is defined as the product of its variances in time and frequency, and we wish to find a window function with minimal uncertainty.  The smaller the uncertainty of a window function is, the more ``accurately'' the window probes the content of the signal at different times and frequencies simultaneously.
The Heisenberg uncertainty principal poses a lower bound on the uncertainty of any signal, and a classical result states that the modulated Gaussians are the optimal window functions, in the sense that they have minimal uncertainty \cite{Time_freq}.

There is a conventional generalization of this localization analysis framework to generalized wavelet transforms. To describe this approach, we first reformulate the localization framework of the STFT described above, and then perform the conventional abstraction. In the reformulation, we observe that the STFT is based on a square integrable representation of the Heisenberg group. We consider two one parameter subgroups, the subgroup of translations defined by $\pi_1(g_1)f(t)=f(t-g_1)$ and the subgroup of modulations defined by $\pi_2(g_2)f(t)=e^{ig_2t}f(t)$, and note that $\pi(g_1,g_2)=\pi_2(g_2)\pi_1(g_1)$ for any transformation parameters $(g_1,g_2)\in\RR^2$. Next we obtain the infinitesimal generators $T_1$ and $T_2$ of these one parameter unitary groups. Namely $\pi_1(g_1)=e^{ig_1 T_1}$ and $\pi_2(g_2)=e^{ig_2 T_2}$ where the infinitesimal generators are defined by $T_1 f(t)=i\frac{\partial}{\partial t}f(t)$ and $T_2 f(t)=tf(t)$. To define the localization concepts, we adopt the quantum mechanical notion of an observable (for more on observables see Subsection \ref{Sec_obs}). 
\begin{definition}
An \textbf{observable} is a self-adjoint or unitary operator $\breve{T}$ in $\cH$. The expected value and the variance of a normalized vector $f\in\cH$ with respect to $\breve{T}$ are defined to be 
\begin{align}
e_f(\breve{T})&=\ip{\breve{T}f}{f}\\
\s_f(\breve{T})&=\norm{\big(\breve{T}-e_f(\breve{T})\big)f}^2
\label{eq:3}
\end{align}

respectively. 
\end{definition}
When the vector $f\in\cH$ is not normalized, we still use the notations (\ref{eq:3}). In this case, $e_f(\breve{T})$ and $\s_f(\breve{T})$ are no longer interpreted as expected value and variance.

In our case of the time-frequency localization, the suitable time and frequency observables are defined by $\breve{Q}f(t)=tf(t)$ and $\breve{P}f(t)=-i\frac{\partial}{\partial t}f(t)$. Indeed, our ``probabilistic'' concepts of time and frequency localization, defined by $e_1,e_2,\s_1$ and $\s_2$ above, coincide with the quantum mechanical localization notions $e_f(\breve{Q}),e_f(\breve{P}),\s_f(\breve{Q})$ and $\s_f(\breve{P})$ respectively. 

The construction of the general localization framework stems from the observation that the pair of infinitesimal generators $T_1,T_2$ coincide with the pair of observable $\breve{P},\breve{Q}$ up to sign. Indeed, the conventional localization analysis framework for generalized wavelet transforms is based on the following scheme. Consider a signal transform based on a square integrable representation of a Lie group $G$, consider a  set of linearly independent infinitesimal generators $T_1,\ldots,T_n$ of the group of transformations $\pi(G)$, and take them as the observables.
This approach can be found in the literature, e.g \cite{Coherent}\cite{2D_wave}, and in papers, e.g \cite{Affine_uncertainty}\cite{Affine_uncertainty2}\cite{Teschke2005}.
The variances of the window function $f$ are defined to be $\s_f(T_1),\ldots,\s_f(T_n)$. Consider two observables $T_k,T_l$. The product $\s_f(T_k)\s_f(T_l)$ is called the uncertainty of $f$ with respect to $T_k,T_l$.  Let us treat the case where $T_k,T_l$ are self-adjoint.
The general uncertainty principle states \cite{folland_unc} 
\begin{equation} \label{eq:genuncertainty}
\s_f(T_k)\s_f(T_l) \geq \frac{1}{4}\ip{[T_k,T_l]f}{f}^2.
\end{equation}
		Here, $[T_k,T_l]=T_kT_l-T_lT_k$ denotes the commutator of $T_k$ and $T_l$. 
In order to obtain an optimal window function with respect to $\s_f(T_k),\s_f(T_l)$, the conventional procedure is to solve the uncertainty equality
\begin{equation}
\s_f(T_k)\s_f(T_l) = \frac{1}{4} \left|\ip{f}{[T_k,T_l]f}\right|^2.
\label{Uequaloty}
\end{equation}
A classical result states that the solution of (\ref{Uequaloty}) satisfies 
$(T_k-a)f=ic(T_l-b)f$ 
for some $a,b,c\in\RR$ \cite{folland_unc}.

In \cite{do_unc} it was indicated that substituting equality in the uncertainty principle instead of inequality, does not lead to a window function with minimal uncertainty in general. Instead, in order to find a window with minimal uncertainty, one should minimize the uncertainty of $f$ using variational methods.

Applying the above procedure, with variational methods for the minimization problem instead of equation (\ref{Uequaloty}), leads to uncertainty minimizing window functions. These optimal window functions were never applied in engineering. Indeed the results are quite strange and counter intuitive, e.g \cite{ss}.
Our assertion is that the conventional generalization is flawed, in the sense that the derived localization notions do not correspond to the ``metaphysical concept'' of localization. An obvious example follows. Consider the signal transform $L^2(\RR)\rightarrow L^2(\RR)$ based on the one parameter group of time translations $\pi(g)f(t)=f(t-g)$. Applied to a signal $s$, the corresponding transform (\ref{eq:1}) returns the convolution of the signal with the window function (up to a complex conjugation and reflection). An obvious choice of an observable is the time observable $\breve{Q}$: the less spread in time the window is, the more accurately it probes the signal at different times. But note that the infinitesimal generator of the time translations is the operator $i\frac{\partial}{\partial t}$ which coincides up to sign with the \textbf{frequency} observable $\breve{P}$. This observable is inadequate for measuring time localization. 

In this paper we define new notions of uncertainty for generalized wavelet transforms. We illustrate how our definitions encapsulate the notion of locality, by connecting our notion of uncertainty to sparsity. Some preliminary results were published in \cite{Adjoint}.

\subsection{Heuristic derivation of the framework}
Let us start by discussing the STFT again.
Consider a window function $f$ positioned at time and frequency zero. Namely, $e_f(\breve{Q})=e_f(\breve{P})=0$.
The STFT can be interpreted as a procedure of measuring the signal content of $s$ at different values of time and frequency $(g_1,g_2)$. This probing of $s$ is calculated by the inner product $\ip{s}{\pi_2(g_2)\pi_1(g_1)f}$.  

Let us set forth some important ingredients that help lead the way to a generalization. First, there are two underlying physical quantities in the STFT, namely \textit{time} and \textit{frequency}. The parameter $g_1$ of $\pi_1$ corresponds to different values of time. The time values are numbers in $\RR$, and they have a ``natural'' Lie group structure, namely $\RR$ with addition. Indeed, time delaying by $g_1$ and then by $g'_1$, results in a time delay of $g_1+g'_1$. 
Thus we define the physical quantity \textit{time} as the Lie group $\{\RR,+\}$. A similar construction holds for frequency. The two Lie groups of physical quantities are accompanied by two maps that represent them as unitary operators. \textit{time} is accompanied by $\pi_1$ that maps each time $g_1$ to the operator that time-translate by $g_1$. \textit{frequency} is accompanied by $\pi_2$ that maps each frequency value to a modulation. Next, to each physical quantity there is a corresponding observable, in our case the time observable $\breve{Q}$ and the frequency observable $\breve{P}$. These observables are tailored to the unitary operators translation and modulation in the following sense. Using our notion of mean time $e_f(\breve{Q})$, it is easy to verify that time-translating $f$ by $g_1$ changes the mean time of $f$ by $g_1$. Namely, $e_{\pi_1(g_1)f}(\breve{Q})=e_{f}(\breve{Q})+g_1$, and similarly for mean frequency $e_{\pi_2(g_2)f}(\breve{P})=e_{f}(\breve{P})+g_2$. Hence the interpretation of $\pi(g_1,g_2)$ as an operator that changes the time and frequency of window functions.
The \textbf{Heisenberg point of view} of quantum mechanics states the following.  ``Translating a window $f$ by applying a unitary operator $U$  is equivalent to keeping the window constant and translating the observable by conjugating it with $U$''. More accurately, for unitary $U$ and any $\bT$
\begin{align}
\label{Hpov1}
 e_{U f}(\bT) &= \ip{\bT U f}{U f}
   &=\ip{U^*\bT U f}{f}=e_f\big(U^*\bT U \big) 
\end{align}
and
\begin{align}
\label{Hpov2}
\s_{U f}(\bT) &= \norm{\big(\bT -e_{U f}(\bT)\big)U f}^2 \\
 &=\ip{U^*\Big(\bT -e_{f}\big(U ^*\bT U \big)\Big)^*UU^*\Big(\bT -e_{f}\big(U ^*\bT U \big)\Big)U f}{f}\\
 &=\ip{\Big(U^*\bT U  -e_{f}\big(U ^*\bT U \big)\Big)f}{\Big(U^*\bT U -e_{f}\big(U ^*\bT U \big)\Big)f}= \s_{f}\big(U ^*\bT U\big).
\end{align}
In our case, for $U=\pi_1(g_1)$ and $\bT=\breve{Q}$, it is easy to verify that 
\begin{equation}
\pi_1(g_1)^*\breve{Q}\pi_1(g_1)= g_1 I+\breve{Q}.
\label{eq:4}
\end{equation}

To interpret (\ref{eq:4}) we turn to notions from spectral theory (for more on spectral theory and observables see Subsection \ref{Sec_obs}). First note that the spectrum of $\breve{Q}$ is $\RR$, which is the set of possible values of the Lie group of time. The spectrum of an observable corresponds to the set of possible outcomes of measurements by this observable. In particular $e_f(\breve{Q})$ is always in the convex hall of the spectrum of $\breve{Q}$. To interpret the right hand side of (\ref{eq:4}), note that the spectral family of projections of $\breve{Q}$ coincides with the spectral family of projections of $\breve{Q}+g_1I$, but the value in the spectrum to which each spectral projection corresponds is translated by $g_1$. In the language of eigenvectors and eigenvalues, which is ill suited in this case but helps illustrate the situation, the set of eigenvectors of the observable is kept constant but the eigenvalues are translated by $g_1$. We interpret  the spectral family of projections of an observable as the ``physical dimension'' of the observable (to be defined precisely in Subsection \ref{Sec_obs}). In $\breve{Q}$ for example, the spectral family of projections partitions $\cH$ to subspaces having windows with different time supports. Taking all of the above into account, (\ref{eq:4}) is interpreted as follows. Transforming the observable $\breve{Q}$ by its corresponding unitary operator $\pi_1(g_1)$, is equivalent to keeping the physical dimension of time intact, while translating the values of time by $g_1$. Namely, time-translations do not change the very definition of what time is, but only change the values of time. 

Let us add one last note before we generalize. Note that the $+$ sign in the right hand side of (\ref{eq:4}) corresponds to the group rule in the Lie group of the physical quantity time. Thus, if we want to generalize (\ref{eq:4}) to other Lie groups, with $\bullet$ denoting the group multiplication, (\ref{eq:4}) should take the form
\begin{equation}
\pi_1(g_1)^*\breve{T}_1\pi_1(g_1)= g_1\bullet\breve{T}_1 .
\label{eq:5}
\end{equation}
where $g_1$ is in the Lie group of the physical quantity, $\pi_1$ is a unitary representation of the physical quantity and $\breve{T}_1$ is the unknown observable corresponding to $\pi_1$. 
We call (\ref{eq:5}) the one
parameter canonical commutation relation, and study it in Section \ref{A one parameter localization framework}.
Note that in (\ref{eq:5}) we multiply $\breve{T}_1$ by a group element using the group multiplication, which may seem ill defined. However, in spectral theory this operation has a precise meaning. 
In the next subsection we offer a short discussion on spectral theory and observables. 

In case $G$ is a group with physical quantities as subgroups, and $\pi$ is a representation, the canonical commutation relation (\ref{eq:5}) can be extended to a canonical commutation relation of the group as a whole, namely
\begin{equation}
\pi(g)^*{\bf\breve{T}}\pi_1(g)= g\bullet{\bf\breve{T}}.
\label{eq:555}
\end{equation}
Here, ${\bf\breve{T}}$ is a tuple of observables to be defined in Subsection \ref{The general global localization framework}. We call (\ref{eq:555}) the multi-canonical commutation relation, and study it in Section \ref{The global localization framework}. In Section \ref{Uncertainty minimizers as sparsifying windows} we show how the uncertainty defined for the multi-canonical observable ${\bf\bT}$ is correlated with the sparsifying capability of the generalized wavelet transform. For this, we introduce a model for sparse signals in the context of generalized wavelet transforms. For a sparse signal $s$, the smaller the uncertainty of a window function $f$ is, the more $V_f[s]$ is sparse in some sense.

\subsection{Observables} 
\label{Sec_obs}

In this paper an observable in a separable Hilbert space $\cH$ is a self-adjoint or unitary operator. We denote observables with capital letters with a ``breve'', e.g $\breve{T}$. In the following discussion we show how to interpret a self-adjoint or unitary operator, which is a mapping of vectors to vectors, as an entity that defines and measures physical quantities. The interpretation relies on the spectral theorem. 

An observable in an infinite dimensional separable Hilbert space $\cH$ does not admit an eigen-decomposition in general. Instead it admits a more subtle notion of spectral decomposition, called a projection-valued Borel measure, or PVM. 
 This form of spectral decomposition of unitary or self-adjoint operators is guaranteed by the  spectral theorem. We begin by defining a PVM in our case of unitary or self-adjoint operators. 
\begin{definition}
\label{defi_PVM}
Let $\cH$ be a separable (complex) Hilbert space, and $\SS$ be $\RR$ or $e^{i\RR}$. Let ${\cal B}$ be the standard Borel $\sigma$-algebra of $\SS$, and let ${\cal P}$ be the set of orthogonal projections in $\cH$. A mapping $P:{\cal B}\rightarrow {\cal P}$ is called a \textbf{projection valued Borel measure (PVM)} if
\begin{enumerate}
	\item $P(\SS)=I$ and $P(\emptyset)=0$.
	\item
	If $\{B_n\}_{n\in\NN}$ is a sequence of pairwise disjoint Borel sets, then for every $k\neq j$, $P(B_k)$ and $P(B_j)$ are projections to two orthogonal subspaces, and
	\[P(\bigcup_{n\in\NN}{B_n})=\sum_{n\in\NN}P(B_n).\] 
\end{enumerate}
\end{definition}

Next we describe how a PVM $P:{\cal B}\rightarrow {\cal P}$ is interpreted as a physical quantity, that we shall call \textit{quantity}$_P$. The set $\SS$ is interpreted as a set of numbers that contains the possible values that \textit{quantity}$_P$ can take. For any Borel set $B$ of values of \textit{quantity}$_P$, $P(B)$ is interpreted as the projection upon the subspace ${\rm{Image}}(P(B))\subset\cH$ having windows with values of \textit{quantity}$_P$ in $B$.
Let us make this interpretation concrete with an example. In $\cH=L^2(\RR)$ and $S=\RR$, consider the PVM $P$ that maps every Borel set $B$ to the projection upon the space of functions having support in $B$. Namely for any $f\in L^2(\RR)$ and $x\in\RR$
\begin{equation}
 [P(B)f](x)=\left\{ 
\begin{array}{ccc}
	f(x) & \rm{if} & x\in B \\
	0  & \rm{if} & x\notin B .
\end{array}\right.
\label{pvm_Q}
\end{equation} 
This PVM corresponds to the physical quantity \textit{time} (or position). Indeed, $P(B)$ projects to the space of windows with \textit{time} support in $B$.

The above notion of a physical quantity based on a PVM is related to the notion of an observable by the spectral theorem. The theorem states that any self-adjoint or unitary operator corresponds to a unique PVM (modulu sets of measure 0 in $\cal B$) and vice versa.
Here, we present a ``Riemann-Stieltjes'' formulation of the spectral theorem (see e.g \cite{Spectral_Hilbert}).

\begin{theorem}
\label{spectral theorem}
Let $\breve{T}$ be a self-adjoint or unitary operator in the separable Hilbert space $\cH$. Let $\SS=\RR$ in case $\bT$ is self-adjoint, and $\SS=e^{i\RR}$ in case $\bT$ is unitary.
Then,
there is a PVM, $P:\SS\rightarrow {\cal P}$
such that
\begin{enumerate}
	\item 
	\begin{equation}
\breve{T}=\int_{\SS}\lambda \ dP(\lambda)
\label{eq:spectral_theorem}
\end{equation}
where the integral in (\ref{eq:spectral_theorem}) is defined as follows.
Let ${\bf x}=\{x_0,\ldots,x_n\}$ denote a finite Riemann partition of $\SS$. Let ${\rm diam}({\bf x})$ denote the maximal diameter of intervals in ${\bf x}$. Denote by $[x_k,x_{k+1}]$ a general interval in the partition, and by $[x_0,x_n]$ the union of the intervals of ${\bf x}$. We have
\[\bT= \lim_{ \scriptsize
\begin{array}{c}
	{\rm diam}({\bf x})\rightarrow 0 \cr
	[x_0,x_n]\rightarrow \SS
\end{array}
} \sum_{k=0}^{n-1} x_k P\big([x_k,x_{k+1})\big)  \]
where the limit is in the strong topology in case $\bT$ is unbounded (and thus self-adjoint), and in the operator norm topology otherwise.
	\item
	For $\SS=\RR$,
	$f$ is in the domain of $\bT$ if and only if
	\begin{equation}
	\int_{\RR}\lambda^2 \ d\norm{P\big((-\infty,\lambda]\big)f}^2<\infty
	\label{eq:spectral_domain1}
	\end{equation}
	where the integral in (\ref{eq:spectral_domain1}) is the Riemann-Stieltjes integral with respect to the weight function $\norm{P\big((-\infty,\ \cdot\ ]\big)f}^2:\RR\rightarrow\RR$.
\end{enumerate}

\end{theorem}

\begin{remark}(Functional calculus)
\label{Band_limit_poly}
A smooth function $\phi:\SS\rightarrow\CC$ of an observable $\bT$ is defined to be the normal operator
\begin{equation}
\phi(\bT)=\int_{\SS}\phi(\lambda) dP(\lambda).
\label{eq:functional_calculus1}
\end{equation}
defined on the domain of vectors $f\in\cH$ satisfying 
\[\int_{\RR}\abs{\phi(\lambda)}^2 \ d\norm{P\big((-\infty,\lambda]\big)f}^2<\infty.\]
This definition is consistent with polynomials of $\bT$ in the following sense. If $f\in \cH$ is band-limited, namely there exists some compact subset $B\subset\SS$ such that $f=P(B)f$, and if  $\{q_n\}_{n\in\NN}$ is a sequence of polynomials satisfying 
\[\lim_{n\rightarrow\infty}\norm{p_n-\phi}_{L^{\infty}(B)}=0\]
then
\[\lim_{n\rightarrow\infty}\norm{p_n(\bT)f - \phi(\bT)f}_{\cH}=0\]
where $p_n(\bT)$ is in the sense of compositions, additions, and multiplication by scalars of $\bT$, and $\phi(\bT)$ is in the sense of (\ref{eq:functional_calculus1}).
\end{remark}

We say that two observables $\bT_1,\bT_2$, with the same set of values $\SS$, are dimensionally equivalent, if $\bT_1=\phi(\bT_2)$ where $\phi:\SS\rightarrow\SS$ is a diffeomorphism (smooth, with smooth inverse).  
A \textbf{physical dimension} is an equivalence class of dimensionally equivalent observables.
This definition is intuitive. For example,   what makes the time observable an observable of time is its spectral family of projections, and not the specific value corresponding to each projection.

Returning to  (\ref{eq:5}), the group multiplication of $\bT$ by $\lambda$ satisfies
\[ g\bullet\breve{T}  = \int_{\SS} g\bullet\lambda  \ dP(\lambda) \]
where $\breve{T} = \int_{\SS} \lambda \ dP(\lambda)$.
  In other words, transforming the observable by its corresponding operator keeps the physical dimension intact, and only changes the values of the physical quantity.

\section{The one parameter localization framework} 
\label{A one parameter localization framework}

\subsection{Definition of the framework}

We are now ready to introduce our first generalization of uncertainty. 
We define a \textbf{physical quantity} as one of the following numerical Lie groups
\[\{\RR,+\} \quad,\quad \{e^{i\RR},\cdot\} \quad,\quad \{\ZZ,+\} \quad,\quad \{e^{2\pi i \ZZ/N},\cdot\} \quad (N\in\NN).\]
This set of Lie groups exhausts up to isomorphism a set of zero dimensional and one dimensional Lie groups that satisfy some regularity conditions (abelian and connected one dimensional-locally compact groups, or cyclic-discrete groups).
{The phase space $G$ in our construction of wavelets is defined to be a manifold direct product of physical quantities, with assumptions listed in Assumption \ref{ass_voice}.}

{ General wavelet transform frameworks usually stem from generalizing and abstracting the two classical examples of the STFT and the continuous 1D wavelet transform. The classical general formulation of general wavelet transforms was developed in \cite{gmp0}\cite{gmp}. There, the space $G$ is a locally compact topological group, and the mapping $\pi$ is a strongly continuous square integrable representation. Our construction stems from a special case of the classical general framework, where we assume that the group $G$ has a manifold direct product structure of physical quantities, while the group structure is not a group direct product in general. In the classical theory, given a square integrable representation, there is a complete characterization of the set of vectors that are allowed to be taken as windows, namely admissible vectors. These are given as the vectors in the domain of a uniquely defined operator, called the Duflo-Moore operator. Moreover, a reconstruction formula of the wavelet transform is given in term of this Duflo-Moore operator. These properties, proven in the classical theory for wavelet transforms based on group representation, are also true in some special transforms which are not based on group representations, such as the Curvelet transform.
While our construction is based on a special case of a square integrable representation, we do allow a slight generalization. For the framework to include the important example of the Curvelet transform, we take the properties of the classical theory as assumptions, rather than having them as theorems resulting from the group representation structure.}
The following list summarizes our assumptions on generalized wavelet transforms.

\begin{assumption}[Generalized wavelet transform]
\label{ass_voice}
A generalized wavelet transform is constructed by, and assumed to satisfy, the following.
\begin{enumerate}
	\item 
	Consider a tuple of physical quantities $G_1,\ldots,G_n$ where $G_k$ is called $quantity_k$. We denote by the same notation $\bullet$ the group product of each $G_k$.
	\item
	$G$ is a (manifold) direct product of the manifolds $G_1,\ldots,G_n$ (note that $G$ is not a group in general).
	\item
	We consider a radon measure $dg$ on the manifold $G$.
	\item
	$\pi_k$ are a strongly continuous unitary (SCU) representations of $G_k$ in the separable Hilbert space $\cH$, $k=1,\ldots,n$. Namely $\pi_k(g_k)$ is a unitary operator in $\cH$ and $\pi_k(g_k\bullet g'_k)=\pi_k(g_k)\circ\pi_k(g'_k)$ for any $g_k,g_k'\in G_k$. Here $\circ$ is composition. 
	\item
	For any $g=(g_1,\ldots,g_k)\in G$, we define $\pi(g)=\pi_1(g_1)\circ\ldots\circ\pi_n(g_n)$ (note that $\pi$ is not a group representation in general).
	\item
	There exists a densely defined positive self-adjoint operator $A$ on $\cH$, with densely defined inverse, such that $V_f[f]\in L^2(G)$ if and only if $f$ is in the domain $\cA$ of $A$. The domain $\cA$, which is dense in $\cH$, is called the space of admissible vectors.  
	In our context we also call $\cA$ \textbf{the window space}, and call vectors in $\cA$ \textbf{windows}.
	\item
	\label{resolution_id}
	Given windows $f,h$, and signals $s,q$, the \textbf{wavelet transform} $V_f:\cH\rightarrow L^2(G)$ defined by $V_f[s](g)=\ip{s}{\pi(g)f}$ satisfies
	\[\ip{Ah}{Af}_{\cH}\ip{s}{q}_{\cH}=\ip{V_f[s]}{V_h[q]}_{L^2(G)}.\]
\end{enumerate}
\end{assumption}

\begin{remark}
$ $
\begin{enumerate}
	\item 
	\label{reconstruction1}
	Item \ref{resolution_id} in Assumption \ref{ass_voice} can be read off as a reconstruction formula. Namely
\begin{equation}
\ip{Ah}{Af}_{\cH}s=\int V_f[s](g) \pi(g)h \ dg
\label{eq:inversionR}
\end{equation}
where the convergence in the definition of the integral is in the weak topology.
\item
$V_f$ is also called the \textbf{analysis operator} corresponding to the window $f$. For a window $h$, $V_h^*$ is called the \textbf{synthesis operator} corresponding to $h$. For $Q\in L^2(G)$, we have
\begin{equation}
V_h^*[Q] = \int Q(g) \pi(g)h dg
\label{eq:wave_inversion1}
\end{equation}
where the integral is defined in the weak sense as in \ref{reconstruction1}.
The reconstruction formula (\ref{eq:inversionR}) can be written in the form
$\ip{Ah}{Af}_{\cH}s=V_h^*V_f[s]$.
\end{enumerate}
\end{remark}

\begin{remark}
\label{ass_voice_square}
An important class of a generalized wavelet transforms is when $G$ is a locally compact topological group, $dg$ is the left Haar measure, and $\pi$ is a square integrable (irreducible) representation.
In this case, sections 6 and 7 of Assumption \ref{ass_voice} are theorems, and $A$ is called the Duflo-Moore operator \cite{gmp0}\cite{gmp}. 
\end{remark}
Generalized wavelet transforms based on square integrable representations include the wavelet and the Shearlet transforms, and the STFT.
An important example in which $G$ is not a group, but Assumption \ref{ass_voice} is still satisfied, is the continuous Curvelet transform.

A generalized wavelet transform is interpreted as a procedure of measuring the content of a signal by probing it at different values of $quantity_1,\ldots,quantity_n$. Since our goal is to measure these values as accurately as possible, we want to define corresponding observables, and notions of localization. 
\begin{definition}
Let $\pi_k$ be a  SCU representation of the physical quantity $G_k$. An observable $\bT_k$ satisfying the \textbf{canonical commutation relation}
\begin{equation}
\pi_k(g_k)^* \bT_k \pi_k(g_k) = g_k\bullet \bT \quad , \quad \forall g_k\in G_k 
\label{eq:5a}
\end{equation}
is called a \textbf{canonical observable} of $\pi_k$.
\end{definition}
To each representation $\pi_k$ from Assumption \ref{ass_voice} we define a corresponding canonical observable $\breve{T}_k$.
Once we have the canonical observables, we may define the uncertainty of a window $f$ as some combination of the variances $\{\s_f(\breve{T}_k)\}_{k=1}^n$, and look for an optimal window that minimizes this uncertainty.
For example, in the STFT if we define the uncertainty of a window either as the product or as the sum of it's time and frequency variances, the optimal windows in either case are modulated Gaussians.

\subsection{Analysis of the canonical commutation relation}
\label{Solving the canonical commutation relation0}

In this subsection we show how to restrict the pair $\{\pi,\bT\}$ to a special case, called a canonical system. For canonical systems, there is a procedure for solving the canonical commutation relation (\ref{eq:5a}), given in Subsection \ref{Solving the canonical commutation relation}. We motivate the definition of a canonical system using heuristic arguments on the roles of $\pi$ and $\bT$.
Since in both this section and the next we study a single representation of one physical quantity, we omit subscripts. Namely, we denote the physical quantity by $G$, its representation by $\pi$ and the canonical observable by $\bT$. 

First we recall some basic facts from harmonic analysis \cite{H}.
A character of an abelian group $K$ is a homomorphism $\chi:K\rightarrow \{e^{i\RR},\cdot\}$. The set of characters of $K$, denoted by $\hat{K}$, is an abelian group with the group rule $(\chi\eta)(k)=\chi(k)\eta(k)$ for $\eta,\chi\in\hat{K}$ and $k\in K$, where in the right hand side the multiplication is in $\CC$. The Pontryagin duality states that the group of characters of $\hat{K}$ is isomorphic to $K$. If $K$ is compact then $\hat{K}$ is discrete and vice-versa. 

In our case, the group is a physical quantity $G$.
The following list exhausts the Lie groups of physical quantities and their Pontryagin duals (up to isomorphism). If $G=\{\RR,+\}$ then $\hat{G}= G$, if $G=\{e^{i\RR},\cdot\}$ then $\hat{G}=\{\ZZ,+\}$ and vise-versa, and if $G=\{e^{2\pi i\ZZ/N},\cdot\}$ then $\hat{G}= G$. We assume, with abuse of notation, that $\hat{G}$ is equal to a physical quantity. When we want to treat $\hat{G}$ as a group of characters, we denote it by $\chi(G)$.

Let $\pi$ be a SCU representation of the physical quantity $G$, and let $\bT$ be a corresponding canonical observable.
First we characterize the spectrum ${\rm spec}(\bT)$ of $\bT$.
\begin{proposition} 
\label{Prop_spectrum_of_observ}
Let $\pi$ be a SCU representation of the physical quantity $G$, and let $\bT$ be a canonical observable of $\pi$. Then
\begin{itemize}
	\item 
	If $G=\RR$ or $G=e^{i\RR}$, then ${\rm spec}(\bT)=G$.
	\item
	If $G=\ZZ$ or $G=e^{2\pi i \ZZ/N}$, then ${\rm spec}(\bT)\subset \RR$ or ${\rm spec}(\bT)\subset e^{i\RR}$ respectively.
\end{itemize}
\end{proposition}
\begin{proof}
First, the spectrum of any normal operator is non-empty, so there exists $\lambda_0\in{\rm spec}(\bT)$ where $\lambda_0\in\CC$.
consider the canonical commutation relation
\[\forall g\in G. \quad \pi(g)^*\bT\pi(g)=g\bullet \bT\]
On the one hand, note that conjugating any operator $A$ with a unitary operator, doesn't change the spectrum of $A$, so
${\rm spec}\big(\pi(g^{-1})\bT\pi(g)\big)={\rm spec}\big(\bT\big)$. On the other hand, note that
\[{\rm spec}(\bT)={\rm spec}\big(\pi(g^{-1})\bT\pi(g)\big)={\rm spec}\big(g\bullet\bT \big) = g\bullet{\rm spec}\big(\bT\big)  =\{g\bullet\lambda \ |\ \lambda\in{\rm spec}(\bT)\}\]
This is true for any $g\in G$, so 
\begin{equation}
G\bullet{\rm spec}(\bT) ={\rm spec}(\bT) \quad , \quad  G\bullet\lambda_0 \subset{\rm spec}(\bT).
\label{eq:for spectrum proof}
\end{equation}

As a result of (\ref{eq:for spectrum proof}), the following list exhausts all of the cases of $\bT$ and $G$.
Since $\bT$ is unitary or self-adjoint, ${\rm spec}(\bT)$ is a subset of $e^{i\RR}$ or of $\RR$ respectively.
As a result, if $G=\RR$ or $G=\ZZ$, we must have ${\rm spec}(\bT)\subset \RR$, and if  $G=e^{i\RR}$ or $G=e^{2\pi i \ZZ/N}$, we must have ${\rm spec}(\bT)\subset e^{i\RR}$.
In case $G=\RR$ or $G=e^{i\RR}$, we must have ${\rm spec}(\bT)=G$.
\end{proof}

Since the role of $\bT$ is to measure $quantity_G$, we further demand the following assumption.
\begin{assumption}
\label{ass_sigma_G}
${\rm spec}(\bT)=G$.
\end{assumption}

Next we show that under Assumption \ref{ass_sigma_G}, the roles in the canonical commutation relation (\ref{eq:5a}) of the observable $\bT$ and the representation $\pi$ are interchangeable in some sense. To see this we need to derive an observable from the representation $\pi$, and to generate a representation from the observable $\bT$. 
We start by deriving the observable from $\pi$.
In case $G$ is one-dimensional, by Stone's theorem on one parameter unitary groups, there is a self-adjoint generator $T$ of the unitary $\pi(G)$ \cite{Stone}. Namely, every element of $\pi(G)$ can be written as $e^{it T}$, where $t\in\RR$. In case $G$ is zero-dimensional, there is an element $T\in\pi(G)$ that generates $\pi(G)$. Namely, every element of $\pi(G)$ can be written as $T^n$, where $n\in\ZZ$.
Now, the idea is that the canonical observable $\bT$ can be treated as a generator of a unitary group $\breve{\pi}(\hat{G})$, which can be treated as a representation of the physical quantity $\hat{G}$, whereas the generator $T$ of the unitary group $\pi(G)$ can be taken as a canonical observable of $\breve{\pi}(\hat{g})$. 

We show the construction for the case of $G=\{\ZZ,+\}$. The other cases are treated similarly. 
Define $T:=\pi(-1)$, and note that $T$ generates the unitary group $\pi(G)$.

\begin{claim}
\label{claim_replace4}
$T e^{iq\breve{T}}=e^{iq(\breve{T}+I)}T$.
\end{claim}

\begin{proof}
The canonical commutation relation reads
\[T^{g}\breve{T}T^{-g}=\pi(g)^*\breve{T}\pi(g)=\breve{T}+gI \quad , \quad g\in\ZZ. \]
By taking $g=1$ we have
\[T\breve{T}=(\breve{T}+I)T.\]
Thus by induction
\begin{equation}
T\breve{T}^k=(\breve{T}+I)^kT.
\label{eq:help1}
\end{equation}

The idea now is to use the series expansion of the exponential map and to substitute (\ref{eq:help1}) term by term to get
\begin{equation}
T e^{iq\breve{T}}=T\big(\sum_{k=0}^{\infty}\frac{(iq)^k}{k!}\breve{T}^k\big)=\big(\sum_{k=0}^{\infty}\frac{(iq)^k}{k!}(\breve{T}+I)^k\big)T=e^{iq(\breve{T}+I)}T.
\label{eq:}
\end{equation}
To make this formal, we need a density argument. 
Following Remark \ref{Band_limit_poly}, we consider the space of band-limited signals $\cH_{\rm bl}$ with respect to $\bT$.
By Remark \ref{Band_limit_poly}, for band-limited vectors we have 
\begin{equation}
\lim_{K\rightarrow\infty}\sum_{k=0}^{K}\frac{(iq)^k}{k!}\bT^k\ f = e^{i q\bT}f.
\label{eq:in_claim_4t5}
\end{equation}
By the unitarity of $T$ and by
\[T\big(\sum_{k=0}^{K}\frac{(iq)^k}{k!}\breve{T}^k\big)=\big(\sum_{k=0}^{K}\frac{(iq)^k}{k!}(\breve{T}+I)^k\big)T\]
equation (\ref{eq:in_claim_4t5}) shows that $\sum_{k=0}^{K}\frac{(iq)^k}{k!}(\breve{T}+I)^k y$ converges to $T e^{iq\breve{T}} T^* y$ for any $y$ in the dense subspace $T\cH_{\rm bl}$.
Since the series $\sum_{k=0}^{K}\frac{(iq)^k}{k!}(\breve{T}+I)^k y$ also converges to $e^{iq(\breve{T}+I)}$ in the dense subspace $\cH_{\rm bl}$, by continuity of $T e^{iq\breve{T}} T^*$ and $e^{iq(\breve{T}+I)}$ we must have
\[T e^{iq\breve{T}}T^*f=e^{iq(\breve{T}+I)}f\]
For any $f\in\cH$.
\end{proof}

Let us now define $\breve{\pi}$ and show the canonical commutation relation for our case of $G=\{\ZZ,+\}$.
By Claim \ref{claim_replace4} and since $I$ commutes with every operator, we have 
\begin{equation}
T e^{iq\breve{T}}=e^{iq(\breve{T}+I)}T=e^{iq} e^{iq\breve{T}}T.
\label{eq:last_cheee1}
\end{equation}
Consider the mapping $\{\RR,+\}\rightarrow {\cal U}(\cH)$, $q\mapsto e^{iq\breve{T}}$. 
 By Assumption \ref{ass_sigma_G}, ${\rm spec}(\bT)=\ZZ$, so by Remark \ref{Band_limit_poly}, $e^{iq\breve{T}}=e^{i(q+2\pi)\breve{T}}$ for every $q\in\RR$. Thus we define
\[\breve{\pi}:\{e^{i\RR},\cdot\}\rightarrow {\cal U}(\cH) \quad , \quad \breve{\pi}(e^{iq})=e^{iq\breve{T}}\]
and note that $\breve{\pi}$ is a SCU representation of $\hat{G}$.
To conclude, (\ref{eq:last_cheee1}) can now be written as
\[\forall e^{iq} \in\hat{G}. \quad \breve{\pi}(e^{iq})^*T\breve{\pi}(e^{iq})=e^{iq}\bullet T.\]

Let us now study the spectrum of $T$ in the general case.
In the case where $G$ is one dimensional, by Proposition \ref{Prop_spectrum_of_observ}, ${\rm spec}(T)=\hat{G}$. For the other cases we adopt an assumption
\begin{assumption}
\label{ass_sigma_hat_G}
${\rm spec}(T)=\hat{G}$.
\end{assumption}

We summarize our construction and assumptions in the following definition.
\begin{definition}
\label{def:canonical system}
 $\{G,\pi,T,\hat{G},\breve{\pi},\bT\}$ is called a \textbf{canonical system}, if
 $\pi$ and $\breve{\pi}$ are representations of the physical quantities $G$ and $\hat{G}$ respectively, $T$ and $\bT$ are generators  (or infinitesimal generators) of $\pi(G)$ and $\breve{\pi}(\hat{G})$ respectively satisfying ${\rm spec}(\bT)=G$ and ${\rm spec}(T)=\hat{G}$, and $\bT$ is a canonical observable of $\pi$.
\end{definition}

Note that the representations in a canonical system must be faithful. Otherwise, if for $g\neq e$ in $G$ we have $\pi(g)=I$, then
\[\bT=\pi(g)^*\bT\pi(g)= g\bullet \bT\]
which is a contradiction, since the mapping $g'\mapsto g\bullet g'$ has no fixed points. This is also true for $\breve{\pi}$.
To conculde the above results, the following list exhausts all possibilities of canonical systems.
\begin{proposition}
 Let $\{G,\pi,T,\breve{\pi},\bT\}$ be a canonical system. Then
\label{canonical_pairs}
\begin{enumerate}
	\item 
	If $G=\{\RR,+\}$: $\hat{G}= G$. $T$ and $\breve{T}$ are self-adjoint with ${\rm spec}(T)={\rm spec}(\breve{T})=\RR$. $\pi(g)=e^{ig T}$ and $\breve{\pi}(q)=e^{iq\breve{T}}$ are SCU faithful representations of $\{\RR,+\}$. Here, $g$ denotes elements of $G$ and $q$ denotes elements of $\hat{G}$.
	\item
	If $G=\{e^{i\RR},\cdot\}$: $\hat{G}=\{\ZZ,+\}$. $T$ is self-adjoint with ${\rm spec}(T)=\ZZ$ and $\breve{T}$ is unitary with ${\rm spec}(T)=e^{i\RR}$. $\pi(e^{i \theta})=e^{i\theta T}$ is a SCU faithful representation of $\{e^{i\RR},\cdot\}$, and $\breve{\pi}(n)=\breve{T}^{n}$ is a SCU faithful representation of $\{\ZZ,+\}$. Here, $e^{i\theta}$ with $\theta\in \left[0,2\pi\right)$ denotes elements of $G$ and $n$ denotes elements of $\hat{G}$.
	\item
	If $G=\{\ZZ,+\}$: $\hat{G}=\{e^{i\RR},\cdot\}$. The rest is as in case (2), with the roles of $T,\pi$ and $\breve{T},\breve{\pi}$ interchanged.
	\item
	If $G=\{e^{2\pi i \ZZ/N},\cdot\},\  (N\in\NN)$: $\hat{G}= G$. $T$ and $\breve{T}$ are unitary with ${\rm spec}(T)={\rm spec}(\breve{T})=e^{2\pi i\ZZ/N}$, and $\pi(e^{\frac{2\pi i}{N}n})=T^{n}$ and $\breve{\pi}(e^{\frac{2\pi i}{N}m})=\breve{T}^{m}$ are SCU faithful representations of $\{e^{2\pi i \ZZ/N},\cdot\}$. Here, $e^{\frac{2\pi i}{N}n}$,$e^{\frac{2\pi i}{N}m}$ with $n,m\in (\ZZ \ {\rm mod}\ N)$ denote elements of $G$ and $\hat{G}$ respectively.
\end{enumerate}
\end{proposition}

\subsection{Solving the canonical commutation relation}
\label{Solving the canonical commutation relation}

In this subsection we present a general procedure for finding a canonical observable for a given representation of a physical quantity. 
The construction is guaranteed under the assumption that $\{G,\pi\}$ are members of a canonical system.
We base our construction on the the Stone-von Neumann-Mackey theorem \cite{SVNM_original}, and give a restricted version of the theorem for abelian groups, the proof of which can be found in \cite{SVNM}.

Let us first recall the definition of generalized Heisenberg groups (see e.g \cite{SVNM}).
Let $K$ be a locally compact abelian Lie group, and let $\chi(K)$ be its dual group of characters. Consider the following unitary operators on $L^2(K)$ (where the Haar measure is used to define the inner product). Generalized left translation operators:
\[L(k) f(t)=f(k^{-1}\bullet t) \quad , \quad k\in K\]
and generalized modulation operators:
\[M({\chi}) f(t)=\chi(t)f(t) \quad , \quad \chi\in \chi(K).\]
These operators satisfy the commutation relation
\begin{equation}
[L(k),M({\chi})]=L(k)^*M({\chi})^*L(k)M({\chi})=\overline{\chi(k)}I. 
\label{general_canonical_com}
\end{equation}
Thus, the following set of unitary operators is a Lie group of operators on $L^2(K)$, called the Heisenberg group associated with $K$
\[ J=\{e^{2\pi i t}L(k) M({\chi})\ |\ t\in \left[\left.0,1\right)\right. ,\  k\in K ,\  \chi\in\chi(K)\}. \]
As a unitary group, $J$ has a natural representation on $L^2(K)$, namely $\gamma(h)=h$ for any $h\in J$. We denote elements of $J$ in coordinates by $(t,k,\hat{g})$, where $\hat{g}\in \hat{G}\cong\chi(G)$.

\begin{theorem}[Stone - von Neumann - Mackey]
\label{SVNMS}
$ $
\begin{enumerate}
	\item  The representation $\gamma$ is irreducible. Namely, $L^2(K)$ has no non-trivial proper closed subspace invariant under $J$.
	\item Let $\cH$ be a Hilbert space and $\rho$ an irreducible SCU representation of $J$ in $\cH$, such that $\rho(e^{2\pi i t}I)=e^{2\pi i t}I$ for all $t\in\left[\left.0,1\right)\right.$. Then $\rho$ is unitarily equivalent to $\gamma$. Namely, there exists
	a unique (up to a constant) isometric isomorphism $U:{\cH}\rightarrow L^2(K)$ satisfying
	\begin{equation}
U\rho(h) U^*=\gamma(h) \quad , \quad {\rm for\ any\ } h\in J.
\label{S_V_N_M}
\end{equation}
\item
In case $\rho$ from $2$ is reducible, there exists an orthogonal sum decomposition of Hilbert spaces
\[\cH=\bigoplus_{n\in\kappa}\cH_n\]
where $\kappa$ is a finite or countable index set,
	such that each $\cH_n$ is invariant under $\rho(J)$, and $\rho$ is irreducible in $\cH_n$.
	For each $n\in\kappa$ there exists a unique (up to a constant) isometric isomorphism 
	$U_n:{\cH}_n\rightarrow L^2(K)$ satisfying
	\begin{equation}
U_n\rho(h)|_{\cH_n}U_n^*=\gamma(h)\quad , \quad {\rm for\ any\ } h\in J
\end{equation}
where $\rho(h)|_{\cH_n}$ is the restriction of $\rho$ to $\cH_n$.
\end{enumerate}
\end{theorem}

Next we formulate a uniqueness property of the decomposition in 3 of Theorem \ref{SVNMS}.
It's proof relies on the notion of direct integral decomposition of representations. Since this is the only part in the paper in which we use direct integrals, in the Appendix we only give restricted definitions, limited to our specific needs. For a general exposition we refer the reader to Chapter 3.4 of \cite{Fuhr_wavelet}. Given a representation $\rho(g)$ of $\cH$, and $N\in\NN\cup\{\infty\}$, we denote by $\cH^N$ the direct product of $\cH$ with itself  $N$ times, if $N$ is finite, and define $\cH^N$ to be the space of square summable $\cH$ sequences if $N=\infty$. We denote by $\rho(g)^{[N]}$ the representation in $\cH^N$ defined for $\{f_n\}_{n=1}^N\in\cH^N$ by
\[\rho(g)^{[N]}\{f_n\}_{n=1}^N=\{\rho(g)f_n\}_{n=1}^N.\]

\begin{proposition} 
\label{unique_decompose_size}
Consider two representations $\rho$ and $\rho'$ in the same Hilbert space $\cH$, satisfying 3 of Theorem \ref{SVNMS}, for a physical quantity $K$. Denote by $\rho|_K$ and $\rho'|_K$ the restrictions of the representations $\rho$ and $\rho'$ to the subgroup of translations $K$ of $J$ respectively, and assume $\rho|_K=\rho'|_K$. Let $\kappa$ and $\kappa'$ be the index sets from 3 of Theorem \ref{SVNMS}, corresponding to $\rho$ and $\rho'$ respectively. Then  $\kappa$ and $\kappa'$ are of the same size.
\end{proposition}

\begin{proof}
By 3 of Theorem \ref{SVNMS}, $\rho$ and $\rho'$ are equivalent to the two direct product representations $\gamma^{[{\kappa}]}$ (acting on $L^2(K)^{\abs{\kappa}}$) and $\gamma^{[{\kappa'}]}$ (acting on $L^2(K)^{\abs{\kappa'}}$) respectively. By (\ref{eq:direct_decomposition_of_translation}), the representation $\gamma$ restricted to $K$, $\gamma|_K$, has the direct integral decomposition
\[\gamma|_K(k)\cong \int_{\hat{K}}^{\oplus}\chi^{[1]} \ d\mu(\chi). \]
So 
\[\rho|_K(k)\cong \int_{\hat{K}}^{\oplus} \chi ^{[\kappa]}\ d\mu(\chi) \quad ,\quad \rho'|_K(k)\cong \int_{\hat{K}}^{\oplus} \chi^{[\kappa']} \ d\mu(\chi) \]
By Proposition \ref{unique_direct_int}, the multiplicities in a direct integral decomposition are unique. Therefore, since $\rho|_K(k)= \rho'|_K(k)$, we must have $\abs{\kappa}=\abs{\kappa'}$.
\end{proof}

To bridge the gap between our theory and the Stone - von Neumann - Mackey theorem, we define a representation of Heisenberg groups corresponding to  canonical systems.

\begin{definition}
\label{def:schrodinger}
Let $\{G,\pi,T,\hat{G},\breve{\pi},\bT\}$ be a canonical system, and let $J$ be the Heisenberg group assosiated with $G$. The mapping $\Pi:J\rightarrow {\cal U}(\cH)$, defined by
\begin{equation}
\Pi(t,g,\hg)= e^{2\pi it}\pi(g)\breve{\pi}(\hg), \quad t\in\left[0,1\right),\ g\in G,\ \hg\in\hat{G}.
\label{eq:Heisenberg_observable}
\end{equation}
is called the \textbf{Schr\"odinger representation} of the canonical system $\{G,\pi,T,\breve{\pi},\bT\}$.
\end{definition}

The following proposition shows that Schr\"odinger representations are representations of $J$.

\begin{proposition}
\label{Schro_is_Heis}
Let $\{G,\pi,T,\hat{G},\breve{\pi},\bT\}$ be a canonical system. Then
there exists an isomorphism 
\[\hat{G}\rightarrow \chi(G) \quad , \quad \hg\mapsto \chi_{\hg}\]
such that 
\begin{equation}
\pi(g)^*\breve{\pi}(\hg)^*\pi(g)\breve{\pi}(\hg)=\overline{\chi_{\hg}(g)}I.
\label{H_rep_commutation}
\end{equation}
\end{proposition}

\begin{proof}
Let us treat the case where $G=\hat{G}=\RR$. The other cases are treated similarly.
\[\forall k\in\ZZ_+ . \quad \pi(g)^*\bT^k\pi(g)=\big(\pi(g)^*\bT\pi(g)\big)^k=(\bT+g)^k,\]
so by the series expantion of the exponential map (and using a density argument as before)
hg\[\pi(g)^*\breve{\pi}(\hg)\pi(g)=e^{iq (\bT+g)}=e^{ig\hg}\breve{\pi}(\hg).\]
Therefore
\[\pi(g)^*\breve{\pi}(\hg)\pi(g)\breve{\pi}(\hg)^*=e^{ig\hg}I.\]
By substituting $\hg\mapsto-\hg$, we get
\[\pi(g)^*\breve{\pi}(\hg)^*\pi(g)\breve{\pi}(\hg)=\overline{e^{ig\hg}}I,\]
and the corresponding mapping is $\hg\mapsto \chi_{\hg}$ where $\chi_{\hg}(g)= e^{ig\hg}$.

\end{proof}

By Proposition \ref{Schro_is_Heis}, the Schr\"odinger representation $\Pi$ is a representation of $J$ satisfying the conditions in the Stone - von Neumann - Mackey theorem (Theorem \ref{SVNMS}).
Thus we have the following corollary.
\begin{corollary}
\label{From_Stone}
Let $\{G,\pi,T,\hat{G},\breve{\pi},\bT\}$ be a canonical system, with Schr\"odinger representation $\Pi$. Then there exists an orthogonal sum decomposition of Hilbert spaces
\[\cH=\bigoplus_{n\in\kappa}\cH_n\] such that each $\cH_n$ is invariant under $\Pi$. 
Moreover, in every $\cH_n$, $\Pi(g)|_{\cH_n}$ is unitarily equivalent to $\gamma(g)$ (where $\gamma({g})$ is the natural representation of the Heisenberg group of $G$ in $L^2(G)$). Namely, there exist unique isometric isomorphisms $U_n:\cH_n\rightarrow L^2(G)$ such that
\[\forall g\in G . \quad U_n\Pi(g)|_{\cH_n}U_n^* =\gamma(g).  \]
\end{corollary}

To construct a canonical observable for a given $\pi$, we assume that there exists a canonical system containing $\{G,\pi\}$.
First let us assume that the corresponding $\Pi$ is an irreducible representation of $J$. Corollary \ref{From_Stone} can be utilized as follows. 
Given $\pi,G$, 
we first construct an isometric isomorphism 
$U:\cH\rightarrow L^2(G)$ such that 
\begin{equation}
U\pi(g)U^*=L({g})
\label{eq:IsoIso}
\end{equation}
 for any $g\in G$. A solution of (\ref{eq:IsoIso}) is guaranteed to exist. 
Consider the multiplicative operator in $L^2(G)$, 
\begin{equation}
\bQ_G f(g)=g f(g),
\label{eq:canonical_ob_in_H}
\end{equation} 
where the multiplication in (\ref{eq:canonical_ob_in_H}) is the usual arithmetic multiplication.
It is straightforward to show that $\bQ_G$ is a canonical observable of $L(g)$ in $L^2(G)$.   Now, we can pull back the canonical observable $\bQ_G$ to $\cH$ using $U$. Namely,
\begin{equation}
\breve{T}=U^*\bQ_G U
\label{eq:canonical_construction}
\end{equation}
 is a canonical observable of $\pi$. Indeed,
\begin{align*}
\pi(g)^*\breve{T}\pi(g)&=\pi(g)^*U^*\bQ_G U\pi(g)\\
&=U^*L({g})^*\bQ_G L({g})U
=U^*g\bullet\bQ_G U=g\bullet U^*\bQ_G U = g \bullet \breve{T} .
\end{align*}
where the operation ``$g\bullet (\cdot)$'' commutes with unitary operators since it is either the multiplication by the scalar $g$, or the addition with the scalar operator $gI$.

The following proposition extends this analysis to the reducible case.

\begin{proposition}
\label{help_construct_canoni} 
Let $\pi$ be a SCU representation of $G$, such that there exists a canonical system containing $\{G,\pi\}$.  Then there exists an index set $\k$ of size uniquely defined by $\pi$, a
 decomposition of $\cH$ to invariant subspaces of $\pi(g)$, $\bigoplus_{n\in\kappa}\cH_n$, and a sequence of isometric isomorphisms 
$\{U_n:\cH_n\rightarrow L^2(G)\}_{n\in\kappa}$, such that 
\begin{equation}
U_n\pi(g)|_{\cH_n}U_n^*=L(g)
\label{eq:IsoIso2}
\end{equation}
 for any $g\in G$.  Moreover, for any decomposition $\cH=\bigoplus_{n\in\kappa}\cH_n$ and $\{U_n:\cH_n\rightarrow L^2(G)\}_{n\in\kappa}$ that satisfies the above,
\begin{enumerate}
	\item  The operator
	\begin{equation}
\breve{T}_n=U_n^*\bQ_GU_n
\label{eq:canonical_construction2}
\end{equation}
is a canonical observable of $\pi(g)|_{\cH_n}$.
\item
\label{2_help_construct_canoni}
$\bT=\bigoplus_{n\in\kappa}\bT_n$
is a canonical observable of $\pi(g)$.
\end{enumerate}
\end{proposition}
Note that the uniqueness of $\abs{\k}$ in Proposition \ref{help_construct_canoni}  follows Proposition \ref{unique_decompose_size}.

In practice, 
finding the decomposition $\cH=\bigoplus_{n\in\kappa}\cH_n$, given a representation $\pi$, may seem like a convoluted task. Indeed, this decomposition only makes sense in view of the unknown observables $\breve{T}_n$, since it is a decomposition to irreducible subspaces of the Schr\"odinger representation. In the following discussion we formulate a more accessible version of Proposition \ref{help_construct_canoni}.

Under the assumptions of Proposition \ref{help_construct_canoni}, define the isometric isomorphism 
\begin{equation}
U:\cH\rightarrow L^2(G)^{\abs{\kappa}} \quad , \quad U=\bigoplus_{n\in\kappa}U_n.
\label{eq:isoiso3f4}
\end{equation} 
 Consider the left translation $L(g)^{[\kappa]}:L^2(G)^{\abs{\kappa}}\rightarrow L^2(G)^{\abs{\kappa}}$. 
Consider the multiplicative operator $\bQ_G^{[\kappa]}:L^2(G)^{\abs{\kappa}}\rightarrow L^2(G)^{\abs{\kappa}}$ defined by
\[\bQ_G^{[\kappa]}\left\{F_n(g)\right\}_{n\in\kappa}=\left\{gF_n(g)\right\}_{n\in\kappa}.\]
Proposition \ref{help_construct_canoni} states that $U$ intertwines $\pi$ and $L(g)^{[\kappa_m]}$, and $\bT=U^*\bQ_G^{[\kappa]} U$.
Note that in this construction, the space $L^2(G)^{\abs{\kappa}}$ is isomorphic to the space $L^2(X)=L^2(G\times {\cal Y})$, where ${\cal Y}$ is the standard discrete measure space $\{n\}_{n\in\kappa}$. Under this isomorphism, the left translation $L(g)^{[\kappa]}$ takes the following form in $L^2(X)$. For any $h\in L^2(X)$,
\[L_X(g)h(g',y)=h(g^{-1}\bullet g',y).\]
Moreover, the observable $\bQ_G^{[\kappa]}$ takes the form $\bQ_X h(g,y)=gh(g,y)$ in $L^2(X)$.
Motivated by this observation,
another technique for constructing a canonical multi-observable for a SPWT is explained next. First, find an isometric isomorphisms $\Psi:\cH\rightarrow L^2(G\times{\cal Y})$, where ${\cal Y}$ is some manifold with Radon measure, and $\Psi$ maps $\pi(g)$ to translations $L_X(g)$. Then, consider the multiplicative operator
$\bQ_X:L^2(X)\rightarrow L^2(X)$ defined by
\[\bQ_X h(g,y)= g h(g,y).\]
Last, define the canonical observable of $\pi$ to be
$\bT = \Psi^*\bQ_X \Psi$. This construction guarantees the canonical commutation relations (\ref{eq:5a}). When $\pi$ is a representation of $quantity_G$, we call $\Psi$ the $\boldsymbol{quantity_G}$ \textbf{transform}, and call $L^2(X)$ the $\boldsymbol{quantity_G}$ \textbf{domain}. We summarize this discussion in a theorem.

\begin{theorem}
\label{Theorem:pull_trans1}
Let $\{G,\pi\}$ be members of a canonical system. Then there exists a manifold ${\cal Y}$ with a Radon measure, where for $X=G\times{\cal Y}$ there exists an isometric isomorphism $\Psi:\cH\rightarrow L^2(X)$
 (the $quantity_G$ transform) that intertwines $\pi$ with translations along $G$. Namely, $\pi(g)= \Psi^*L_X(g)\Psi$. 
For any such transform $\Psi$, the observable $\bT = \Psi^*\bQ_X \Psi$ is a canonical observable of $\pi$.
\end{theorem}

\begin{remark}
\label{Remark:pull_trans1}
In the analysis preceding Theorem \ref{Theorem:pull_trans1}, it was shown that there exists a discrete manifold ${\cal Y}$ corresponding to Theorem \ref{Theorem:pull_trans1}. In practice, it is beneficial to consider also non-discrete manifolds ${\cal Y}$. We illustrate how a non-discrete manifold can be constructed in the framework of Theorem \ref{Theorem:pull_trans1} in the following exmample. In case $\kappa=\NN$, we have $\bigoplus_{n\in\NN}\cH_n \cong L^2(G\times\NN)$. It is possible to map $L^2(G\times\NN)$ to $L^2(G\times\RR)$ (with an isometric isomorphism), by using the fact that the space $L^2(\NN)=l^2$ is isometrically isomorphic to $L^2(\RR)$ via an orthogonal basis expansion. In this construction, we consider an orthogonal basis $\{\eta_n\}_{n\in\NN}\subset L^2(\RR)$, and consider the isometric isomorphism
\[W:L^2(G\times\NN)\rightarrow L^2(G\times\RR), \quad ,\quad W\ \{f_n\}_{n\in\NN}=\sum_{n\in\NN}f_n\otimes \eta_n\]
where $[f_n\otimes \eta_n](g,y)=f_n(g)\eta_n(y)$. The $quantity_G$ transform is then related to Proposition \ref{help_construct_canoni} by $\Psi=W\circ U$, where $U$ is defined in (\ref{eq:isoiso3f4}).
\end{remark}
Using non-discrete ${\cal Y}$ spaces simplify the construction in the Curvelet transform and the Shearlet transforms of Subsections \ref{The Curevelet transform} and \ref{The Shearlet transform}.

\subsection{Characterization of the set of canonical observables}

Note that for a representation $\pi$, a canonical representation containing $\{G,\pi\}$ is not uniquely defined. Therefore, a canonical observable is not uniquely defined for a given $\pi$.
The following theorem characterizes the set of all possible canonical observables of a given representation of a physical quantity.

\begin{proposition}
\label{Class_Pirre}
 Consider a canonical system $\{G,\pi,T,\hat{G},\breve{\pi},\bT\}$, represented in $\cH$. Let $\breve{\cal T}$ be the set of observables $\bT'$ in $\cH$ that belong to some other canonical system of the form $\{G,\pi,T,\hat{G},\breve{\pi}',\bT'\}$.
Then
\[\breve{\cal T}=\{U^* \breve{T}  U\ |\ U\in{\mathcal U}(\cH)\ {\rm commutes\ with}\ T\}.\]
\end{proposition}
Note that $U$ commutes with $T$ if and only if  $U$ commutes with $\pi(g)$ for any $g\in G$.

\begin{proof}
For the first direction,
if $U$ commutes with $\pi(g)$ for any $g\in G$, then
\[\pi(g)^*U^*\bT U \pi(g)=U^*\pi(g)^*\bT  \pi(g)U =U^*g\bullet\bT  U=g\bullet U^*\bT U\]

For the other direction, denote by $\{G,\pi,T,\hat{G},\breve{\pi},\bT\}$ the given canonical system, and by $\Pi$ the corresponding Schr\"odinger representation.
Let $\bT'$ be another canonical observable with canonical system $\{G,\pi,T,\hat{G},\breve{\pi}',\bT'\}$ and Schr\"odinger representation $\Pi'$.  By 
Corollary \ref{From_Stone} and by Proposition \ref{unique_decompose_size}, there are two orthogonal sum decomposition of $\cH$
\[\cH=\bigoplus_{n\in\kappa} \cH_n = \bigoplus_{n\in\kappa} \cH'_n \]
with the same index set $\kappa$,
and sequences of isometric isomorphisms $W_n,W'_n$ such that
\[W_n\ \Pi(h)|_{\cH_n}\ W_n^*=\gamma(h) \quad ,\quad W_n'\ \Pi'(h)|_{\cH_n'}\ W_n'^*=\gamma(h)  \]
for any $h\in J$. 
Thus, for any $n$
\begin{equation}
\Pi'(h)|_{\cH_n'}=W_{n}^{\prime *}  W_n\ \Pi(h)|_{\cH_n}\ W_n^*  W'_{n} \quad  \ \ {\rm for\ any\ }\ h\in J.
\label{eq:in_class_obse}
\end{equation}
Here $W_{n}^{\prime *}  W_n$ are isometric isomorphisms $\cH_n \rightarrow \cH'_n$.
Note that for $g$ in the subgroup $G\subset J$, we have $\Pi(g)=\Pi'(g)=\pi(g)$. Thus,
restricting (\ref{eq:in_class_obse}) to the subgroup $G$, we get
\[\pi(g)|_{\cH'_n}=W_{n}^{\prime *}  W_n\ \pi(g)|_{\cH_n}\ W_n^*  W'_{n} \quad  \ \ {\rm for\ any\ }\ g\in G\] 
Consider the unitary operator $U=\bigoplus_{n\in\kappa}  W_{n}^{\prime *}  W_n$. We have
\[\pi(g)=U \pi(g)U^*  \quad  \ \ {\rm for\ any\ }\ g\in G\] 
so $U$ commutes with $\pi$.
Moreover, restricting (\ref{eq:in_class_obse}) to the subgroup $\hat{G}\subset J$ and using $U$, 
we get
\[\breve{\pi}'(g)=U \breve{\pi}(g)U^*  \quad  \ \ {\rm for\ any\ }\ g\in \hat{G}.\] 
This identity holds for the generators as well, and we have
$\bT'=U \bT U^*$.
\end{proof}

\subsection{Examples}

We present three examples of our localization theory. First, the finite STFT (FSTFT) is a version of the STFT used in numerical applications. A standard approach for window design for FSTFT is to consider an optimal window for the continuous STFT, namely a Gaussian, and to discretize it to obtain a window of the finite STFT. Instead, in our approach we formulate the localization framework directly in the finite dimensional signal space.
The second example is the 1D wavelet transform, and is given to motivate the construction in Section \ref{The global localization framework}. Last, we give a localization framework for the Curvelet transform.

\subsubsection{The finite short time Fourier tansform}
\label{Localization framework for finite short-time-Fourier-tansform}

Consider the Heisenberg group $J$ corresponding to $G=\{e^{2\pi i \ZZ/N},\cdot\}$ (see e.g \cite{FSTFT1}). We call $G$ $time$, and $\hat{G}$ $frequency$. We call the center of $J$, which is isomorphic to $\{e^{i\RR},\cdot\}$ $phase$.  Let $\cH=L^2(G)$. Consider the subgroup $J'$ of $J$ having phase in $\{e^{2\pi i \ZZ/N},\cdot\}$, called $reduced\ phase$. We call $J'$ the (classical) \textbf{finite Heisenberg group}. The group $J'$ is isomorphic to the semi-direct product $(reduced\  phase \times time )\rtimes frequency$. Consider the canonical SCU faithful representation of $J'$ in $\cH$
\[\pi(z,g,q)=z\pi_1(g)\pi_2(q) \quad , \quad z,g_1,g_2\in e^{2\pi i \ZZ/N}\]
where $\pi_1(g_1)f(x)=f(g_1^{-1}x)$, $\pi_2(g_2)f(x)=\chi_{g_2}(x)f(x)$. Here 
\[\chi_{e^{2\pi i m /N}}(e^{2\pi i k /N})=e^{2\pi i mk /N},\] 
where $e^{2\pi i m /N},e^{2\pi i k /N} \in e^{2\pi i \ZZ/N}$ are generic elements. The representation $\pi$ is irreducible.
By the fact that $J'$ is unimodular,
the space of admissible functions is $L^2(G)$, and $A=I$ (see e.g \cite{Fuhr_wavelet} 
Theorem 2.25, and Assumption \ref{ass_voice} for square integrable representations).

The natural choices for canonical time and frequency observables are
\[\bQ f(x)=xf(x) \quad , \quad \bP f(x)= f(e^{2\pi i /N}x).\]
Note that $\cF[\bP f](\w)= \bQ \hat{f}(\w)=\w\hat{f}(\w)$.
We define the uncertainty 
\begin{equation}
S(f)=w_1\s_f(\bQ)+w_2\s_f(\bP)
\label{eq:unc_FSTFT}
\end{equation}
where $w_1,w_2\in\RR_+$ are weights.

\subsubsection{The 1D wavelet transform}
\label{1DWaveExample}

The 1D wavelet transform comprises dilations and time translations of a window in $L^2(\RR)$. In this section we recall the canonical observables developed in \cite{Adjoint}.
Positive dilations and time tanslations are defined by
\begin{equation}
\pi_1(g_1)f(t)  =  f(t-g_1) 
\label{eq:timeshift} 
\end{equation}
\begin{equation}
\pi_2(g_2)f(t)  =  e^{-g_2 /2} f\left(e^{-g_2}t\right) 
\label{eq:dil} 
\end{equation}
for $g_1\in G_1=\RR$ and $g_2\in G_2=\RR$.
To include also negative dilations, we introduce the reflection physical quantity $G_2=\{-1,1\}$ with representation $\pi_3(g_3)f(x)= f(g_3 x)$.
Note that often the wavelet transform is defined only with positive dilations, in which case it is not based on a direct sum of two irreducible representation.
The group 
\[G=(time \ translations)\rtimes( positive \ dilations \times  reflections)= \RR\rtimes (\RR\times\{-1,1\})\]
 is the 1D affine group, represented by $\pi(g)=\pi_1(g_1)\pi_2(g_2)\pi_3(g_3)$ in the 1D wavelet transform.
The representation $\pi$ is square integrable (and specifically irreducible) in $L^2(\RR)$.
The Duflo-Moore operator $A$ in  $L^2(\RR)$ is given by
\[[\cF Af](\w) = \frac{1}{\sqrt{\w}}\hf(\w)\ ,\]
and the space of admissible functions  is
\[{\cal A}=\left\{f\in L^2(\RR)\ \Big|\ \int_{-\infty}^{\infty}\frac{1}{\w}\abs{\hf(\w)}^2 d\w \leq\infty\right\}.\]

Let us introduce canonical observables.
The canonical observable for $\pi_3$ can be chosen to be $\cF[\bT_3 f](\w) = {\rm sign(\w)}[\cF f](\w)$. A perfectly localized window $f$ with respect to $\bT_3$ is one with support of $\hf$ in $\RR_+$ or $\RR_-$. Since $\bT_3$ measures the ``weight of the support of $\hf$ in $\RR_{\pm}$'', and since functions with frequency support in $\RR_{\pm}$ correspond to time signals with counterclockwise and clockwise phase respectively, we call $\bT_3$ the phase direction observable.
A natural choice for the canonical observable $\breve{T}_1$ is the time observable
	$\breve{T}_1 f(t) = t f(t)$.
Next, it is accustomed to call the physical quantity represented by dilations \textit{scale}. In \cite{Adjoint} a scale canonical observable $\breve{T}_2$ was defined by
\begin{equation} \label{eq:AdjT2}
	\cF \breve{T}_2 \cF^* \hat{f}(\omega) = - \ln(\abs{\omega}) \hat{f}(\omega).
\end{equation}
 Note that this choice of $\breve{T}_2$ is plausible from a physical point of view. Scale is related to wavelength, so a multiplication operator in the frequency domain is a suitable choice. 

Next we show that our definition of $\bT_2$ is based on a canonical system. The isometric isomorphism of Proposition \ref{help_construct_canoni}  is constructed as follows.
The invariant subspaces of Proposition \ref{help_construct_canoni}  are
\[L_{\pm}^2(\RR)=\{f\in L^2(\RR)\ |\ {\rm support}(\hf) \subset \RR_{\pm}\},\]
where $L^2(\RR)=L^2_+(\RR)\oplus L^2_{-}(\RR)$.
Consider the two warping transforms ${\cal W}_{\pm}: L_{\pm}^2(\RR)\rightarrow L^2(\RR)$ defined by
\begin{equation}
\tilde{f}_{\pm}(c)=[{\cal W}_{\pm}\hf](c)=e^{-c/2}\hf(\pm e^{-c}).
\label{eq:Warping1}
\end{equation}
The inverse warping transforms ${\cal W}^{-1}_{\pm}: L^2(\RR)\rightarrow L_{\pm}^2(\RR)$ are given by 
\[L^2_{\pm}(\RR)\ni\hat{f}_{\pm}(\w)=[{\cal W}^{-1}\tilde{f}_{\pm}](\pm\w)=\abs{\w}^{-\frac{1}{2}}\tilde{f}_{\pm}(-\ln(\abs{\w})).\]
Define the positive and negative scale transforms by
\[U_{\pm}: L_{\pm}^2(\RR)\mapsto L^2(\RR) \quad , \quad U_{\pm}={\cal W}\cF,\]
and define the scale transform, that maps functions in the time domain to the scale domain  by
\[U:L^2(\RR)\rightarrow L^2(\RR)^2 \quad , \quad U=U_+\oplus U_-.  \]
  Define the standard observable in the scale domain
	\[\bQ^{[2]}:L^2(\RR)^2\rightarrow L^2(\RR)^2 \quad , \quad \bQ^{[2]}(\tilde{f}_+(c),\tilde{f}_-(c))=(c\tilde{f}_+(c),c\tilde{f}_-(c)).\]
It is now straight forward to show that $\bT_2=U^*\bQ^{[2]}U$.

\begin{remark}
Let us explain our choice of the physical quantity $G_2$. 
It is possible to define the wavelet transform using dilations defined by
\[\pi'_2(g_2)f(x)=g_2^{-\frac{1}{2}}f(g_2^{-1}x)\]
for $g_2'$ in the group $G_2'=\{\RR_+,\cdot \}$. A canonical scale observable $\bT'$ in this case can be defined by $\cF\bT_2'\cF^{-1}\hf(\w)=\w^{-1}\hf(\w)$. Indeed, by $\cF\pi'_2(g_2)\cF^{-1}\hf=\pi'_2(g_2^{-1})\hf$, we have 
\begin{equation}
\begin{split}
\big[\cF\pi_2'(g_2)^*\bT_2'\pi_2'(g_2)\cF^{-1}\hf\big](\w) =  &  \cF\pi_2'(g_2)^* \cF^{-1}\Big(\w^{-1} g_2^{\frac{1}{2}}\hf(g_2\w)\Big)\\
 =&g_2\w^{-1}\hf(\w) = \big[\cF g_2\bullet\bT_2'\cF^{-1}\hf\big](\w).
\end{split}
\label{eq:}
\end{equation}
However, this construction is not based on a canonical system. For canonical systems, using Proposition \ref{help_construct_canoni}, the discussion can be pulled forwards to $L^2(G_2)^{2}$, where the canonical observable is defined as $\bQ^{[2]}$. Intuitively, it is sensible to define the integral over $G_2$ in the calculation of the expected values and variances, using the Haar measure of $G_2$. To see this, for $\tilde{f}\in L^2(G_2)^2$, we think of $|\tilde{f}(g_2)|^2$ as the signal content at scale $g_2$, we think of $e_{\tilde{f}}(\bQ^{[2]})$ as the center of mass of scales, and think of $\s_{\tilde{f}}(\bQ^{[2]})$ as the spread about the center of mass. Defining $e_{\tilde{f}}(\bQ^{[2]})$ and $\s_{\tilde{f}}(\bQ^{[2]})$ using the Haar measure of $L^2(G_2)^2$ assures that the integral has the interpretation of a sum, or a weighted average, over the group $G_2$. If we use $\bT_2'$ as a canonical observable, the integration in $e_f(\bT_2')$ and $\s_f(\bT_2')$ is not based on the Haar measure. 
\end{remark}

The following list collects some translation laws of the observables $\breve{T}_1$ and $\breve{T}_2$.

\begin{eqnarray}
e_{\pi_1(g_1)f}(\breve{T}_2) & = & e_{f}(\breve{T}_2)  \label{eq:expValaff2} \\
\s_{\pi_1(g_1)f}(\breve{T}_2) & = & \s_{f}(\breve{T}_2)  \label{eq:VarTransaff1} \\
e_{\pi_2(g_2)f}(\breve{T}_1) & = & e^{g_2} e_{f}(\breve{T}_1)  \label{eq:expValaff3} \\
\s_{\pi_2(g_2)f}(\breve{T}_1) & = & e^{2 g_2} \s_{f}(\breve{T}_1)  \label{eq:VarTransaff2} 
\end{eqnarray}

If we ignore the less important phase direction observable, the uncertainty $S(f)$ of a mother wavelet $f$ is defined as the sum or the product of $\sigma_f(\breve{T_1}),\sigma_f(\breve{T_2})$.
Next we recall an asymptotic minimizer of $S$, namely a sequence of windows $f_n$ with uncertainty converging to zero as $n\rightarrow\infty$ \cite{Adjoint}.
The construction is as follows:
\begin{itemize}
\setlength\itemsep{0em}
	\item Choose a two times differentiable bump function $\hat{f}(\omega)$ supported in $(0,1)$. An example is a cubic B-spline
  \item Choose $\kappa(n)$ such that $n=o(\kappa(n))$. Example: $\kappa(n)=n^2$.
  \item Define $f_n$ by 
  \begin{equation}
\hat{f_n}(\omega)=\frac{1}{\sqrt{n}}\hat{f}\left(\frac{\omega - \kappa(n)}{n}\right)
\label{eq:minimiz}
\end{equation}
  and normalize to $\|f_n\|=1$.
\end{itemize}
The following proposition holds:
\begin{proposition}
\label{prop:unc0}
The function system $f_n$ satisfies
\begin{eqnarray}
e_{f_n}(\breve{T}_1) =0\quad\quad &\ &  \s_{f_n}(\breve{T}_1) \stackrel{n \rightarrow \infty}{\longrightarrow} 0 \label{wave_min1}\\
 e_{f_n}(\breve{T}_2) \stackrel{n \rightarrow \infty}{\longrightarrow} -\infty &\ & \s_{f_n}(\breve{T}_2) \stackrel{n \rightarrow \infty}{\longrightarrow} 0. \label{wave_min2}
\end{eqnarray}
\end{proposition}

We draw the following qualitative conclusion from this example: the smaller the scale of a window is, the more simultaneous time-scale localization is possible. 

Here, we want to discuss the shortcomings and limitations of the 1D wavelet uncertainty $S(f)$ as defined above. Note that for large $n$, $e_{f_n}(\breve{T}_2)$ is large and negative, so $f_n$ measures small scales. The measurements of ``macroscopic'' scales in the wavelet transform with the mother wavelet $f_n$ is performed using $\pi_2(g_2)f_n$ with large $g_2$. Note that for measuring macroscopic scales we use $g_2\rightarrow\infty$ as $n\rightarrow\infty$, and in this case $\s_{\pi_2(g_2)f_n}(\breve{T}_1)\rightarrow\infty$. Moreover, $\s_{\pi_2(g_2)f_n}(\breve{T}_2)$ stays constant, so the uncertainty in measuring macroscopic scales of signals using $f_n$ tends to infinity as $n\rightarrow\infty$. To conclude, $f_n$ with large $n$ is a bad mother wavelet for measuring signals having macroscopic scales. 

Let us explain the reason for this bad result. When we construct a wavelet transform we choose a mother wavelet $f$, and take the inner product of the signal $s$ with the set $\{\pi(g)f\ |\ g\in G\}$, called the \textbf{orbit} of $f$. For any other mother wavelet of the form $\pi(g')f$, where $g'\in G$, the wavelet transform is the same up to a right translation in the domain $G$. Indeed
\[V_{\pi(g')f}[s](g)=\ip{s}{\pi(g)\pi(g')f}=V_{f}[s](g\bullet g').\]
 Thus, when analyzing a wavelet transform, the object of interest is not the mother wavelet itself, but the orbit of the mother wavelet.
The standard uncertainty $S(f)=\s_f(\breve{T}_1) \s_f(\breve{T}_2)$ (or $\s_f(\breve{T}_1) + \s_f(\breve{T}_2)$) is a measure of the uncertainty of an individual window $f$, and it is not invariant under the group action of $\pi$ on $f$. Hence, it is not suitable as an uncertainty measure of a wavelet transform as a whole. In Section \ref{The global localization framework} we present a generalization of the time-frequency Heisenberg uncertainty to generalized wavelets, that encapsulates the global uncertainty of the orbit of a window. Such a quantity captures the uncertainty in measuring physical quantities with the signal transform as a whole.

\subsubsection{The Curevelet transform}
\label{The Curevelet transform}

In this subsection boldface lower case letters, e.g ${\bf x}$, denote vectors in $\RR^2$.
The Curvelet transform comprises translations, rotations, and anisotropic dilations of a window in $L^2(\RR^2)$ \cite{Curvelet}.
Translation by ${\bf g}_1\in \RR\times\RR$ is defined as usual by $\pi_1({\bf g}_1)f({\bf x})=f({\bf x}-{\bf g}_1)$. Consider the rotation matrix operator, with $g_2\in e^{i\RR}$,
\[R_{g_2}= \left(
\begin{array}{cc}
	\Re (g_2) & \Im (g_2)\\
	-\Im (g_2) & \Re (g_2)
\end{array}
\right).\] 
Rotaton by $g_2\in\RR$ of $L^2(\RR^2)$ functions is defined by
\[\pi_2(g_2)f({\bf x}) = f(R_{g_2}^{-1}{\bf x}).\]
Consider the anisotropic dilation matrix operator, with $g_3\in\RR$,
\[D_{g_3}= \left(
\begin{array}{cc}
	e^{g_3} & 0\\
	0 & e^{\frac{1}{2}g_3}
\end{array}
\right).\] 
Anisotropic dilation by $g_3\in\RR$ of $L^2(\RR^2)$ functions is defined by
\[\pi_3(g_3)f({\bf x}) = e^{-\frac{3}{4}g_3} f(D_{g_3}^{-1}{\bf x}).\]
Last, reflections by $g_4\in\{-1,1\}$ are represented by
\[\pi_4(g_4)f({\bf x})=g_4f({\bf x}).\]
The Curvelet transform is based on the operators
\[\pi(g)=\pi_1({\bf g}_1)\pi_2(g_2)\pi_3(g_3)\pi_4(g_4).\]

To construct canonical observables, we transform the discussion to the frequency domain. We have
\[\hat{\pi}_1({\bf g}_1)\hat{f}(\boldsymbol{\w}) := \cF\pi_1({\bf g}_1)\cF^{-1} \hat{f}(\boldsymbol{\w}) = e^{-i\boldsymbol{\w} \cdot {\bf g}_1} \hat{f}(\boldsymbol{\w}).\]
To derive $\hat{\pi}_2$ and $\hat{\pi}_3$, note that for a general invertible matrix $B\in \RR^{2\times 2}$,
\begin{align}
[\cF f(B{\bf x})](\boldsymbol{\w})=  & \iint_{\RR^2}e^{-2\pi i \boldsymbol{\w}\cdot{\bf x} }f(B{\bf x}) d{\bf x}\\
= & J(B)^{-1}\iint_{\RR^2}e^{-2\pi i \boldsymbol{\w}\cdot B^{-1}{\bf x} }f({\bf x}) d{\bf x}=J(B)^{-1}\hat{f}(B^{-T}\boldsymbol{\w}),
\label{eq:General_B_Fourier}
\end{align}
where $J(B)$ is the Jacobian of $B$, and $B^{-T}$ is the transpose of $B^{-1}$.
Therefore,
\[\hat{\pi}_2(g_2)\hat{f}(\boldsymbol{\w}):= \cF\pi_2(g_2)\cF^{-1} \hat{f}(\boldsymbol{\w}) = \hat{f}({R}_{g_2}\boldsymbol{\w})\]
\[\hat{\pi}_3(g_3)\hat{f}(\boldsymbol{\w}):= \cF\pi_3(g_3)\cF^{-1} \hat{f}(\boldsymbol{\w}) = e^{\frac{3}{4}g_3}\hat{f}({D}_{g_3}\boldsymbol{\w})\]
Let us define the canonical observables directly in the frequency domain. 
For translation, the natural definition is the position observable
\[\big( \bT^1_1 \hf(\boldsymbol{\w}),\bT^2_1 \hf(\boldsymbol{\w})\big) = \big(i\frac{\partial}{\partial \w_1}\hf(\boldsymbol{\w}),i\frac{\partial}{\partial \w_2}\hf(\boldsymbol{\w}) \big).\]
For rotations define the angle observable
\[\bT_2 \hf(\boldsymbol{\w})={\rm Arg}(\boldsymbol{\w}) \hf(\boldsymbol{\w}),\]
where ${\rm Arg}:\RR^2\rightarrow\CC$ is defined by ${\rm Arg}(\w_1,\w_2)=e^{i\theta}$ for $\theta$ satisfying $(\w_1+i\w_2)=\abs{\w_1+i\w_2}e^{i\theta}$.
For dilations, define the anisotropic scale observable by the arithemetic average of scales along the axis,
\[\bT_3 \hf(\boldsymbol{\w})=\big(-\frac{1}{2}\ln(\abs{\w_1})-\ln(\abs{\w_2}) \big)\hf(\boldsymbol{\w}),\]
which is equal to the geometric average 
\begin{equation}
\bT_3 \hf(\boldsymbol{\w})=\big(-\ln(\abs{\w_1}^{\frac{1}{2}}\abs{\w_2})\big)\hf(\boldsymbol{\w}).
\label{eq:geo_un_scale}
\end{equation}
As before, we define the uncertainty by 
\[S(f)= w_1\big(\s_f(\bT^1_1)+\s_f(\bT^2_1)\big) + w_2\s_f(\bT_2)+ w_3\s_f(\bT_3),\]
for some choice of weights $w_1,w_2,w_3>0$.

Last we introduce the angle transform and the anisotropic scale transform, as described in Theorem \ref{Theorem:pull_trans1}.
The angle transform $\Theta:L^2(\RR^2)\rightarrow L^2(e^{i\RR}\times\RR_+)$ is defined to be
\[[\Theta f](g_2,r) = \hf(r\Re(g_2),r\Im(g_2)).\]
Note that $\Theta$ intertwines rotations with translations along the $G_2$ axis.

For the anisotropic scale transform, 
note that anisotropic dilations keep the variable $q=\ln\Big(\frac{\abs{\w_1}^{1/2}}{\abs{\w_2}}\Big)$ constant, and translates the variable $g_2=-\ln(\abs{\w_1}^{\frac{1}{2}}\abs{\w_2})$. Inverting this gives
\[\w_1=\pm e^qe^{-g_3} \quad , \quad \w_2=\pm e^{-\frac{1}{2}q}e^{-\frac{1}{2}g_3}.\]
This leads to the following construction of the anisotropic scale transform.
Consider the four subspaces of $L^2(\RR^2)$ with frequency supports in each of the four quadrants of $\RR^2$
\[L_{n}^2(\RR^2) = \{\hf\in L_2(\RR) \ |\ {\rm support}(\hf)\subset \RR_{\pm}\times\RR_{\pm}\},\]
and $n=1,\ldots,4$ is some ordering of the quadrants.
Consider the four anisotropic warping transforms ${\cal W}_{n} :L_{n}^2(\RR^2)\rightarrow L^2(\RR^2)$, $n=1,\ldots,4$, defined by
\begin{equation}
[{\cal W}_{n}\hf](g_3,q)=e^{\frac{1}{4}q}e^{-\frac{3}{4}g_3}\hf(\pm e^q e^{-g_3},\pm e^{-\frac{1}{2}q} e^{-\frac{1}{2}g_3}).
\label{eq:Warping2}
\end{equation}
Let us define 
\[{\cal W}=\bigoplus_{n=1}^4{\cal W}_{n}:L^2(\RR^2)\rightarrow L^2(\RR^2)^4 \quad , \quad [{\cal W}\hf](g_3,q;n)= [{\cal W}_{n}\hf](g_3,q).\]
We define the anisotropic scale transform by ${\cal C}={\cal W}\cF$, and call $L^2(\RR^2)^4$ the anisotropic scale domain. Note that ${\cal C}$ intertwines anisotropic dilations with translations along the $G_3$ axis.

\section{The global localization framework}
\label{The global localization framework}
In this section we construct a framework for defining all of the canonical observables of a generalized wavelet transform ``at once''. {This framework is a special case of the one parameter localization framework, where $G$ is a group with a nested semi-direct product structure. It is thus also a special case of the classical general wavelet theory of square integrable representations. 
In this global framework, it is possible to define variances that are constant on orbits. This means that the corresponding uncertainty describes the localization behavior of the wavelet transform as a whole, instead of describing the individual localization of a window.}
For motivation, we start with the example of the 1D wavelet transform.  This transform is based on the affine group, which has a semi-direct product structure.

\subsection{Semi-direct products}

A group $G$ is called a \textbf{semi-direct product} of a normal subgroup $N\triangleleft G$ and a subgroup $H\subset G$, if $G=NH$ and $N\cap H=\{e\}$. This is denoted by $G=N\rtimes H$. If $G=N\rtimes H$, then each element $g\in G$ can be written in a unique way as $nh$ where $n\in N$, $h\in H$. Thus we can identify elements of $G$ with ordered pairs, or coordinates $(n,h)\in N\times H$. In the coordinate representation, the group multiplication takes the form
\[(n,h)(n',h')\sim nhn'h'=n\ hn'h^{-1}\ hh'\sim (n\ hn'h^{-1},hh').\]
Since $N$ is a normal subgroup, $A_h(n')=hn'h^{-1}$ is in $N$. Moreover, $A_h$ is a smooth group action of $H$ on $N$, and a smooth automorphism of $N$ for each $h\in H$.

When $N,H$ are isomorphic to physical quantities, $G=N\times H$ is interpreted as the group of ordered pairs of $quantity_1,quantity_2$. Each coordinate of $G$ corresponds to the physical dimension of the corresponding physical quantity, and the value at this coordinate corresponds to the value of the physical quantity.
Thus, the semi-direct product structure allows us to make the following philosophical argument apply to groups: ``physical quantities may change their values under the application of transformations, but they retain their dimensions.'' Namely, multiplying a group element $g$ of $G$ with another, may change the values of the coordinates of $g$, but may not change the ordered pair structure itself. Recall that this philosophical statement was employed only for observables up until now. This interpretation holds also in the case where $N$ and $H$ are direct products of physical quantities.

\begin{example}
In the case of the affine group, we have $G=N\rtimes H$, where $N\sim \{\RR,+\}$ is the subgroup of translations and $H\sim\{\RR,+\}\times\{-1,1\}$ is the subgroup of dilations and reflections. The group product takes the following form in coordinates
\[(n,h_1,h_2)\bullet(n',h'_1,h_2')=(n+h_2e^{h_1}n',h_1+h_1',h_2h_2').\]
Namely, $A_{(h_1,h_2)}(n')=h_2e^{h_1}n'$.
\end{example}

\subsection{The global localization framework for the wavelet transform}
\label{The global localization framework for the wavelet transform}

Let us now motivate the construction of the global localization framework for the case of the 1D wavelet transform. In this anlysis we ignore the less important phase direction observable, and consider the subgroup $time\  translations\rtimes dilations$ of the affine group, called the reduced affine group.
By (\ref{eq:VarTransaff2}), the time variance of a window $f$ is multiplied by $e^{2g_2}$ when the window is dilated by $g_2$. This agrees with the observation in Subsection \ref{1DWaveExample} that the smaller the mean
scale of a window is, the more simultaneous time-scale localization is possible when using the standard uncertainty $\s_f(\breve{T}_1)\s_f(\breve{T}_2)$. As is noted in Subsection \ref{1DWaveExample}, in generalized wavelet transforms we are interested in an uncertainty which is invariant under the action $\pi$ on the window $f$. In the 1D wavelet transform we may define the global variances as $\Sigma_f(\breve{T}_1)=e^{-2e_f(\breve{T}_2)}\s_f(\breve{T}_1)$ and $\Sigma_f(\breve{T}_2)=\s_f(\breve{T}_2)$. As required, the global variances are constant on orbits. Namely, $\Sigma_{\pi(g)f}(\breve{T}_1)=\Sigma_f(\breve{T}_1)$ for any $g\in G$, and similarly for $\Sigma_f(\breve{T}_2)$. We define the global uncertainty as
\begin{equation}
S(f)=w_1\Sigma_f(\breve{T}_1)+w_2\Sigma_f(\breve{T}_2)=w_1e^{-2e_f(\breve{T}_2)}\s_f(\breve{T}_1)\ +\ w_2\s_f(\breve{T}_2)
\label{eq:min_wav}
\end{equation}
for weights $w_1,w_2\in\RR_+$.
By finding a minimizer to the global uncertainty we avoid the misleading result discussed in Subsection \ref{1DWaveExample}.

Let us formulate this example in a way that allows generalization. We are interested in the transformations of $\sigma_f(\bT_1),\sigma_f(\bT_2)$ under the application of $\pi(g)$ on $f$. From the Heisenberg point of view , it is enough to know the transformations of $\bT_1,\bT_2$ under conjugation with $\pi(g)$. 
Indeed, by (\ref{Hpov1}) and (\ref{Hpov2}) we have for $m=1,2$
\begin{equation}
e_{\pi(g)f}(\bT_m) = e_f\big(\pi(g)^*\bT_m \pi(g)\big) \quad , \quad
\s_{\pi(g)f}(\bT_m) =\s_{f}\big(\pi(g)^*\bT_m \pi(g)\big).
\label{eq:H_POV_wavelet}
\end{equation}

Let us define the multi-observable ${\bf \bT}: \cH\rightarrow \cH^2$ by ${\bf \bT}f=(\bT_1f,\bT_2f)$, and define conjugation by $\pi(g)^*{\bf \bT}\pi(g)=(\pi(g)^*\bT_1\pi(g),\pi(g)^*\bT_2\pi(g))$. It is readily verified that 
\begin{equation}
\pi(g)^*{\bf \bT}\pi(g) = (g_1I+e^{g_2}\bT_1,g_2I+\bT_2).
\label{eq:wave_com}
\end{equation}
By the fact that $A_{g_2}(g_1')= e^{g_2}g_1'$, (\ref{eq:wave_com}) can be written in the form
\begin{align}
\label{global_canoni1}
\pi(g)^*{\bf \bT}\pi(g)=g\bullet {\bf \bT}
\end{align}
where the right hand side of (\ref{global_canoni1}) is given by functional calculus (Remark \ref{Band_limit_poly}) as 
\[g\bullet {\bf \bT} = \Big(\int \big(g_1+A_{g_2}(\lambda_1)\big)dP(\lambda_1) , \int \big(g_2+\lambda_2\big)dP(\lambda_2)\Big).\]
We interpret (\ref{global_canoni1}) as a canonical commutation relation, relating the multi-observable ${\bf \bT}$ with the representation $\pi$. We call (\ref{global_canoni1}) the multi-canonical commutation relation, and call ${\bf \bT}$ a canonical multi-observable.

We define expected values and variances of ${\bf \bT}$ by ${\bf e}_{f}({\bf \bT}) = \Big(e_f(\bT_1),e_f(\bT_2)\Big)$ and $\boldsymbol{\s}_{f}({\bf \bT}) = \Big(\s_f(\bT_1),\s_f(\bT_2)\Big)$.
Observe that the transformation rules (\ref{eq:expValaff2})-(\ref{eq:VarTransaff2}) can be derived from the Heisenberg point of view (\ref{eq:H_POV_wavelet}) and the commutation relation (\ref{global_canoni1}), and written as
\begin{equation}
{\bf e}_{\pi(g)f}({\bf \bT})=  g\bullet e_{f}({\bf \bT})=\Big(g_1+e^{g_2}e_f(\bT_1),g_2+e_f(\bT_2)\Big),
\label{eq:wave_exp_trans}
\end{equation}
\begin{equation}
 \boldsymbol{\s}_{\pi(g)f}({\bf \bT}) =\Big(e^{2g_2}\s_f(\bT_1),\s_f(\bT_2)\Big) = \Big(A_{g_2}^2\s_f(\bT_1),\s_f(\bT_2)\Big).
\label{eq:wave_var_trans}
\end{equation}
where $A_{g_2}^2=A_{g_2}\circ A_{g_2}$ is defined by composition.

Our goal is to design a ``global variance'' which is constant on orbits $\pi(G)f$. Given $f$ and its orbit $\pi(G)f$, there is always an element $y\in\pi(G)f$ with $e_y(\bT_1)=e_y(\bT_2)=0$. Namely, for $g$ given in coordinates by $(g_1,g_2)=\big(e_f(\bT_1),e_f(\bT_2)\big)$, $y=\pi(g^{-1})f$. Indeed, by (\ref{eq:wave_exp_trans}),
\begin{equation}
{\bf e}_{y}({\bf \bT})={\bf e}_{\pi(g^{-1})f}({\bf \bT})=g^{-1}\bullet {\bf e}_{f}({\bf \bT})=g^{-1}\bullet g=(0,0).
\label{eq:g7kasdasdsd4}
\end{equation}
Now, by (\ref{eq:wave_var_trans}) and (\ref{eq:g7kasdasdsd4}),
\[\Sigma_f(\breve{T}_1)=e^{-2e_f(\breve{T}_2)}\s_f(\breve{T}_1)=e^{-2g_2}\s_f(\breve{T}_1)=\s_{\pi(g^{-1})f}(\bT_1)=\s_y(\bT_1),\]
\[\Sigma_f(\breve{T}_2)=\s_f(\breve{T}_2)=\s_{\pi(g^{-1})f}(\bT_2)=\s_{y}(\bT_2).\]
This shows that $\Sigma_f(\breve{T}_1),\Sigma_f(\breve{T}_1)$ are the variances of the unique element $y\in\pi(G)f$ having zero expected values, and thus $S(f)$, as defined in (\ref{eq:min_wav}), is constant on orbits.

\subsection{Semi-direct product wavelet transforms}
\label{The general global localization framework}

We are now ready to introduce the general setting of the global localization framework. The following assumption strengthen Assumption \ref{ass_voice}.

\begin{assumption}[Semi-direct product wavelet transform]
\label{ass_voice2}
A Semi-direct product wavelet transform is constructed by, and assumed to satisfy, the following.
\begin{enumerate}
\item
The group $G$ is a nested semi-direct product group, namely
\begin{align}
\label{eq:general_semi2}
&G = H_0\\
&H_0=  (N_0\times N_1)\rtimes H_1 \\
&H_m  =  N_{m+1} \rtimes H_{m+1} \quad , \quad m=1,\ldots, M-2 \\
&H_{M-1}=N_M .
\end{align}
Here, $N_0$ is the center of $G$.  
For $m=0,\ldots,M$,
 $N_m$ is a group direct product of physical quantities, $G_m^1\times \ldots\times G_m^{K_m}$, where $K_m\in\NN$. 
We denote elements of $N_m$ in coordinates by ${\bf g}_m=(g_m^1,\ldots,g_m^{K_m})$, and elements of $H_m$ by ${\bf h}_m$.
Note that ${\bf h}_m=({\bf g}_{m+1},\ldots,{\bf g}_M)$. For the center, we also denote $Z=N_0$, and $K_z=K_0$, and denot elements of $Z$ in coordinates by ${\bf z}=(z^1,\ldots,z^{K_z})$. 
  \item
  We consider the representations $\pi_m({\bf g}_m)=\pi_m^1(g_m^1) \circ\ldots \circ \pi_m^{K_m}(g_m^{K_m})$ of $N_m$, $m=0,\ldots,M$ in $\cH$. We assume  that
	$\pi(g)=\pi_0({\bf g}_0)\circ\ldots \circ \pi_M({\bf g}_M)$ is a square integrable representation of $G$. Namely, $\pi$ is a SCU irreducible representation, and there is a vector $f\in\cH$ such that $V_f[f]\in L^2(G)$. Here, $V_f[f]$ is defined in (\ref{eq:1}). %
	\item
	The semi-direct product wavelet transform based on $\pi$ and on the window $f\in\cH$, satisfying $V_f[f]\in L^2(G)$, is defined to be $V_f:\cH\rightarrow L^2(G)$, as defined in (\ref{eq:1}).
\end{enumerate}
\end{assumption}

We abbreviate semi-direct product wavelet transforms by SPWT.	
	By the theory of square integrable representations of locally compact topological groups, we have the following theorem (for example see \cite{gmp}).

\begin{proposition}
Any SPWT also satisfies Assumption \ref{ass_voice}.
\end{proposition}

	Note that this assumption includes Schr\"odinger representations of Heisenberg groups based on tuples of physical quantities, like the STFT and the FSTFT. Indeed, Heisenberg groups can be written as 
	$J=(phase\ rotations \times translations)\rtimes modulations$. The 1D wavelet transform and the Shearlet transform are also SPWTs.

	\begin{remark}
	\label{ass_inversion2}
	Consider a SPWT.
	\begin{enumerate}
		\item
	As a result of the semi-direct product structure, the group multiplication has the following form in coordinates
	\begin{equation}
	\begin{split}
	g\bullet g' = & ({\bf z},{\bf g}_1,\ldots,{\bf g}_M)\bullet ({\bf z}',{\bf g}_1',\ldots,{\bf g}_M') \\
	 =    & \Big(({\bf z},{\bf g}_1) \bullet A_{z,1}\big({\bf h}_1;({\bf z}',{\bf g}_1')\big) \ ,\  {\bf g}_{2} \bullet A_{2}({\bf h}_{2};{\bf g}_{2}')    \ ,\ \ldots \ , \ {\bf g}_{M-1} \bullet A_{M-1}({\bf h}_{M-1};{\bf g}_{M-1}')\ ,\ {\bf g}_M \bullet {\bf g}_M'\Big)
	\end{split}
	\label{eq:vvvvvvvvvvv}
	\end{equation}
	where $A_m$ are smooth automorphisms with respect to ${\bf g}_m'$ if $m\geq 2$, and with respect to $({\bf z}',{\bf g}_1')$ for $m=(z,1)$. Moreover, $A_m$ with respect to ${\bf h}_m$,  are smooth group actions of $H_m$ on $N_m$  for $m\geq 2$, and on $Z\times N_1$ for $m=(z,1)$. Here, ${\bf h}_m$ ($m=1,\ldots,M$) are coordinates corresponding to $g$.	
		\item
		We can write a formula for the group inverse of $g\in G$ in coordinates. For $m=2,\ldots,M$, let $({\bf g}_m',{\bf h}_m^{-1})$ be the inverse of $({\bf g}_m,{\bf h}_m)$. We use
\[({\bf g}_m,{\bf h}_m)\bullet ({\bf g}_m',{\bf h}_m^{-1})= ({\bf g}_m\bullet A_m({\bf h}_m;{\bf g}_m'), {\bf h}_m\bullet{\bf h}_m^{-1})=(\bf{e},\bf{e})\]
to get ${\bf g}_m'= A_m({\bf h}_m^{-1};{\bf g}_m^{-1})$. The inverse of $({\bf z},{\bf g}_1,{\bf h}_1)$ is given by $\left(A_{z,1}\big({\bf h}_1^{-1};({\bf z}^{-1},{\bf g}_1^{-1})\big),{\bf h}_1^{-1}\right)$.
\end{enumerate}	 
\end{remark}

Next we explain how the center $Z$ of $G$ may be omitted in a SPWT.
Assume that $Z= G^{1}_z\times\ldots\times G^{K_z}_z$, where $G^k_z$ is a physical quantity. Denote $\pi_{z,1}({\bf z},{\bf g}_1)=\pi_z({\bf z})\circ\pi_1({\bf g}_1)$, where $\pi_z({\bf z})=\pi_0({\bf g}_0)$ is the representation of the center. 
		A character of a group is a unitary representation of the group in $\CC$.
		Any irreducible representation $\pi$ of $G$, restricted to the center of the group $Z$, is a character times the identity operator.  
		Therefore, $\pi_z({\bf z})=\chi(z)I$ for some character $\chi$ of $Z$, and $I$ the identity operator in $\cH$.
		As a result, wavelet transform $V_f[s]$ of any $s\in\cH$ is completely determined by the values of $V_f[s]$ on the cross section
		\begin{equation}
		G_z=\{g\in G\ |\ {\bf z} = {\bf e}\} \cong G/Z
		\label{eq:center_cross}
		\end{equation}
		where $\bf{e}$ is the unit element of $Z$ in coordinates, and $G/Z$ is the quotient group of $G$ relative to $Z$.
		Indeed, for any $g\in G$, we have in coordinates $\pi({\bf z},{\bf g}_1,{\bf h}_1) = \chi({\bf z})\pi({\bf e},{\bf g}_1,{\bf h}_1)$, so 
		\begin{equation}
		V_f[s]({\bf z},{\bf g}_1,{\bf h}_1)= \overline{\chi({\bf z})}V_f[s]({\bf e},{\bf g}_1,{\bf h}_1).
		\label{eq:no_Z_localization}
		\end{equation}
		Thus in a SPWT, restricting $V_f[s]$ to the domain $G_z$, preserves the invertibility of the SPWT. For this reason, in SPWT like the STFT, $V_f[s](g)$ is calculated only for $g\in G_z$.
		
We can now show that $G_z$ is a nested semi-direct product group with trivial center. For this, 
let us analyze the automorphism ${A}_{z,1}({\bf h}_1;\cdot):Z\times N_1\rightarrow Z\times N_1$. Let $z'\in G$ be an element in the subgroup $Z$, and $g_1'\in G$ be an element in the subgroup $N_1$. Any generic element in the subgroup $Z\times N_1$ can be written uniquely as $g_{z,1}'=z'g_1'$.
	Let $g_{z,1}$ be another element in $Z\times N_1$, and $h_1,h_1'$ elements in $H_1$. Any two generic elements in $G$ can be written as $g_{z,1}'h_1'=z'g_1'h_1'$ and $g_{z,1}h_1$. By the semi-direct product structure we have 
	$g_{z,1}h_1g_{z,1}'h_1'= g_{z,1} \ h_1g_{z,1}'h_1^{-1} \ h_1 h_1'$, where $h_1g_{z,1}'h_1^{-1}=A_{z,1}(h_1;g_{z,1}')\in Z\times N_1$.
	By the fact that the center commutes with every element,
	\begin{equation}
	A_{z,1}(h_1;g_{z,1}')=h_1g_{z,1}'h_1^{-1}=z'h_1g_1'h_1^{-1}=z'\ A_{z,1}(h_1;g_1').
	\label{eq:A_0_from_A_1}
	\end{equation}
	We denote the projection of $A_{z,1}(h_1;g_1')$ to the subgroup $N_1\subset Z\times N_1$ by $A_1(h_1;g_1')$, and the projection to $Z$ by $A_z(h_1;g_1')$.
	 By the direct product structure, the projections  $Z\times N_1\rightarrow N_1$ and $Z\times N_1\rightarrow Z$ are homomorphisms. Thus,  $A_1(h_1;g_1)$ and $A_z(h_1;g_1)$ are smooth homomorphism $N_1\rightarrow N_1$ and $N_1\rightarrow Z$ respectively. Moreover,
		\begin{equation}
	A_{z,1}(h_1;g_{z,1}')=z'A_{z}(h_1;g_1')\ A_{1}(h_1;g_1').
	\label{eq:A_0_from_A_100}
	\end{equation}
	By the fact that $A_{z,1}(h_1;g_{z,1}')$ is invertible with respect to $g_{z,1}'$, and by the direct product structure, (\ref{eq:A_0_from_A_100}) shows that $A_1(h_1;g_1')$ is also invertible, and thus an automorphism with respect to $g_1'$.
	Moreover,  $A_1(h_1;g_1')$ is a group action of $H_1$ on $N_1$ with respect to $h_1$. Indeed
\begin{equation}
\begin{split}
  A_{z}(h_1'h_1;g_{1}')\ A_{1}(h_1'h_1;g_{1}') = & \ A_{z,1}(h_1'h_1;g_{1}')
	=h_1'h_1g_1'h_1^{-1}h_1'^{-1}\\
	= & h_1'A_{z,1}(h_1;g_1')h_1'^{-1}
	=A_{z}(h_1;g_1')\  h_1'A_{1}(h_1;g_1')h_1'^{-1}\\
	=& A_{z}(h_1;g_1')A_{z}\Big(h_1';A_{1}(h_1;g_1')\Big) \ \ A_{1}\Big(h_1';A_{1}(h_1;g_1')\Big),
	\end{split}
	\label{eq:A_0_from_A_12}
	\end{equation}
	so $A_{1}(h_1'h_1;g_{1}')=A_{1}\Big(h_1';A_{1}(h_1;g_1')\Big)$. To conclude, the group $G_z$ is the nested semi-direct product group, 
\begin{align}
\label{eq:general_semi2G_z}
&G = H_0\\
&H_m  =  N_{m+1} \rtimes H_{m+1} \quad , \quad m=0,\ldots, M-2 \\
&H_{M-1}=N_M,
\end{align}
with $A_m$ mappings equal to those of $G$.

			Observe that the representation operators $\pi_z^k(z_k)=\chi(z_k)|_{G_0^k}I$ commute with every operator, so they do not have a canonical observable as defined in (\ref{eq:5a}). Moreover, by (\ref{eq:no_Z_localization}), no localization in the $Z$ direction is possible. This motivates us to develop our localization theory in $G_z$ instead of $G$. Working in $G_z$ instead of in $G$ also makes sense since observables interact with $\pi(g)$ via conjugations, which are elements of the inner automorphism group of $G$, isomorphic to $G_z$. Note that $\pi$ may be restricted to the cross-section $G_z$, but it is no longer a representation of $G_z$ in general.

Last, we present a convenient way to represent the automorphisms $A_m({\bf h}_m; \ \cdot\ )$ in a SPWT. Let us denote in short $ A=A_m({\bf h}_m; \ \cdot\ )$ the automorphism $N_m\rightarrow N_m$ for some $m=1,\ldots,M$. 
	Since $N_m$ is a direct product of the physical quantities $G_m^1\times \ldots\times G_m^{K_m}$, $A$ can be written as an invertible matrix ${\bf A}$ as explained next.
	For each $k=1,\ldots,K_m$, consider the projections $P_k:N_m\rightarrow G_m^k$, defined in coordinates by
	\[{\bf P}_k(g_m^1,\ldots,g_m^{K_m})=(e,\ldots,e,g_m^k,e,\dots,e)\]
	where $e$ are the unit elements of each $G_m^{k'}$, $k'=1,\ldots,K_m$. Since $N_m$ is a direct product, $P_k$ are homomorphisms. For each pair $k,k'=1,\ldots,K_m$ consider the homomorphism $a_{k,k'}=P_{k}AP_{k'}:G_m^{k'}\rightarrow G_m^k$. The automorphism $A$ can now be written in coordinates as the homomorphism valued invertible matrix ${\bf A}$ with entries $a_{k,k'}$. 
	The multiplication of ${\bf A}$ by ${\bf g}_m$ is defined by
	\[\forall k=1,\ldots,K_m\ , \quad [{\bf A}{\bf g}_m]_k = a_{k,1}(g_m^1)\bullet \ldots \bullet a_{k,K_m}(g_m^{K_m}).\]
	Here, invertibility means that there is another homomorphism valued matrix ${\bf A}^{-1}$ with ${\bf A}{\bf A}^{-1}={\bf A}^{-1}{\bf A}={\bf I}$. We denote these matrices in the extended notation by ${\bf A}_m({\bf h}_m)$.
By (\ref{eq:A_0_from_A_100}), the matrix ${\bf A}_{z,1}({\bf h}_1)$ has the block form
	\begin{equation}
	{\bf A}_{z,1}({\bf h}_1)=
	\begin{blockarray}{ccc}
Z & N_1 & \\
\begin{block}{(cc)c}
  {\bf I} & {\bf A}_z({\bf h}_1) & Z \\
  {\bf 0} & {\bf A}_1({\bf h}_1) & N_1 \\ 
\end{block}
\end{blockarray}
	\label{eq:matrix_of_N0}
	\end{equation}
	where ${\bf 0}$ is a matrix with all entries equal to the trivial homomorphism mapping to the unit element, ${\bf I}$ is the matrix of identity automorhisms on the diagonal, and ${\bf A}_z({\bf h}_1),{\bf A}_1({\bf h}_1)$ are the matrix representation of $A_z({\bf h}_1,\cdot), A_1({\bf h}_1,\cdot)$. Here, ${\bf A}_1({\bf h}_1)$ is invertible, and ${\bf A}_z({\bf h}_1)$ is not in general.
	
		\begin{remark}
	\label{remark:MatrixAuto}
		In 1 of Assumption \ref{ass_voice2} we usually assume that each $N_m$ is a group direct product $G_m^1\times \ldots\times G_m^{K_m}$, where for each $k=1,\ldots,K_m$, $G_m^k=G_m$ is a physical quantity of the same type. Here, the entries of the matrix ${\bf A}$ are homomorphisms $G_m\rightarrow G_m$. 	
	For example, if $G_m=\{\RR,+\}$, any automorphism is an invertible matrix in $\RR^{K_m \times K_m}$.
	If $G_m=\{\ZZ,+\}$, any automorphism is an invertible matrix in $\ZZ^{K_m \times K_m}$, with inverse in $\ZZ^{K_m \times K_m}$. 
	Another example is the case where $N_m=G_m$ consists of one coordinate.  In this case, for $N_m=G_m=\{\RR,+\}$, $A$ is a multiplication by a nonzero scalar. For $N_m=G_m=\{\ZZ,+\}$,  $A$ is a multiplication by $\pm 1$. For $N_m=G_m=\{e^{i\RR},+\}$, $A$ is a multiplication of the exponent by $\pm 1$. Last, for $N_m=G=\{e^{2\pi i \ZZ/N}\}$, $A$ is a multiplication of the exponent by a co-prime of $N$.
	\end{remark}

	In the rest of this paper we assume that for each $k=1,\ldots,K_m$, $G_m^k=G_m$ is a physical quantity of the same type, as in Remark \ref{remark:MatrixAuto}. This assumption is taken to simplify formulations, though it is not always necessary.	

\subsection{Canonical multi-observables and localization}
	
In this subsection we define concepts of localization corresponding to the structure of semi-direct product wavelet transforms.
Particularly, we define the general multi-canonical observable, extending the multi-canonical observable of the 1D wavelet transform.
{ Canonical observables are operators that define the physical quantities structure of $G$ in the Hilbert space of signals $\cH$. As such, their structure is intimately related to the structure of the nested semi-direct product group $G$. The coordinates of $G$ represent different physical quantities, and when two elements $g,g'$ of $G$ are multiplied, each entry in the tuple $g$ interacts not only with the corresponding entry of $g'$, but also with entries of other physical quantities. To accommodate this property in the canonical observables, it is necessary to define them collectively as a tuple of observables, measuring ``at once'' all of the physical quantities. In \cite{Teschke2005} it was proposed to consider tuples of operators in order to define uncertainty principles comprising more than two operators, with an extension of 1D variances to multi dimensional covariances. We consider a similar setting. However, the motivation and application of our theory is different, namely to induce the group structure on the tuple of observables via a canonical commutation relation, rather than to define a multidimensional uncertainty principle for general operators.
Moreover, in the application of the theory in \cite{Teschke2005}, the operators used for defining the uncertainties are generators of $\pi(g)$, which we have shown to be inappropriate.}

{For our purposes, some commutativity assumptions on the observables are needed to guarantee self-adjointness in the $\RR$ or $\ZZ$ case, and thus to allow the use of a multidimensional spectral theorem.}
\begin{definition}
Consider a SPWT. We call the sequence of observables
	$\{\breve{T}^k_m\}_{\tiny \hspace{-2mm}\vspace{1mm}
	\begin{array}{l}
		m=1,\ldots,M \cr
		\ k=1,\ldots, K_m
	\end{array}}
	$
	a \textbf{multi-observable}, if for every $m=1,\ldots, M$, the observables $\bT_m^{1},\ldots,\bT_m^{K_m}$ commute. We denote  ${\bf \bT}_m=(\bT_m^{1},\ldots,\bT_m^{K_m}):\cH\rightarrow \cH^{K_m}$, and  denote the multi-observable by
  ${\bf \bT}=({\bf \bT}_1,\ldots,{\bf \bT}_M ):\cH\rightarrow\cH^{N}$ for $N=\sum_{m=1}^{M} K_m$. 
\end{definition}	
	Let ${\bf \bT}$ be a multi-observable. By the commutativity of $\bT_m^{1},\ldots,\bT_m^{K_m}$ for any $m=1,\ldots, M$, the spectral families of projections $P_m^1,\ldots,P_m^{K_m}$ of $\bT_m^{1},\ldots,\bT_m^{K_m}$ also commute.
	Therefore, 
	each ${\bf \bT}_m$ has spectral decomposition
	\begin{equation}
	{\bf \bT}_m = \int \boldsymbol{\lambda}_m d{\bf P}_m(\boldsymbol{\lambda}_m)=\int\ldots\int (\lambda_m^1,\ldots, \lambda_m^{K_m}) dP_m^1(\lambda_m^{1}) \ldots  dP_m^{K_m}(\lambda_m^{K_m}).
	\label{eq:87655hfh}
	\end{equation}
	Here, ${\bf P}_m$ is the PVM, mapping Borel sets of $N_m$ to projections in $\cH$, such that for any sequence of Borel sets $B_m^1,\ldots,B_m^{K_m}$ of $G_m^1,\ldots,G_m^{K_m}$ respectively, 
	\[{\bf P}_m\big( B_m^1\times\ldots\times B_m^{K_m}\big)=P_m^1(B_m^1)\circ\ldots\circ P_m^{K_m}(B_m^{K_m}) .\]
	Smooth functions $F:N_m\rightarrow N_m$ act of ${\bf \bT}_m$ by
	\begin{equation}
	F({\bf \bT}_m) = \int F(\boldsymbol{\lambda}_m) d{\bf P}_m(\boldsymbol{ \lambda}_m).
	\label{eq:nnnd54kh}
	\end{equation}
  By definition, representation operators $\pi(G)$ act on multi-observables  ${\bf \bT}:\cH\rightarrow\cH^N$ by conjugation according to the formula
	\[\begin{array}{l}
		\pi(g)^*{\bf \bT}\pi(g)=\pi(g)^*(\bT_1,\ldots,\bT_N )\pi(g)=\\
		\big(\pi(g)^*\bT_1\pi(g),\ldots,\pi(g)^*\bT_N\pi(g) \big).
	\end{array}\]			
Next we define the canonical multi-observable of the whole group $G$.

	\begin{definition}
  Consider a SPWT, and a multi-observable $\bf\bT$. We call ${\bf \bT}$ a \textbf{canonical multi-observable} of $\pi$,
	if ${\bf\bT}$ and $\pi$ satisfy the \textbf{multi-canonical commutation relation}
	\begin{equation}
\forall g\in G_z\ , \quad \pi(g)^*{\bf \bT}\pi(g)= g\bullet{\bf \bT} 
\label{global_observe}
\end{equation}
where 
\[g\bullet{\bf\bT}=
\left({\bf g}_1\bullet A_1({\bf h}_1;{\bf \bT}_1)\ ,\ \ldots\ ,\ {\bf g}_{M-1}\bullet A_{M-1}({\bf h}_{M-1};{\bf \bT}_{M-1})\ ,\ {\bf g}_M\bullet {\bf \bT}_M\right),\]
and $\bullet$ is the group product in $G_z$.
\end{definition}

Note that if ${\bf \bT}$ is a canonical multi-observable, then any one of its entries $\bT_m^k$ is a canonical observable of the corresponding $\pi_m^k$. This shows that  canonical multi-observable is a stronger definition than a sequence of canonical observables.

Let us now define notions of localization.
Since each $N_m$ consists of a number of physical quantities, the natural generalization of variances are covariances.
  Consider a SPWT, and a canonical multi-observable ${\bf\bT}$ of $\pi$.
	For each $m=1,\ldots,M$ and $f\in\cH$, we define the \textbf{multi-expected value} $\boldsymbol{e}_f({\bf\bT}_m)$ as the vector with entries
	\begin{equation}
	\Big[\boldsymbol{e}_f({\bf\bT}_m)\Big]_{k}=e_f(\bT_m^k).
	\label{eq:multi-expected-value}
	\end{equation}
	We define the \textbf{multi-covariance} $\boldsymbol{\s}_f({\bf\bT}_m)$ as the  matrix in $\CC^{K_m\times K_m}$ with entries
	\begin{equation}
	\Big[\boldsymbol{\s}_f({\bf\bT}_m)\Big]_{k,k'}=\ip{\Big(\bT_m^k-e_f(\bT_m^k)\Big) f}{\Big(\bT_m^{k'}-e_f(\bT_m^{k'})\Big) f}.
	\label{eq:multi-covariance}
	\end{equation}
  Next we define one dimensional variances along directions in $N_m$. Let ${\bf w}_m$ be a column vector in $\CC^{K_m}$, interpreted as a direction in $N_m$. Define the \textbf{directional variance}
	\begin{equation}
	\s^{{\bf w}_m}_f({\bf\bT}_m) = {\bf w}_m^*\boldsymbol{\s}_f({\bf\bT}_m){\bf w}_m.
	\label{eq:directional_variance}
	\end{equation}
Note that 
the multi-covariance matrix $\boldsymbol{\s}_f({\bf\bT}_m)$ is self-adjoit. Moreover, the variance $\s^{{\bf w}_m}_f({\bf\bT}_m)$ is non-negative. Indeed
\begin{equation}
0\leq \norm{ \sum_{m=1}^{K_m} w_m^k\Big(\bT_m^k-e_f(\bT_m^k)\Big) f }^2={\bf w}_m^*\boldsymbol{\s}_f({\bf\bT}_m){\bf w}_m= \s^{{\bf w}_m}_f({\bf\bT}_m),
\label{eq:directional_variance_meaning}
\end{equation}
which shows that $\boldsymbol{\s}_f({\bf\bT}_m)$ is positive semidefinite. Equation (\ref{eq:directional_variance_meaning}) gives an interpretation to the directional variance, as the variance of the normal observable ${\bf\bT}^{{\bf w}_m}=\sum_{m=1}^{K_m} w_m^k\bT_m^k$. Namely
\begin{equation}
\s^{{\bf w}_m}_f({\bf\bT}_m)=\s_f({\bf\bT}^{{\bf w}_m}).
\label{eq:direct_var_linear_comb}
\end{equation}

We may now define a scalar variance as a combination of directional variances 
\begin{equation}
\sum_{d=1}^{D_m}\s^{{\bf w}^m_d}_f({\bf\bT}_m)
\label{eq:finite_directions_variance}
\end{equation}
where ${\bf w}^m_1,\ldots,{\bf w}^m_{D_m}$ are directions. Note that $\s^{{\bf w}^m_d}_f({\bf\bT}_m)$ can be written as the Frobenius scalar product of the rank one self-adjoint positive semi-definite matrix ${\bf w}_d^m{\bf w}_d^{m *}$ with $\boldsymbol{\s}_f({\bf\bT}_m)$. Since rank one self-adjoint positive semi-definite matrices span the space of self-adjoint  positive semi-definite matrices, we define \textbf{scalar variances} using a simpler formulation of  (\ref{eq:finite_directions_variance}) as
\begin{equation}
\s^{{\bf W}^m}_f({\bf\bT}_m) = \ip{{\bf W}^m}{\boldsymbol{\s}_f({\bf\bT}_m)}_{\rm F}
\label{eq:infinite_directions_variance}
\end{equation}
where ${\bf W}^m$ is a self-adjoint positive semidefinite matrix, and the inner product in (\ref{eq:infinite_directions_variance}) is the Frobenius inner product. 
We call ${\bf W}^m$ the weight matrix corresponding to $N_m$.
For example, the choice of ${\bf W}^m=\bf{I}$ amounts to summing the variances along the axis of $N_m$. The choice of ${\bf W}^m$ as a matrix with all entries equal to the same positive constant, corresponds to an isotropic scalar variance.

Last, we study how multi expected values and variances of multi-observables are transformed by $\pi$. We divide the analysis to two cases. The self-adjoint case, where $G_m$ is $\RR$ or $\ZZ$, and the unitary case, where $G_m$ is $e^{i\RR}$ or $e^{2\pi i \ZZ/N}$.

\begin{proposition}
\label{Prop:scalar_var_translation1}
Consider a SPWT, and let $m$ be an index such that $G_m$ is $\RR$ or $\ZZ$. Let ${\bf \bT}$ be a canonical multi-observable. Then for any $g\in G_z$
\begin{equation}
{\bf e}_{\pi(g)f}({\bf\bT}_m) = {\bf g}_m\bullet{\bf A}_m({{\bf h}_m}){\bf e}_f({\bf \bT}_m),
\label{eq:multi_e_transform}
\end{equation}
\begin{equation}
\s^{{\bf W}_m}_{\pi(g)f}({\bf \bT}_m) = \s^{{\bf A}_m({{\bf h}_m}){\bf W}_m{\bf A}_m({{\bf h}_m})^*}_{f}({\bf \bT}_m),
\label{scalar_var_translation1}
\end{equation}
for any weight matrix ${\bf W}_m$, where $\bullet$ is the arithmetic sum of numbers.
\end{proposition}

\begin{proof}
By Remark \ref{remark:MatrixAuto}, $A_m({{\bf h}_m};\ \cdot\ )$ is the invertible $G_m$ valued matrix operator
\[{\bf A}_m({\bf h}_m)=\left(
\begin{array}{ccc}
	a^{1,1}_m({{\bf h}_m}) & \ldots & a^{1,K_m}_m({{\bf h}_m} )\cr
	\vdots  &    & \vdots \cr
	a^{K_m,1}_m({{\bf h}_m}) & \ldots & a^{K_m,K_m}_m({{\bf h}_m})
\end{array}
\right).\]
By the multi-canonical commutation relation (\ref{global_observe}), we have
\begin{equation}
\pi(g)^*\bT^k_m \pi(g)=g_m^k I+\sum_{l=1}^{K_m} a^{k,l}_m({{\bf h}_m})\bT^l_m.
\label{eq:nni6s}
\end{equation}
So
\[e_{\pi(g)f}(\bT^k_m) = g_m^k+\sum_{l=1}^{K_m} a^{k,l}_m({{\bf h}_m})e_f(\bT^l_m) \]
which gives 
(\ref{eq:multi_e_transform}).

For any directional variance, with direction ${\bf w}$, we have by the Heisenberg point of view (\ref{Hpov2}), by (\ref{eq:nni6s}) and by (\ref{eq:directional_variance_meaning})
\begin{align}
\label{directional_var_orbit1}
\s^{{\bf w}}_{\pi(g)f}({\bf \bT}_m) & =  \s^{{\bf w}}_{f}(\pi(g)^*{\bf \bT}_m\pi(g)) \\
&= \norm{\sum_{k=1}^{K_m}w_k\left(\sum_{l=1}^{K_m} a^{k,l}_m({{\bf h}_m})\bT^l_m\ -\ \sum_{l=1}^{K_m} a^{k,l}_m({{\bf h}_m})e_f(\bT^l_m)\right)f}^2 \\
&= \norm{\sum_{l=1}^{K_m} \Big(\sum_{k=1}^{K_m}a^{k,l}_m({{\bf h}_m})w_k \Big) \left(\bT^l_m\ -\ e_f(\bT^l_m)\right)f}^2 \\
&= \s^{{\bf A}_m({{\bf h}_m}){\bf w}}_{f}({\bf \bT}_m).
\end{align}
As a result, by expanding any self-adjoint positive semidefinite matrix ${\bf W}_m$ using the rank-one self-adjoint positive semidefinite matrices based on the eigenvectors of ${\bf W}_m$, we obtain (\ref{scalar_var_translation1}).

\end{proof}

For the unitary case, we present a restricted result.
\begin{proposition}
\label{Prop:scalar_var_translation2}
Consider a SPWT, and let $m$ be an index such that $G_m$ is $e^{i\RR}$ or $e^{2\pi i \ZZ/N}$. If ${\bf A}_m({\bf h}_m)={\bf I}$, then for any $g\in G_z$
\begin{equation}
{\bf e}_{\pi(g)f}({\bf\bT}_m) = {\bf g}_m \bullet {\bf e}_f(\bf \bT_m).
\label{eq:multi_e_transform22}
\end{equation}
\begin{equation}
\s^{{\bf W}_m}_{\pi(g)f}({\bf \bT}_m) = \s^{{\bf W}_m}_{f}({\bf \bT}_m),
\label{scalar_var_translation122}
\end{equation}
for any weight matrix ${\bf W}_m$, where $\bullet$ is the arithmetic product of numbers, and operates element-wise.
\end{proposition}

\begin{proof}
By the assumption that ${\bf A}_m({\bf h}_m)={\bf I}$, the restriction of the canonical commutation relation (\ref{global_observe}) to $N_m$ reads
\begin{equation}
\pi({\bf g})^*{\bf \bT}_m\pi({\bf g})={\bf g}_m\bullet {\bf \bT}_m.
\label{eq:transform_mean_proj2340000}
\end{equation}
Therefore,
\begin{equation}
{\bf e}_{\pi({\bf g})f}({\bf \bT}_m)= \ip{\pi({\bf g})^*{\bf \bT}_m\pi({\bf g})f}{f} = \ip{{\bf g}_m\bullet {\bf \bT}_m f}{f} = {\bf g}_m\bullet {\bf e}_{f}({\bf \bT}_m).
\label{eq:transform_mean_proj2343}
\end{equation}
where ${\bf g}_m\bullet (\cdot)$ commutes with the inner product since it is a multiplication by a scalar.
As a result of (\ref{eq:directional_variance_meaning}), (\ref{eq:transform_mean_proj2343}), and (\ref{eq:transform_mean_proj2340000}), for any directional variance with direction ${\bf w}$, we have
\begin{align}
\label{directional_var_orbit2267}
\s^{{\bf w}}_{\pi(g)f}({\bf \bT}_m) &= \s^{{\bf w}}_{f}\big(\pi(g)^*{\bf \bT}_m\pi(g)\big)  \\
&=\norm{\sum_{k=1}^{K_m}w_k\left({\bf g}_m\bullet\bT^k_m\ -\ {\bf g}_m\bullet e_f(\bT^k_m)\right)f}^2 \\ 
& = \norm{\sum_{k=1}^{K_m}w_k\left(\bT^k_m\ -\ e_f(\bT^k_m)\right)f}^2 
= \s^{{\bf w}}_{f}({\bf \bT}_m).
\end{align}
Similarly to the proof of Proposition \ref{Prop:scalar_var_translation1}, (\ref{directional_var_orbit2267}) extends to weight matrices ${\bf W}_m$.
\end{proof}

\subsection{Solving the multi-canonical commutation relation}
\label{Solving the global canonical commutation relation}

To solve (\ref{global_observe}), we develop a theory analogous to Subsection \ref{Solving the canonical commutation relation}.
First we define an extension of canonical systems for $\{G,\pi,{\bf \bT}\}$, where $\{G,\pi\}$ satisfy Assumption \ref{ass_voice2}, and ${\bf \bT}$ is a canonical multi-observable.
We denote by ${\bf T}_m = (T_m^1,\ldots,T_m^{K_m})$ the generators of $\pi^1_m(g_m^1),\ldots,\pi^{K_m}_m(g_m^{K_m})$ respectively. Note that by the direct product structure of $N_m$, the operators $\pi^1_m(g_m^1),\ldots,\pi^{K_m}_m(g_m^{K_m})$ commute, and thus $T_m^1,\ldots,T_m^{K_m}$ commute. We denote ${\bf T}=({ \bf T}_1,\ldots,{ \bf T}_M)$.
For self-adjoint ${\bf \bT}_m$ we define for $\hat{\bf g}\in \hat{G}$
\begin{equation}
\begin{array}{ll}
{\rm if\ } G_m=\RR {\rm \ :\ \ } &{\breve{\pi}}_m(\hat{\bf g}_m) = e^{i\hat{\bf g}_m\cdot {\bf \bT}_m} \\
{\rm if\ } G_m=\ZZ {\rm \ :\ \ } &{\breve{\pi}}_m(\hat{\bf g}_m) = \hat{\bf g}_m^{{\bf \bT}_m}:=e^{\ln(\hat{\bf g}_m)\cdot {\bf \bT}_m}
\end{array}
\label{eq:generator_breve1}
\end{equation}
where $ {\bf q}_m\cdot {\bf \bT}_m=\sum_{k=1}^{K_m}{ q}_m^k\bT_m^k$.  For unitary ${\bf \bT}_m$ we define for $\hat{\bf g}\in \hat{G}$
\begin{equation}
\begin{array}{ll}
	{\rm if\ } G_m=e^{i\RR} {\rm \ :\ \ }&{\breve{\pi}}_m(\hat{\bf g}_m) = {\bf \bT}_m^{\hat{\bf g}_m}:=e^{\hat{\bf g}_m\cdot \ln({\bf \bT}_m)}\\
	{\rm if\ } G_m=e^{2\pi i \ZZ/N}{\rm \ :\ \ }&{\breve{\pi}}_m(\hat{\bf g}_m) = {\bf \bT}_m^{-\frac{i}{2\pi}\ln(\hat{\bf g}_m)}:=e^{\frac{-i}{2\pi}\ln(\hat{\bf g}_m)\cdot \ln({\bf \bT}_m)}.
\end{array}
\label{eq:generator_breve2}
\end{equation}
Note that in the definition of a canonical system (Definition \ref{def:canonical system}) $G$ is a physical quantity. It is straight forward to extend Definition \ref{def:canonical system} to apply also to groups $G$ which are group direct products of physical quantities. We include the Schr\"odinger representation in the following definition (analogous to Definition \ref{def:schrodinger}).
\begin{definition}
$\{N_m,\pi_m,{\bf T}_m,\hat{N}_m,\breve{\pi}_m,{\bf \bT}_m\}$ is called an \textbf{extended canonical system} if 
\begin{enumerate}
	\item 
	$N_m$ is a direct product of physical quantities $G_m\times\ldots\times G_m$, $\hat{N}_m=\hat{G}_m\times\ldots\times \hat{G}_m$,  $\pi_m$ and $\breve{\pi}_m$ are representations of $N_m$ and $\hat{N}_m$ respectively, ${\bf T}_m$ and ${\bf \bT}_m$ are the generators of $\pi_m$ and $\breve{\pi}_m$ respectively, ${\rm spec}({\bf \bT}_m)=N_m$ and ${\rm spec}({\bf T}_m)=\hat{N}_m$, and ${\bT}_m^1,\ldots,\bT_m^{K_m}$ are canonical observable of $\pi^1_m,\ldots,\pi_m^{K_m}$ respectively.
	\item Let $J_m$ be the Heisenberg group associated with $N_m$, then 
	\[\Pi_m(t,g_m,\hat{g}_m) =  e^{2\pi i t}\pi_{m}(g_m)\breve{\pi}_m(\hat{g}_m) , \quad (t,g_m,\hat{g}_m)\in J_m\]
	is called the \textbf{Schr\"odinger representation} of the extended canonical system $\{N_m,\pi_m,{\bf T}_m,\hat{N}_m,\breve{\pi}_m,{\bf \bT}_m\}$.
\end{enumerate}
\label{canonical_system_extend}
\end{definition}
Similarly to Proposition \ref{Schro_is_Heis}, we can show that Schr\"odinger representation is  a representation of $J_m$.
Next we define the analog to a canonical system for the whole group $G$.

\begin{definition}
Consider a SPWT. In the notations of Assumption \ref{ass_voice2}, if for every $m=1,\ldots,M$,  $\{N_m,\pi_m,{\bf T}_m,\hat{N}_m,\breve{\pi}_m\}$ is an extended canonical system, 
 and ${\bf\bT}=({\bf\bT}_1,\ldots,{\bf\bT}_M)$ is a canonical multi-observable, then we call $\{N_m,\pi_m,{\bf T}_m,\hat{N}_m,\breve{\pi}_m\}_{m=1}^M$ a \textbf{multi-canonical system}. 
\end{definition}

The following result extends Proposition \ref{help_construct_canoni}.

\begin{proposition}
\label{proposition55y}
Consider a SPWT, such that $\{\pi,G\}$ are members of a multi-canonical system with canonical multi-observable ${\bf\bT}$. Then for each $m=1,\ldots,M$, there exists a decomposition of $\cH$ to invariant subspaces of $\pi_m$,
\begin{equation}
\cH=\bigoplus_{n\in \kappa_m}\cH_m^n,
\label{eq:decompose_H_global}
\end{equation} 
where $\k_m$ is a discrete index set of size uniquely defined by $\pi$.
For each $m$, there exists a sequence of isometric isomorphisms $U_m^n:\cH_m^n\rightarrow L^2(N_m)$ that satisfy the following,
\begin{enumerate}
	\item 
	Consider the isometric isomorphism $U_m:\cH\rightarrow L^2(N_m)^{\abs{\kappa_m}}$ defined by $U_m=\bigoplus_{n\in\kappa_m}U_m^n$. 
	Consider the pull-forward of $\pi$ to $L^2(N_m)^{\abs{\kappa_m}}$, $\tau_m(g)= U_m\pi(g)U_m^*$.
	We have
	\begin{equation}
\tau_m|_{N_m}({\bf g}_m)=U_m\pi_m({\bf g}_m)U_m^* = L_m({\bf g}_m)^{[\k_m]}.
\label{eq:pull_back_global100}
\end{equation}
where $L_m({\bf g}_m)$ is the left translation in $L^2(N_m)$. 
	\item
	Consider the multiplicative operators $\bQ^k_{N_m}:L^2(N_m)\rightarrow L^2(N_m)$ defined by 
\[\bQ^k_{N_m}f(g_m^1,\ldots,g_m^{K_m})=g_m^k f(g_m^1,\ldots,g_m^{K_m}).\]
Define the multi-multiplicative operator ${\bf \bQ}^{n}_{m}:L^2(N_m)\rightarrow L^2(N_m)^{K_m}$ of the $n$-th copy of  $L^2(N_m)$ in $L^2(N_m)^{\abs{\kappa_m}}$ by
\[{\bf \bQ}^{n}_{m}=(\bQ^1_{N_m},\ldots,\bQ^{K_m}_{N_m}) \quad , \quad n\in\kappa_m.\]
Define the multi-observable ${\bf \bQ}_{m}:L^2(N_m)^{\abs{\kappa_m}}\rightarrow L^2(N_m)^{K_m\abs{\kappa_m}}$ to be
\begin{equation}
{\bf \bQ}_{m} = \bigoplus_{n\in \kappa_m} {\bf \bQ}^{n}_{m}.
\label{eq:MultiQ}
\end{equation}
We have
\begin{equation}
\label{ff8gg833hdd9}
\forall g\in G_z\ . \quad
\tau_m(g)^*{\bf \bQ}_{m} \tau_m(g)={\bf g}_m\bullet{\bf A}({\bf h}_m){\bf \bQ}_{m}.
\end{equation} 

In addition, the canonical multi-observable ${\bf\bT}$ satisfies
\begin{equation}
\forall m=1,\ldots,M \ , \quad {\bf \bT}_m =  U_m^*{\bf \bQ}_{m}U_m .
\label{eq:pull_back_multi_obs1}
\end{equation} 
\end{enumerate}
Moreover, for any sequence of decompositions (\ref{eq:decompose_H_global}), and isometric isomorphisms $U_m^n:\cH_m^n\rightarrow L^2(N_m)$, for $m=1,\ldots,M$ and $n\in\k_m$, that satisfy (\ref{eq:pull_back_global100}) and (\ref{ff8gg833hdd9}), the multi-observable ${\bf\bT}$  defined by (\ref{eq:pull_back_multi_obs1}) is a canonical multi-observable.
\end{proposition}

\begin{proof}

Similarly to the analysis in Subsction \ref{Solving the canonical commutation relation} , by the Stone - von Neumann - Mackey theorem
(Theorem \ref{SVNMS}), for any $m=1,\ldots,M$,
\begin{equation}
\cH=\bigoplus_{n\in \kappa_m}\cH_m^n
\label{eq:decompose_H_global2}
\end{equation}
and each $\Pi_m(h_m)|_{\cH_m^n}$ ($h_m\in J_m$) is unitarily equivalent to the natural representation of $J_m$, $\gamma_m(h_m)=h_m$ in the space
$L^2(N_m)$. Namely, there exist isometric isomorphisms $U_m^n:\cH_m^n\rightarrow L^2(N_m)$ such that
\begin{equation}
U_m^n\Pi_m(h_m)|_{\cH_m^n}U_m^{n\ *} = \gamma_m(h_m).
\label{eq:pull_back_global10}
\end{equation}
Restricting (\ref{eq:pull_back_global10}) to the subgroup $N_m\subset J_m$, we get (\ref{eq:pull_back_global100}).
Restricting (\ref{eq:pull_back_global10}) to the subgroup $\hat{N}_m\subset J_m$, we get
\begin{equation}
U_m^n\bp_m({\bf g}_m)|_{\cH_m^n}U_m^{n\ *} = M_m({\bf g}_m),
\label{eq:pull_back_global1444}
\end{equation}
where $M_m({\bf g}_m)$ are modulations. Equation (\ref{eq:pull_back_global1444}) also applies to the generators, and we get
\begin{equation}
{\bf \bT}_m =  U_m^*{\bf \bQ}_{m}U_m,
\label{eq:pull_back_multi_obs166}
\end{equation} 
which shows (\ref{eq:pull_back_multi_obs1}).
By the fact that $U_m$ maps the spectral family of projections of ${\bf\bT}_m$ to the spectral family of projections of ${\bf\bQ}_{m}$, and keeps the values corresponding to each projection, and by the fact that ${\bf\bT}$ is a canonical multi-observable, we get 
\begin{equation}
\begin{split}
\tau_m(g)^*{\bf \bQ}_{m} \tau_m(g)= &
U_m\pi_m(g)^*U_m^*{\bf \bQ}_{m} U_m\pi_m(g)U_m^*\\
 =&  U_m\pi_m(g)^*{\bf\bT}_m \pi_m(g)U_m^*\\
= & U_m {\bf g}_m\bullet{\bf A}({\bf h}_m){\bf \bT}_{m} U_m^*=
{\bf g}_m\bullet{\bf A}({\bf h}_m){\bf \bQ}_{m}.
\end{split}
\label{eq:5thc0000w}
\end{equation}

which shows (\ref{ff8gg833hdd9}).

The last statement of the proposition follows by pulling backwards ${\bf\bQ}_m$ to $\cH$ via $U_m$, and using a similar calculations to (\ref{eq:5thc0000w}).

\end{proof}

In the following discussion we formulate a more accessible version of
Proposition \ref{proposition55y}.
In the setting of Proposition \ref{proposition55y}, the space $L^2(N_m)^{\abs{\kappa_m}}$ is isomorphic to the space $L^2(X_m)=L^2(N_m\times {\cal Y}_m)$, where ${\cal Y}_m$ is the standard discrete measure space $\{n\}_{n\in\kappa_m}$. The representation $L_m({\bf g}_m)^{[\kappa_m]}$ takes the following form in $L^2(X_m)$. For any $h\in L^2(X_m)$,
\begin{equation}
L_{X_m}({\bf g}_m)h({\bf g}_m',y_m)=h({\bf g}_m^{-1}\bullet{\bf g}_m',y_m)
\label{eq:rrrrrnbh9a}
\end{equation}
Moreover, the multi-observables ${\bf\bQ}_m$ takes the following form in $ L^2(X_m)$. For any $h\in L^2(X_m)$,
\begin{equation}
{\bf\bQ}_{X_m}h({\bf g}_m,y_m)=\big(g_m^1 h({\bf g}_m,y_m) ,\ldots,g_m^{K_m} h({\bf g}_m,y_m) \big).
\label{eq:rrrrt777rrnbh9a}
\end{equation}
Consider the isometric isomorphism $\Psi_m:\cH\rightarrow L^2(X_m)$ that corresponds to $U_m$. Consider the pull-forward representation of $\pi$ to $L^2(X_m)$,  $\rho_m(g)=\Psi_m\pi_m(g)\Psi_m^*$. 
By (\ref{ff8gg833hdd9}) we have
\begin{equation}
\rho_m(g)^*{\bf\bQ}_{X_m}\rho_m(g)={\bf g_m}\bullet {\bf A}_m({\bf h}_m){\bf\bQ}_{X_m}.
\label{eq:to_fulffl_in_23}
\end{equation}
The following theorem formulates Proposition \ref{proposition55y} in terms of the above construction.
\begin{theorem}
\label{Theorem:pull_trans2}
Consider a SPWT, and assume that $\{G,\pi\}$ are members of a multi-canonical system. Then for each $m=1,\ldots,M$, there exists a manifold ${\cal Y}_m$ with a Radon measure, where for $X_m=N_m\times{\cal Y}_m$ there exists an isometric isomorphism $\Psi_m:\cH\rightarrow L^2(X_m)$
 that satisfies $\pi_m({\bf g}_m)= \Psi_m^*L_{X_m}({\bf g}_m)\Psi_m$ and (\ref{eq:to_fulffl_in_23}). 
For any such sequence of transforms $\{\Psi_m\}_{m=1}^M$, the multi-observable ${\bf\bT}=({\bf \bT}_1,\ldots,{\bf \bT}_M)$, defined by ${\bf \bT}_m = \Psi_m^*{\bf \bQ}_{X_m} \Psi_m$, is a canonical multi-observable of $\pi$.
\end{theorem}
 Similarly to Subsection \ref{Solving the canonical commutation relation}, we call $\Psi_m$ the $quantity_{N_m}$ transform, and call $L^2(X_m)$ the $quantity_{N_m}$ domain.
Non-discrete ${\cal Y}_m$ spaces may be used in Theorem \ref{Theorem:pull_trans2} as explained in Remark \ref{Remark:pull_trans1}.

\subsection{Global uncertainties}
\label{Global uncertainties}

In this subsection we define global variances, invariant on orbits, corresponding to each canonical observable ${\bf \bT}_m$. The global uncertainty is then defined to be the sum of the global variances. In this section we are interested in scalar variances $\s^{{\bf W}_m}_{f}({\bf\bT}_m)$. By (\ref{eq:direct_var_linear_comb}), it is enough to focus on variances of the form $\s_f({\bf\bT}_m^{{\bf w}_m})$, where ${\bf\bT}_m^{{\bf w}_m}=\sum_{k=1}^{K_m}w_m^k\bT_m^k$ is a normal operator.  We study the orbit of variances $\{\s^{{\bf w}_m}_{\pi(g)f}({\bf\bT}_m)\ |\ g\in G\}$.
By the Heisenberg point of view (\ref{Hpov2}), and by (\ref{eq:direct_var_linear_comb}), this orbit of variances is equal to the set $\{\s_{f}(\pi(g)^*{\bf\bT}_m^{{\bf w}_m}\pi(g))\ |\ g\in G\}$.
Recall that $Z\subset G$ is represented by $\pi_z(z)$ as the unit operator times a character, and thus $\pi_z(z)$ commutes with any operator. 
Therefore, for $g$ represented in coordinates by $({\bf z},{\bf g}_1,{\bf h}_1)$, we have 
\[\pi({\bf z},{\bf g}_1,{\bf h}_1)^*{\bf\bT}_m^{{\bf w}_m}\pi({\bf z},{\bf g}_1,{\bf h}_1) = \pi({\bf g}_1,{\bf h}_1)^*{\bf\bT}_m^{{\bf w}_m}\pi({\bf g}_1,{\bf h}_1).\]
As a result, it is enough to study the orbits under $G_z$, namely $\{\s^{{\bf W}_m}_{\pi(g)f}({\bf\bT}_m)\ |\ g\in G_z\}$.
We divide the analysis to two cases. The self-adjoint case, where $G_m$ is $\RR$ or $\ZZ$, and the unitary case, where $G_m$ is $e^{i\RR}$ or $e^{2\pi i \ZZ/N}$.

\subsubsection{The self-adjoint case}
\label{Global uncertainties1}

For motivation, we start by considering a SPWT, where for all $m=1,\ldots,M$, $G_m$ is $\RR$.
The definition of the global scalar variance $\Sigma^{{\bf W}_m}_f({\bf \bT}_m)$,  constant on orbits $\pi(G_z)f$, is explained for this special case.  Given $f$, the following analysis shows that there is some element $y$ in the orbit of $f$, having all of its expected values ${\bf e}_f({\bf\bT}_1),\ldots,{\bf e}_f({\bf\bT}_M)$  equal to ${\bf 0}$, which are the unit elements of $N_m$ respectively. Moreover, it shows the way to calculate the unique group element $g\in G_z$ such that $f=\pi(g)y$. 
By (\ref{eq:multi_e_transform}),
\begin{equation}
{\bf e}_f({\bf \bT}) =  {\bf e}_{\pi(g)y}({\bf \bT})=
\left({\bf g}_{1} \bullet {\bf A}_1({{\bf h}_{1}}){\bf e}_y({\bf \bT}_{1})\ \ ,\ \  \ldots\ \   
  ,\ \  {\bf g}_{M-1} \bullet {\bf A}_{M-1}({{\bf h}_{M-1}}){\bf e}_y({\bf \bT}_{M-1})\ \ ,\ \   {\bf g}_{M} \bullet {\bf e}_y({\bf \bT}_{M})\right).
\label{eq:orbit_y_expected}
\end{equation}
The right hand side of (\ref{eq:orbit_y_expected}) can be viewed as the group product in $G_z$ (represented in coordinates) of the element ${\bf g}$ with the element having coordinates ${\bf e}_{y}({\bf \bT})$. 
Thus we have
\[{\bf e}_f({\bf \bT})={\bf g}\bullet {\bf e}_{y}({\bf \bT})={\bf g}.\]
This construction shows that there exists $y$ and a unique $g\in G_z$ such that $y= \pi(g^{-1})f$ has expected values equal to ${\bf 0}$, and shows that $g$ is the group element with coordinates ${\bf e}_f({\bf \bT})$.
Thus, denoting by ${\bf e}_f({\bf \bT})^{-1}$ the inverse group element of ${\bf e}_f({\bf \bT})$ in coordinates of $G_z$, we have
\begin{equation}
y= \pi\Big({\bf e}_f({\bf \bT})^{-1}\Big)f.
\label{eq:fh67ne6}
\end{equation}
Denote by $\left[{\bf e}_f({\bf \bT})^{-1}\right]_{{\bf h}_m}$ the ${\bf h}_m$ component of ${\bf e}_f({\bf \bT})^{-1}$, and note that by Remark \ref{ass_inversion2},
\begin{equation}
\left[{\bf e}_f({\bf \bT})^{-1}\right]_{{\bf h}_m} = \bigg(-{\bf A}_{m'}\Big({\bf e}_f({\bf \bT}_{m+1}),\ldots, {\bf e}_f({\bf \bT}_{M})\Big)^{-1}{\bf e}_f({\bf \bT}_{m'})\bigg)_{m'=m+1}^M.
\label{eq:inverse_e_h}
\end{equation}
Hence, by Proposition \ref{Prop:scalar_var_translation1} and (\ref{eq:fh67ne6}) we have
\[
\s^{{\bf W}_m}_{y}({\bf \bT}_m) = \s^{{\bf A}_m(\left[{\bf e}_f({\bf \bT})^{-1}\right]_{{\bf h}_m}){\bf W}_m{\bf A}_m(\left[{\bf e}_f({\bf \bT})^{-1}\right]_{{\bf h}_m})^*}_{f}({\bf \bT}_m).
\]
This leads us to define the $m$-th scalar global variance to be
\begin{equation}
\Sigma^{{\bf W}_m}_f({\bf \bT}_m) = \s^{{\bf A}_m(\left[{\bf e}_f({\bf \bT})^{-1}\right]_{{\bf h}_m}){\bf W}_m{\bf A}_m(\left[{\bf e}_f({\bf \bT})^{-1}\right]_{{\bf h}_m})^*}_{f}({\bf \bT}_m).
\label{eq:scalar_global_multi-variance1}
\end{equation}
To conclude, $\Sigma^{{\bf W}_m}_f({\bf \bT}_m)$ calculates the scalar uncertainty of the unique window $y\in \pi(G_z)f$, having expected values ${\bf 0}$, and is thus constant on orbits.
Now, we define the uncertainty of the wavelet transform $V_f$ by
\[S(f)= \sum_{m=1}^M\Sigma^{{\bf W}_m}_f({\bf \bT}_m)\]
for some choice of the weights ${\bf W}_m$.

Let us now define the global variance of a self-adjoint canonical multi-observable in the general case. Assume that for some $m$,  $G_m$ is $\RR$ or $\ZZ$. 
In case $G_{m'}\neq \RR$ for some $m'>m$, ${\bf e}_f({\bf \bT}_{m'})$ is not in $N_{m'}$ in general. Therefore, the expression $[{\bf e}_f({\bf \bT})^{-1}]_{{\bf h}_{m'}}$ is meaningless, and (\ref{eq:scalar_global_multi-variance1}) is not well defined. However, there is a way to project ${\bf e}_f({\bf \bT}_{m'})$ to $N_{m'}$ in a way that is consistent with the action of $\pi_{m'}$ on $G_{m'}$, leading to a definition of $\Sigma^{{\bf W}_m}_f({\bf \bT}_m)$ similar to (\ref{eq:scalar_global_multi-variance1}).

\begin{definition}
\label{projected_expected_values0}
The \textbf{projected expected value} of an observable $\bT$, with ${\rm spec}(\bT)=G$ where $G$ a physical quantity, is defined to be the closest point $E_{f}(\bT)\in G$ to $e_f(\bT)$.
\end{definition}

In case there is more than one closest point in $G$ to $e_f(\bT)$, $E_f(\bT)$ is defined in some consistent way. For example, in $\ZZ$ we may round to the smaller integer of the two. In $e^{2\pi i\ZZ/N}$ we project to the point in the clockwise direction. Last, the projection of $e_f(\bT)=0$ in $e^{i\RR}$ and $e^{2\pi i\ZZ/N}$ is not defined.

\begin{remark}
\label{projected_expected_values}
The projected expected values of an observable $\bT$ are given in each of the four cases of physical quantities as follows.
\begin{enumerate}
\item
The projected expected value of a normalized $f$ with respect to the self-adjoint observable $\bT$ with spectrum ${\rm spec}(\bT)=\RR$ is $E_f(\bT)=e_f(\bT)$.
\item
 The \textbf{rounded expected value} of a normalized $f$ with respect to the self-adjoint observable $\bT$ with spectrum ${\rm spec}(\bT)=\ZZ$ is defined to be
\[E_f(\bT) = \left\lfloor e_f(\bT)\right\rfloor\]
where $\left\lfloor x\right\rfloor$ is the closest integer to $x\in\RR$.
	\item 
	The \textbf{expected argument} of a normalized $f$ with respect to the unitary observable $\bT$ with spectrum ${\rm spec}(\bT)=e^{i\RR}$ is defined to be
\[E_f(\bT) = Arg\Big(e_f(\bT)\Big)\]
where $Arg(z)= e^{i\theta}$ for any $z=re^{i\theta}$ with $r,\theta\in\RR$.
	\item 
	The \textbf{rounded expected argument} of a normalized $f$ with respect to the unitary observable $\bT$ with spectrum ${\rm spec}(\bT)=e^{2\pi i\ZZ/N}$ is defined to be
\[E_f(\bT) = \left\lfloor Arg\Big(e_f(\bT)\Big)\right\rfloor\]
where $\left\lfloor e^{i\theta}\right\rfloor$ is the closest point in $e^{2\pi i\ZZ/N}$ to $e^{i\theta}\in e^{i\RR}$.
\end{enumerate}
In each of these cases, we denote the corresponding projection by $\Lambda$. Namely, $\Lambda(z)=z,\left\lfloor z\right\rfloor,Arg(z),\left\lfloor Arg(z)\right\rfloor$ if ${\rm spec}(\bT)=\RR,\ZZ,e^{i\RR},e^{2\pi i\ZZ/N}$ respectively.
\end{remark}

For a SPWT and ${\bf \bT}$ a canonical multi-observable, we define the \textbf{multi-projected expected value} by
\[{\bf E}_f({\bf \bT})=\Big( {\bf E}_f({\bf \bT}_1),\ldots, {\bf E}_f({\bf \bT}_M)  \Big) \]
where for each $m=1,\ldots,M$,
\[{\bf E}_f({\bf \bT}_m)=\Big( E_f(\bT_m^1),\ldots, E_f(\bT_m^{K_m})  \Big).\]

\begin{proposition}
\label{prop:transform_mean_proj}
Consider a SPWT, and a multi-canonical observable ${\bf\bT}$. Then for each $m=1,\ldots,M$, the projected expected values satisfy the one parameter canonical commutation relation
\begin{equation}
{\bf E}_{\pi_m({\bf g}_m)f}({\bf \bT}_m)={\bf g}_m\bullet {\bf E}_{f}({\bf \bT}_m).
\label{eq:transform_mean_proj}
\end{equation}
Equation (\ref{eq:transform_mean_proj}) is not satisfied if $G_m$ is $e^{i\RR}$ or $e^{2\pi i\ZZ/N}$, and $e_f(\bT_m^k)=0$ for some $k$.
\end{proposition}

\begin{proof}
By Remark \ref{projected_expected_values}, and by (\ref{eq:multi_e_transform22}),
\begin{equation}
{\bf E}_{\pi({\bf g}_m)f}({\bf \bT}_m)
=  \Lambda\Big({{\bf e}_{\pi_m({\bf g}_m)f}({\bf \bT}_m)}\Big)  = \Lambda\Big({\bf g}_{m} \bullet  {{\bf e}_{f}({\bf \bT}_m)}\Big)
\label{eq:base_of_indu11}
\end{equation} 
Now, for each of the four cases of $\Lambda$ in Remark \ref{projected_expected_values} we have
\begin{equation}
\begin{split}
\Lambda\Big({\bf g}_{m} \bullet  {{\bf e}_{f}({\bf \bT}_m)}\Big) = {\bf g}_{m} \bullet \Lambda\Big( {{\bf e}_{f}({\bf \bT}_m)}\Big)  =  {\bf g}_{m} \bullet {\bf E}_f({\bf \bT}_m).
\end{split}
\label{eq:base_of_indu12}
\end{equation} 
\end{proof}

In the above notations, note that ${\bf E}_f({\bf \bT})$ is the coordinate representation of some element in $G_z$, so ${\bf E}_f({\bf \bT})^{-1}$ is well defined. Therefore, the following definition of the global scalar variance is legal.

\begin{definition}
\label{def:scalar_global_multi-variance222}
Consider a SPWT, and  a canonical multi-observable ${\bf \bT}$. Let $m$ be an index such that $G_m$ is $\ZZ$ or $\RR$. 
Define the matrix operator 
\begin{equation}
{\bf A}^{-1}_m(f)={\bf A}_m(\left[{\bf E}_f({\bf \bT})^{-1}\right]_{{\bf h}_m}),
\label{eq:A_dilation3}
\end{equation}
where $\left[{\bf E}_f({\bf \bT})^{-1}\right]_{{\bf h}_m}$ is defined as in (\ref{eq:inverse_e_h}), and in case $G_m$ is $e^{i\RR}$ or $e^{2\pi i\ZZ/N}$, and $e_f(\bT)=0$, we define ${\bf A}_m(f)^{-1} = {\bf I}$.
The \textbf{global scalar variance} of ${\bf \bT}_m$ is defined to be
\begin{equation}
\Sigma^{{\bf W}_m}_f({\bf \bT}_m) = \s^{{\bf A}^{-1}_m(f)\ {\bf W}_m\ {\bf A}^{-1}_m(f)^*}_{f}({\bf \bT}_m).
\label{eq:scalar_global_multi-variance222}
\end{equation}
for some weight matrix ${\bf W}_m$. 
\end{definition}

By Proposition \ref{Prop:scalar_var_translation1} and Remark \ref{projected_expected_values}, the value $\Sigma^{{\bf W}_m}_f({\bf \bT}_m)$ is the variance of an element $y\in\pi(G_z)f$ having expected values in the $N_m$ dimensions satisfying
\begin{equation}
{\bf e}_y({\bf \bT})_{{\bf g}_m}=\big[\Lambda \Big({\bf e}_f({\bf \bT})^{-1}\Big)\bullet {\bf e}_f({\bf \bT})\big]_{{\bf g}_m}.
\label{eq:orbit_exp_val0}
\end{equation}
This expected value is in some sense close to the unit element of the group $N_m$. The following  proposition extends this result in the special case of a group $G$, where all of the coordinates ${\bf g}_{m}$ with $G_{m}\neq \RR$, are not dilated in the group product (as defined in the proposition).

\begin{proposition}
\label{ldfgstrh45r7}
Consider a SPWT, such that for every $m$ with $G_{m}\neq\RR$, ${\bf A}_{m}({\bf h}_{m})={\bf I}$. Let ${\bf \bT}$ be a canonical multi-observable. 
Let $f\in\cH$ be a window such that $e_f(\bT_m^k)\neq 0$ for any $m$ such that $G_m$ is $e^{i\RR}$ or $e^{2\pi i\ZZ/N}$, and $k=1,\ldots,K_m$. 
For any $m$ such that $G_m$ is $\RR$ or $\ZZ$, let $\Sigma^{{\bf W}_m}_f({\bf \bT}_m)$ be the global scalar variance of Definition \ref{def:scalar_global_multi-variance222}. 
Define $\Sigma^{{\bf W}_m}_f({\bf \bT}_m)=\s^{{\bf W}_m}_f({\bf \bT}_m)$ for any $m$ such that $G_m$ is $e^{i\RR}$ or $e^{2\pi i\ZZ/N}$.
Then for any index $m$, the global scalar variance $\Sigma^{{\bf W}_m}_f({\bf \bT}_m)$ is constant on orbits $\pi(G_z)f$.
Moreover, 
$\Sigma^{{\bf W}_m}_f({\bf \bT}_m)$ is the variance of the unique element $y\in\pi(G_z)f$ having multi-expected value
\begin{equation}
{\bf e}_y({\bf \bT})={\bf E}_f({\bf \bT})^{-1}\bullet {\bf e}_f({\bf \bT}).
\label{eq:orbit_exp_val}
\end{equation}
\end{proposition}

Proposition \ref{ldfgstrh45r7} is used for defining the global uncertainty as follows. Assume the conditions of Proposition \ref{ldfgstrh45r7} are satisfied. Therefore, the global variances $\Sigma^{{\bf W}_m}_f({\bf \bT}_m)$ are invariant on orbits. Thus, the uncertainty
\begin{equation}
S(f)=\sum_{m=1}^{M}\Sigma^{{\bf W}_m}_f({\bf \bT}_m)
\label{eq:Global_Uncertainty}
\end{equation}
is constant on orbits for any choice of the weights ${\bf W}_m$. Hence, $S(f)$ is  interpreted as an uncertainty of the SPWT $V_f$, and not of the individual window $f$, and is called the \textbf{global uncertainty}. { Of course, the global uncertainty (\ref{eq:Global_Uncertainty}) can be defined alternatively using the product of the global variances, instead of their sum. The product based global uncertainty of the STFT coincides with the classical time-frequency uncertainty. Indeed, $G/Z$ is $\{\RR^2,+\}$ and the semi-direct product reduces to a direct product, in addition to the fact that the set of time-frequency infinitesimal generators coincide with the canonical observables up to sign. However, in the generic case the global uncertainty is novel. We thus see the global uncertainty as a generalization of the classical time-frequency uncertainty.}

To prove Proposition \ref{ldfgstrh45r7}, we present the following lemma, which can be seen as the projected version of (\ref{eq:multi_e_transform}) or as an extension of Proposition \ref{prop:transform_mean_proj}.
\begin{lemma}
\label{LEMMA:proj_trans}
Consider a SPWT, such that for every $m$ with $G_{m}\neq\RR$, ${\bf A}_{m}({\bf h}_{m})={\bf I}$. Let ${\bf \bT}$ be a canonical multi-observable. 
Let $f\in\cH$ be a window such that $e_f(\bT_m^k)\neq 0$ for any $m$ such that $G_m$ is $e^{i\RR}$ or $e^{2\pi i\ZZ/N}$, and $k=1,\ldots,K_m$.
Then for any $g\in G_z$, 
\begin{equation}
{\bf E}_{\pi(g)f}({\bf \bT})={\bf g}\bullet {\bf E}_{f}({\bf \bT}).
\label{eq:LEMMA:proj_trans}
\end{equation}
\end{lemma}

\begin{proof}
We prove 
\begin{equation}
[{\bf E}_{\pi(g)f}({\bf \bT})]_{{\bf h}_m}={\bf h}_m\bullet [{\bf E}_{f}({\bf \bT})]_{{\bf h}_m}.
\label{eq:LEMMA:proj_trans00}
\end{equation}
by induction on $m$. In the base of the induction, $m=M-1$. By Proposition \ref{prop:transform_mean_proj}, noting that ${\bf h}_{M-1}={\bf g}_M$, we have
\begin{equation}
[{\bf E}_{\pi(g)f}({\bf \bT})]_{{\bf h}_{M-1}}  =  {\bf h}_{M-1} \bullet [{\bf E}_f({\bf \bT})]_{{\bf h}_{M-1}}.
\label{eq:base_of_indu1}
\end{equation} 

For the induction step, assume (\ref{eq:LEMMA:proj_trans00}) is true for $m+1$, and prove it for $m$. 
We have
\begin{equation}
[{\bf E}_{\pi(g)f}({\bf \bT})]_{{\bf h}_{m}} 
= \Big(
{\bf E}_{\pi(g)f}({\bf \bT}_{m+1})\ ,\ [{\bf E}_{\pi(g)f}({\bf \bT})]_{{\bf h}_{m+1}} \Big)
\label{eq:hjl8fjs72}
\end{equation}
and by the induction assupmtion, 
\begin{equation}
[{\bf E}_{\pi(g)f}({\bf \bT})]_{{\bf h}_{m+1}} 
=
 {\bf h}_{m+1}\bullet [{\bf E}_{f}({\bf \bT})]_{{\bf h}_{m+1}} .
\label{eq:easy_setp34}
\end{equation}

If $m+1$ has $G_{m+1}\neq \RR$, then by assumption we have ${\bf A}_{m+1}={\bf I}$, and by Proposition \ref{prop:transform_mean_proj}
\begin{equation}
{\bf E}_{\pi(g)f}({\bf \bT}_{m+1})= {\bf g}_{m+1}\bullet {\bf E}_{f}({\bf \bT}_{m+1}).
\label{eq:secondffff1}
\end{equation}
In case $G_{m+1}=\RR$, we have by (\ref{eq:multi_e_transform})
\begin{equation}
\begin{split}
{\bf E}_{\pi(g)f}({\bf \bT}_{m+1})= & {\bf e}_{\pi(g)f}({\bf \bT}_{m+1})= {\bf A}_{m+1}({\bf h}_{m+1}){\bf e}_{f}({\bf \bT}_{m+1}) + {\bf g}_{m+1}\\
= & {\bf A}_{m+1}({\bf h}_{m+1}){\bf E}_{f}({\bf \bT}_{m+1}) + {\bf g}_{m+1}
\end{split}
\label{eq:secondffff2}
\end{equation}
Equations (\ref{eq:secondffff1}) and (\ref{eq:secondffff2}) give the leftmost coordinate ${\bf g}_{m+1}$ of the group product in (\ref{eq:LEMMA:proj_trans00}), and (\ref{eq:easy_setp34}) is the remaining coordinates ${\bf h}_{m+1}$, which proves (\ref{eq:LEMMA:proj_trans00}).

\end{proof}

\begin{proof}[Proof of Proposition \ref{ldfgstrh45r7}]

First consider the case where $G_m\neq\RR$.
By the assumption that ${\bf A}_m({\bf h}_m)={\bf I}$, the global variance is $\Sigma^{{\bf W}_m}_f({\bf \bT}_m)=\s^{{\bf W}_m}_f({\bf \bT}_m)$. Moreover, by Propositions \ref{Prop:scalar_var_translation1} and \ref{Prop:scalar_var_translation2}, $\Sigma^{{\bf W}_m}_f({\bf \bT}_m)$ is constant on orbits.

Next consider the case where $G_m=\RR$.
Denote $y=\pi({\bf E}_{f}({\bf \bT})^{-1})f$. 
By Definition \ref{def:scalar_global_multi-variance222}, and by the fact that ${\bf A}_m(\cdot)$ is a group action of $H_m$,
\[{\bf A}_m\big({\bf E}_{f}({\bf \bT})_{{\bf h}_m}\big){\bf A}^{-1}_m(f) =
{\bf A}_m\big({\bf E}_{f}({\bf \bT})_{{\bf h}_m}\big){\bf A}_m\big(\left[{\bf E}_f({\bf \bT})^{-1}\right]_{{\bf h}_m}\big) ={\bf I}.\]
Therefore, by (\ref{scalar_var_translation1})   
\begin{equation}
\Sigma^{{\bf W}_m}_{f}({\bf \bT}_m)\  =\  
\s^{{\bf A}^{-1}_m(f)\ {\bf W}_m\ {\bf A}^{-1}_m(f)^*}_{\pi\big({\bf E}_{f}({\bf \bT})\big)y}({\bf \bT}_m)
\ =\  \s^{{\bf W}_m}_{y}({\bf \bT}_m).
\label{dfgfghft756uyghdrtyr62}
\end{equation}
Let $\pi(g)f$ be some element in the orbit $\pi(G_z)f$. Then, by Lemma \ref{LEMMA:proj_trans}
\begin{equation}
\label{eq:5455565jh400}
{\bf A}^{-1}_m(\pi(g)f)
=
{\bf A}_m(\left[{\bf E}_{\pi(g)f}({\bf \bT})^{-1}\right]_{{\bf h}_m})
=
{\bf A}_m(\left[{\bf E}_{f}({\bf \bT})^{-1 }\bullet  {\bf g}^{-1}\right]_{{\bf h}_m})
\end{equation}
so
\begin{equation}
\Sigma^{{\bf W}_m}_{\pi(g)f}({\bf \bT}_m) =
\s^{{\bf A}_m(\left[{\bf E}_{f}({\bf \bT})^{-1 }\bullet  {\bf g}^{-1}\right]_{{\bf h}_m}){\bf W}_m{\bf A}_m(\left[{\bf E}_{f}({\bf \bT})^{-1 }\bullet  {\bf g}^{-1}\right]_{{\bf h}_m})^*}_{\pi(g)f}({\bf \bT}_m).
\label{eq:5455565jh4}
\end{equation}
Note that $\pi(g)f=\pi\big({\bf g}\bullet E_{f}({\bf \bT})\big)y$, so by (\ref{scalar_var_translation1}) and by the fact that ${\bf A}_m(\cdot)$ is a group action, (\ref{eq:5455565jh4})  gives
\begin{equation}
\Sigma^{{\bf W}_m}_{\pi(g)f}({\bf \bT}_m) = \s^{{\bf W}_m}_{y}({\bf \bT}_m).
\label{eq:hhmbc4gaa}
\end{equation}
Here, (\ref{eq:hhmbc4gaa}) is true by
\[ {\bf A}_m(  {\bf h}_m   \bullet  {\bf E}_{f}({\bf \bT})_{{\bf h}_m}){\bf A}_m(\left[{\bf E}_{f}({\bf \bT})^{-1 }\bullet  {\bf g}^{-1}\right]_{{\bf h}_m})= {\bf I}.\]
To conclude (\ref{dfgfghft756uyghdrtyr62}) and (\ref{eq:hhmbc4gaa}),
\[\Sigma^{{\bf W}_m}_{\pi(g)f}({\bf \bT}_m) =  \Sigma^{{\bf W}_m}_{f}({\bf \bT}_m).\]
Hence, $\Sigma^{{\bf W}_m}_{f}({\bf \bT}_m)$ is constant on orbits, and equal to $\s^{{\bf W}_m}_{y}({\bf \bT}_m)$, where $y$ is unique for each orbit $\pi(G_z)f$. 
The expected value (\ref{eq:orbit_exp_val}) of $y$ follows (\ref{eq:multi_e_transform}) and (\ref{eq:multi_e_transform22}).
\end{proof}

\subsubsection{The unitary case}
\label{Global uncertainties2}

Next we treat the unitary case, where $G_m$ is $e^{i\RR}$ or $e^{2\pi i\ZZ/N}$.
In this section we restrict ourselves to scalar variances along the axis, namely $\s^{{\bf W}_m}({\bf\bT}_m)=\sum_{k=1}^{K_m}w_m^k\s_f(\bT_m^k)$, for some scalar weights $w_m^k$.
 First note that by the unitarity of each $\bT^k_m$, we have $\s_f(\bT_m^k)=1-\abs{e_f(\bT_m^k)}^2$. 
Therefore, it is enough to study $e_f(\bT_m^k)$, $k=1,\ldots,K_m$.
Consider the $quantity_m$ transform $\Psi_m:\cH\rightarrow L^2(N_m\times{\cal Y}_m)$ guaranteed by Theorem \ref{Theorem:pull_trans2}. 
In the notation of Theorem \ref{Theorem:pull_trans2}, we have
\[L_{X_m}({\bf g}_m)= \Psi_m \pi_m(g) \Psi_m^*,\]
and ${\bf\bT}_m=\Psi_m^* {\bf\bQ}_{X_m} \Psi_m$. 
This means that we can pull forward the whole discussion from the canonical system $\{N_m,\pi_m,{\bf T}_m,\hat{N}_m,\breve{\pi}_m,{\bf \bT}_m\}$ to the concrete $quantity_m$ domain, with the standard translation $L_{X_m}({\bf g}_m)$ and standard multi-observable ${\bf\bQ}_{X_m}$. 
Since expected values and variances are based on inner products, they are invariant under isometric isomorphisms, and we have the following property.
\begin{proposition}
\label{gamma_dilation0}
Under the above construction,
\[{\bf e}_f({\bf \bT}_m) = {\bf e}_{ \Psi_m f}({\bf \bQ}_{X_m}) \quad ,\quad {\boldsymbol{\s}}_f({\bf \bT}_m) = {\boldsymbol{\s}}_{\Psi_m f}({\bf \bQ}_{X_m})  \quad ,\quad   \s^{{\bf W}_m}_f({\bf \bT}_m) = \s^{{\bf W}_m}_{\Psi_m f}({\bf \bQ}_{X_m}).\]
\end{proposition}

By Proposition \ref{gamma_dilation0}, it is enough to study the localization of ${\bf \bQ}_{x_m}$ in $L^2(N_m\times{\cal Y}_m)$. The inner product in $L^2(N_m\times{\cal Y}_m)$ is based on integration (along the $N_m$ axis),  which is based on additions. Since the group multiplication in the unitary case is the arithmetic product, the calculation of the variances is only consistent with the group multiplication in the self-adjoint case, where $G_m$ is  $\RR$ or $\ZZ$ and $\bullet$ is $+$. Hence, there are no localization transformation properties for unitary observables analogous to Proposition \ref{Prop:scalar_var_translation1} in case ${\bf A}_m({\bf h}_m)\neq {\bf I}$.
Defining global variances in the unitary case, invariant on orbits, requires a different approach.

Since Proposition \ref{gamma_dilation0} allows to restrict the analysis to the space $L^2(N_m\times{\cal Y}_m)$ and the multi-observable ${\bf\bQ}_{X_m}$, we omit the subscript $m$, and simply denote the space by $L^2(G^K\times{\cal Y})$, and the multi-observable by ${\bf \bQ}=(\bQ_1,\ldots,\bQ_K)$. Here, $G$ is the physical quantity $e^{i\RR}$ or $e^{2\pi i\ZZ/N}$, and $\bQ_k f(g_1,\ldots g_K,y) = g_k f(g_1,\ldots g_K,y)$. We assume without loss of generality that the signal space is $\cH=L^2(G^K\times{\cal Y})$.

Denote the Fourier transform in $L^2(G^K)$ by $\cF_{G^K}:L^2(G^K)\rightarrow L^2(\hat{G}^K)$. Note that $\bQ_1,\ldots,\bQ_K$ are multiplications by characters of $G^K$, independent of the variable $y$. Thus, by abuse of notation, we treat each $\bQ_k$ as the function $\bQ_k(g_1,\ldots g_K,y)=\bQ_k(g_1,\ldots g_K)=g_k$. Note that each $\bQ_k$ can be treated as the unit frequency element of $\cF_{G^K}$ along the $k$ axis. 
Consider the calculation of the expected values 
\begin{equation}
\begin{split}
{\bf e}_f({\bf \bQ})= &  \Big(\ip{\bQ_1f}{f},\ldots, \ip{\bQ_Kf}{f}\Big)\\
= & \Big(\int_{\cal Y}\int_{G^k} \bQ_1({\bf g},y) \abs{f({\bf g},y) }^2 d {\bf g} dy ,\ldots, \int_{\cal Y}\int_{G^k} \bQ_K({\bf g},y) \abs{f({\bf g},y) }^2 d {\bf g} dy \Big)\\
= & \Big(\int_{G^k} \bQ_1({\bf g}) \int_{\cal Y}\abs{f({\bf g},y) }^2 dy d {\bf g} ,\ldots, \int_{G^k} \bQ_K({\bf g}) \int_{\cal Y}\abs{f({\bf g},y) }^2 dy  d {\bf g} \Big)
\end{split}
\label{eq:r66756544b6h}
\end{equation}
Consider the function $F\in L^1(G^K)$, defined by $F({\bf g})=\int_{\cal Y}\abs{f({\bf g},y)}^2dy$.
By (\ref{eq:r66756544b6h}), the expected values ${\bf e}_f({\bf \bQ})$
are the $-1$ Fourier coefficients, along the axis of $G^K$, of the function $F$. Namely,
\[{\bf e}_f({\bf \bQ})=\Big([\cF_{G^K}F](-1,0,\ldots,0),\ldots, [\cF_{G^K}F](0,\ldots,0,-1)\Big),\]
where $(0,\ldots,0,-1,0,\ldots,0)$ denotes the character $\overline{\bQ_k}$ in coordinates of $\hat{G}^K$.
By Remark \ref{remark:MatrixAuto}, $A_m({\bf h}_m,\cdot)$ is written as the matrix ${\bf A}_m({\bf h}_m)$,
with homomorphisms ${\bf a}_{k',k}({\bf h}_m):G\rightarrow G$ as entries. These homomorphisms are multiplication of the exponent by the real numbers $a_{k',k}({\bf h}_m)$, namely 
\[G\ni e^{i\w}\mapsto {\bf a}_{k',k}({\bf h}_m)(e^{i\w}) =e^{ia_{k',k}({\bf h_m})\w}\in G.\]
Consider the orbit ${\bf A}_m(H_m)\bQ_k$. Note that homomorphisms of $G$, applied on the value of characters of $G^K$, map them to characters. Namely, for the character $\bQ_{k}\in \chi(G)$,
\[{\bf a}_{k',k}({\bf h}_m)\circ\bQ_{k}\in \chi(G).\]
 Thus, the collection of entries of the orbit ${\bf A}_m(H_m){\bf\bQ}$, is a set of characters. By the Heisenberg point of view (\ref{Hpov2}), the orbit of expected values 
\[{\cal O}_f^m=\{e_{\pi(g)f}(\bQ_m^k)\ |\ k=1,\ldots,K_m  \ \ , \ \ g\in G_k \}\]
 is a set of values of $\cF_{G^K}F$. Since we are interested in defining a variance over the whole orbit, we define $\Sigma_f({\bf\bT})$ as some norm of the Fourier coefficients of $F$ in ${\cal O}_f^m$.

\begin{remark}
Consider the special case where ${\cal Y}=\{1\}$, and the collection of entries of the orbit ${\bf A}_m(H_m){\bf\bQ}$ are all of the frequencies. In this case we define the mean square average expected value
\begin{equation}
{\cal E}_f({\bf\bQ}) = \norm{\cF_{G^K}\abs{f}^2}^2_2.
\label{eq:Global_var_unitary}
\end{equation}
By Parseval's theorem we can calculate (\ref{eq:Global_var_unitary}) in the $G^K$ domain by
 \begin{equation}
{\cal E}_f({\bf\bQ}) =\norm{\abs{f}^2}^2_2=\int_{G^{K}}\abs{f(\bf g)}^4d{\bf g},
 \label{eq:hhhhhhhhhh654}
 \end{equation} 
and define the global variance  by
\[\Sigma_f({\bf\bQ}) =1-\abs{\int_{G^{K}}\abs{f(\bf g)}^4d{\bf g}}^2.\]
 Note that we want to minimize $\Sigma_f({\bf\bQ})$ under $\norm{f}^2=1$, so the definition promotes localization. 
\end{remark}

\subsection{Examples}

In this section we give five examples. First,
the observables of the STFT from Subsection \ref{Motivation for defining new uncertainty principles} and the observables of FSTFT from Subsection \ref{Localization framework for finite short-time-Fourier-tansform} constitute canonical multi-observables. 
Additionally, the global multi-observable of the 1D wavelet transform was developed in Subsection \ref{The global localization framework for the wavelet transform}. Next we develop the localization theory of the Shearlet transform and the finite wavelet transform.

\subsubsection{The Shearlet transform}
\label{The Shearlet transform}

The Shearlet transform is a modification of the Curvelet transform, making it a SPWT. The modification is based on replacing rotations with shears.
Hence, the Shearlet transform comprises translations, shears, and anisotropic dilations of a window in $L^2(\RR^2)$ \cite{Shearlet}. In \cite{Affine_uncertainty2}, the Shearlet transform was studied as a generalized wavelet transform, including the group structure and the representation generators. There, for the localization notions, the canonical observables were defined to be the generators of the representations $\pi_m$. In this section we apply our localization theory for the Shearlet transform.

Translation by ${\bf g}_1\in \RR\times\RR$ is defined as usual by $\pi_1({\bf g}_1)f({\bf x})=f({\bf x}-{\bf g}_1)$. Consider the shear matrix operator, with $g_2\in\RR$,
\[S_{g_2}= \left(
\begin{array}{cc}
	1 & g_2\\
	0 & 1
\end{array}
\right).\] 
Shear by $g_2\in\RR$ of $L^2(\RR^2)$ functions is defined by
\[\pi_2(g_2)f({\bf x}) = f(S_{g_2}^{-1}{\bf x}).\]
Consider the anisotropic dilation matrix operator, with $g_3\in\RR$,
\[D_{g_3}= \left(
\begin{array}{cc}
	e^{g_3} & 0\\
	0 & e^{\frac{1}{2}g_3}
\end{array}
\right).\] 
Anisotropic dilation by $g_3\in\RR$ of $L^2(\RR^2)$ functions is defined by
\[\pi_3(g_3)f({\bf x}) = e^{-\frac{3}{4}g_3} f(D_{g_3}^{-1}{\bf x}).\]
Last, consider the reflection by $g_4\in\{-1,1\}$
\[\pi_2(g_4)f({\bf x}) = f(g_4{\bf x}).\]
Note that the standard definition of the Shearlet transform is based on the anisotropic dilation, with $g_3\in\RR$,
\[\tilde{D}_{g_3}= \left(
\begin{array}{cc}
	g_3 & 0\\
	0 & {\rm sign}({g_3})\sqrt{\abs{g_3}}
\end{array}
\right).\] 
The standard definition incorporates dilations and reflections. Our version gives rise to a Shearlet group that is isomorphic to the standard Shearlet group, and compatible with our localization theory.

The Shearlet group 
\begin{equation}
\begin{split}
G=&\big(translations\big)\rtimes\big(shears\rtimes (dilations\times reflections)\big)\\
=&\big(\RR\times\RR\big)\rtimes\big(\RR\rtimes (\RR\times \{-1,1\})\big)
\end{split}
\label{eq:f5hss6ppp}
\end{equation}
has an empty center, and is represented in the Shearlet transform by
\[\pi(g)=\pi_1(g_1)\pi_2(g_2)\pi_3(g_3)\pi_4(g_4).\]
In the notation of Assamption \ref{ass_voice2}, we have $G_1=G_2=G_3=\RR$, $G_4=\{-1,1\}$,  $N_1=G_1\times G_1$, $N_2=G_2$, $N_3=G_3$, $N_4=G_4$, and
\[G=\big(G_1\times G_1\big)\rtimes\big(G_2\rtimes (G_3\rtimes G_4)\big)\]
where the last $\rtimes$ is actually $\times$.
The actions ${\bf A}_m({\bf h}_m)$ in the semi-direct product group structure are given next. 
Denote the reflection matrix operator $I_{g_4}=g_4I$, and observe
\[D_{g_3}I_{g_4}S_{g_2}I_{g_4}D_{-g_3} = S_{g_4e^{\frac{1}{2}g_3}g_2},\]
so
\[{\bf A}_2(g_3,g_4)g_2=g_4e^{\frac{1}{2}g_3}\ g_2.\]
To calculate ${\bf A}_1(g_2,g_3,g_4)$, observe 
\[f\Big([I_{g_4}D_{-g_3}S_{-g_2}]^{-1}\big([I_{g_4}D_{-g_3}S_{-g_2}]{\bf x}-{\bf g}_1\big)\Big)
=f\left({\bf x}-[I_{g_4}D_{-g_3}S_{-g_2}]^{-1}{\bf g}_1\right)\]
so
\[{\bf A}_1(g_2,g_3,g_4){\bf g}_1=S_{g_2}D_{g_3}I_{g_4}{\bf g}_1=
g_4\left(
\begin{array}{cc}
	e^{g_3} & e^{\frac{1}{2}g_3}g_2\\
	0 & e^{\frac{1}{2}g_3}
\end{array}
\right){\bf g}_1.
\]

To construct a canonical multi-observable, we transform the discussion to the frequency domain. By (\ref{eq:General_B_Fourier}),
\[\hat{\pi}_2(g_2)\hat{f}(\boldsymbol{\w}):= \cF\pi_2(g_2)\cF^{-1} \hat{f}(\boldsymbol{\w}) = \hat{f}(\hat{S}_{g_2}^{-1}\boldsymbol{\w})\]
where $\hat{S}_{g_2}$ is the orthogonal shear, defined by
\[\hat{S}_{g_2}=\left(
\begin{array}{cc}
	1 & 0\\
	-g_2 & 1
\end{array}
\right).\]
Morefover,
\[\hat{\pi}_3(g_3)\hat{f}(\boldsymbol{\w}):= \cF\pi_3(g_3)\cF^{-1} \hat{f}(\boldsymbol{\w}) = e^{\frac{3}{4}g_3}\hat{f}(D_{g_3}\boldsymbol{\w}),\]
and $\hat{\pi}_4(g_4)\hat{f}(\boldsymbol{\w})= \hat{f}(g_4\boldsymbol{\w})$.
Let us define the  canonical multi-observable directly in the frequency domain. 
For translation, the natural definition is
\[{\bf \bT}_1 \hf(\boldsymbol{\w}) = \big(i\frac{\partial}{\partial \w_1}\hf(\boldsymbol{\w}),i\frac{\partial}{\partial \w_2}\hf(\boldsymbol{\w}) \big).\]
We call the physical quantity translated by shears \textit{slope}, and define the slope observable
\[\bT_2 \hf(\boldsymbol{\w})=-\frac{\w_2}{\w_1} \hf(\boldsymbol{\w}).\]
Note that the slope $-\frac{\w_2}{\w_1}$, corresponding to the point $(\w_1,\w_2)$ in the frequency domain,  is a measure of direction or angle.
For dilations, we take the anisotropic scale observable
\[\bT_3 \hf(\boldsymbol{\w})=-\ln(\abs{\w_1})\hf(\boldsymbol{\w}).\]
Last, for reflections
\[\bT_4 \hf(\boldsymbol{\w})={\rm sign}(\w_1)\hf(\boldsymbol{\w}).\]
It is straight forward to check that ${\bf \bT}$ is a canonical multi-observable.

By (\ref{eq:scalar_global_multi-variance1}), the global variances are defined to be
\begin{equation}
\Sigma^{{\bf W}_m}_{\hf}({\bf\bT}_m)= \s_{\hat{f}}^{A_m([e_{\hf}(\bT)^{-1}]_{{\bf h}_m}){\bf W}_m A_m([e_{\hf}(\bT)^{-1}]_{{\bf h}_m})^*}({\bf\bT}_m).
\label{eq:scalar_shear_V}
\end{equation}
To calculate (\ref{eq:scalar_shear_V}) in practice, we use the inversion formula in Remark \ref{ass_inversion2} on $[e_{\hf}(\bT)^{-1}]_{{\bf h}_m}$.
Recall that $({\bf g}_m,{\bf h}_m)^{-1}|_{{\bf g}_m}={\bf A}_m({\bf h}_m)^{-1}{\bf g}_m^{-1}$, so
\[({\bf g}_1,g_2,g_3,g_4)^{-1}=\Bigg( -g_4\left(
\begin{array}{cc}
	e^{g_3} & e^{\frac{1}{2}g_3}g_2\\
	0 & e^{\frac{1}{2}g_3}
\end{array}
\right)^{-1}{\bf g}_1, -e^{-\frac{1}{2}g_3}g_2,-g_3,g_4\Bigg).\]
Thus, the scalar global variance of ${\bf \bT}_1$, with weight ${\bf W}_1$, is given by
\[\Sigma^{{\bf W}_1}_{\hf}({\bf\bT}_1)= \s^{{\bf A}_1(f){\bf W}_1 {\bf A}_1(f)^* }_{\hf}({\bf \bT}_1),\]
where
\[{\bf A}_1(f)={\bf A}_1\left(-e^{-\frac{1}{2}e_{\hf}(\bT_3)}e_{\hf}(\bT_2),-e_{\hf}(\bT_3),e_{\hf}(\bT_4)\right).\]
The scalar global variance of $\bT_2$ is given by
\[\Sigma^{{\bf W}_2}_{\hf}({\bT}_2)= \s^{{\bf A}_2(f){\bf W}_2 {\bf A}_2(f)^* }_{\hf}({ \bT}_2),\]
where
\[{\bf A}_2(f)={\bf A}_2\left(-e_{\hf}(\bT_3),e_{\hf}(\bT_4)\right).\]
Last, $\Sigma^{{\bf W}_3}_{\hf}({\bT}_3)=\s^{{\bf W}_3}_{\hf}({\bT}_3)$.
Next we define the Shearlet global uncertainty.
Note that windows supported on the domain
\[\RR_+\times\RR = \{\boldsymbol{\w}\in\RR^2\ |\ \w_1>0\}\]
 are perfectly concentrated with respect to $\bT_4$. Thus, we restrict our search to windows supported on $\RR_+\times\RR$, and define
\[S(\hf)= \Sigma^{{\bf W}_1}_{\hf}({\bf\bT}_1)+\Sigma^{{\bf W}_2}_{\hf}({\bT}_2)+\Sigma^{{\bf W}_3}_{\hf}({\bT}_3)\]
for some choice of the weights ${\bf W}_1,{\bf W}_2,{\bf W}_3$.

Let us introduce the transformations corresponding to Theorem \ref{Theorem:pull_trans2}.
Shearing translates the variable $g_2=-\frac{\w_2}{\w_1}$, and keeps $\w_1$ constant. The inversion of this change of variable is $\w_2=-g_2\w_1$, and $\w_1$ kept unchanged. By normalizing this change of variable, we define the slope transform to be
\[\Psi:L^2(\RR^2)\rightarrow L^2(\RR^2) \quad , \quad[\Psi f](g_2,\w_1)=\abs{\w_1}^{\frac{1}{2}}\hf(\w_1,-g_2\w_1).\]
The anisotropic scale transform for $\hf$ supported on $\w_1>0$ is given in by
\[[{\cal W}_+ \hf](g_3,y)=e^{-\frac{3}{4}g_3}\hf(e^{-g_3},y e^{-\frac{1}{2}g_3}),\]
and for $\hf$ supported on $\w_1<0$ by
\[[{\cal W}_- \hf](g_3,y)=e^{-\frac{3}{4}g_3}\hf(-e^{-g_3},-y e^{-\frac{1}{2}g_3}).\]
Together, we define
\[{\cal W}:L^2(\RR^2)\rightarrow L^2(\RR^2)^2 \quad , \quad {\cal W}= {\cal W}_+ \oplus {\cal W}_-. \]

\subsubsection{The finite wavelet transform}

The following generalized wavelet transform can be traced back to \cite{finite_wavelet}.
Consider the set $G=e^{2\pi i\ZZ/N}$ for a prime number $N$. This set is a finite field with $e^{2\pi in/N}+e^{2\pi im/N}:=e^{2\pi i(n+m)/N}$ and $e^{2\pi in/N}\cdot e^{2\pi im/N}:=e^{2\pi inm/N}$. Consider the space $L^2(e^{2\pi i\ZZ/N})$. We define the finite wavelet transform directly in the frequency domain. Translations are defined in time by $\pi_1(g_1)f(q)=f(q-g_1)$, or in frequency by $\hat{\pi}_1(g_1)\hf(q)=q\cdot g_1\hf(q)$, where $q\cdot g_1$ is the multiplication in the field. Here $\hat{\pi}_1$ is a representation of the additive group $G_+$ of $e^{2\pi i\ZZ/N}$. Consider the multiplicative group $G_{\times}$ of $e^{2\pi i\ZZ/N}$, namely the group $e^{2\pi i\ZZ/N}\setminus\{0\}$ with the field's multiplication as the group product.
Dilations are defined by $\hat{\pi}_2(g_2)\hf(q)=\hf(g_2\cdot q)$, for $g_2\in G_{\times},q\in e^{2\pi i\ZZ/N}$. Here $\hat{\pi}_2$ is a representation of $G_{\times}$.
It can be shown that $\pi_1,\pi_2$ are unitary representations, and $\pi(g)=\pi_1(g_1)\pi_2(g_2)$ is a unitary representation of the finite affine group 
\[translations\rtimes dilations= G_+\rtimes G_{\times}.\]
 In the following we represent elements $e^{2\pi i n/N}$ in short by $n$.
The representation $\pi$ has two irreducible subspaces, namely 
\begin{equation}
\begin{split}
\cH_0= & \{f\in L^2(e^{2\pi i\ZZ/N})\ |\ \forall\ n\neq 0, \hf(n)=0 \ \} \\
\cH_1= & \{f\in L^2(e^{2\pi i\ZZ/N})\ |\ \hf(0)=0\}.
\end{split}
\label{eq:}
\end{equation}
As a representation of a finite group, irreducible in each of $\cH_0$ and $\cH_1$, $\pi$ satisfies Assumption \ref{ass_voice2} in each of these irreducible subspaces. Moreover, since $\pi(g)|_{\cH_0}$ and $\pi(g)|_{\cH_1}$ are not unitarily equivalent, by finite group representation theory, $V_{\cH_0}[\cH_0]$ and $V_{\cH_1}[\cH_1]$ are orthogonal subspaces of $L^2(G)$, and the reconstruction formula 7 of Assumption \ref{ass_voice} holds also in the reducible space $ L^2(e^{2\pi i\ZZ/N})$ (this is by the canonical decomposition of the representation $\pi$, see e.g \cite{Finite_representation_book}).

The following choice of $\bT_1$ and $\bT_2$ is a multi-canonical observable. Define $[\bT_1f](n)=e^{2\pi in/N}f(n)$ in time,
or $[\cF\bT_1\cF^{-1}\hf](n)=\hf(n-1)$ in frequency.
For scale, define
$[\cF\bT_2\cF^{-1}\hf](2^m)=e^{2\pi i m/N}\hf(2^m)$, where $2^m$ is defined using the multiplication in the field $e^{2\pi i\ZZ/N}$, and $m\in\ZZ$. Note that for $m\in \ZZ$, $2^m$ exhausts the elements in $e^{2\pi i\ZZ/N}$, so $\bT_2$ is well defined. Since $2^m$ is an ``exponential scale'' of the doamin of definition of $\hf$, and $e^{2\pi i m/N}$ is in a ``linear scale'' in the image of $\hf$, the scale observable $\bT_2$ is interpreted as a frequency logarithmic observable.

It can be shown that the orbit of ${\bf A}_2(g_2)\bT_1$ are the multiplicative operators by all of the characters of $G_+$, except for the unit character. Thus, by (\ref{eq:Global_var_unitary}), we define the global scalar variance $\Sigma_f(\bT_1)= 1-\Big(\sum_{n=0}^{N-1} \abs{f(n)}^4\Big)^2$, for $f\in\cH_1$. For scale, we define the global scalar variance $\Sigma_f(\bT_2)=\s_f(\bT_2)$. As a result of the construction in Subsection \ref{Global uncertainties2},  $\Sigma_f(\bT_1)$ and $\Sigma_f(\bT_2)$ are invariant on orbits. We define the global uncertainty of the window $f$ by
$S(f)=w_1\Sigma_f(\bT_1)+w_2\Sigma_f(\bT_2)$ for some weights $w_1,w_2\in\RR_+$.  

\section{Uncertainty minimizers as sparsifying windows}
\label{Uncertainty minimizers as sparsifying windows}

In this section we show how the global localization framework lends itself to estimating the localization of ambiguity functions, which control the sparsity of the wavelet transform in some sense. 

\subsection{Ambiguity functions}
Consider a SPWT.  Given a window $f\in \cH$, it's ambiguity function is defined to be $V_f[f]$. The ambiguity function accommodates an important property given in Proposition \ref{Amb_is_kernel} below. This property relies on convolution of $L^2(G)$ functions, defined for $F,Q\in L^2(G)$ by
	\[[F*Q](g)= \int F(q^{-1}\bullet g)Q(q)d\mu(q).\]
		Here, $d\mu(q)$ is the left Haar measure of $G$.
To gain intuition on $F*Q$, we can adopt the usual signal processing interpretation of convolution. Namely, $F*Q$ is interpreted as ``filtering, or blurring, $Q$ using the kernel $F$''.		The following proposition can be found e.g in \cite{Fuhr_wavelet}.
		
\begin{proposition}
\label{Amb_is_kernel}
Let $f\in\cH$ be an admissible window with $\norm{Af}=1$, where $A$ is the Duflo-Moore operator (see 7 of Assumption \ref{ass_voice}, and Remark \ref{ass_voice_square}). Consider the image space of the wavelet transform, $V_f[\cH]$. Then
	$V_fV_f^*$ is the orthogonal projection $L^2(G)\rightarrow V_f[\cH]$. Moreover, for any $Q\in L^2(G)$, $V_fV_f^*[Q] =V_f[f]*Q$.
\end{proposition}

The following is a result of Proposition \ref{Amb_is_kernel}.
	\begin{corollary}
	\label{Amb_is_kernel2}
	Let $f\in\cH$ be an admissible window with $\norm{A f}=1$. Then the image space of the wavelet transform, $V_f[\cH]$, is a reproducing kernel Hilbert space with kernel $V_f[f]$. Precisely, for any $Q\in V_f[\cH]$, $Q= V_f[f]*Q$.
	\end{corollary}

By the ``blurring'' interpretation of the convolution, Corollary \ref{Amb_is_kernel2} is interpreted as follows. ``Any function in $V_f[\cH]$ is blurry, with the blurring kernel $V_f[f]$'', or ``the pixels of $V_f[\cH]$ are based on the point spread function $V_f[f]$''. This interpretation demonstrates the utility in well localized ambiguity functions. Namely, the more $V_f[f]$ is localized, the less each value of $V_f[s]$ is correlated with its neighbors, and thus the more ``information'' each value of $V_f[s]$ carries. 

{ Many papers studied the spread of the ambiguity function in phase space, in the special case of the STFT. For example, \cite{Amb_uncertainty} extended time-frequency uncertainty principles to the ambiguity function, and \cite{Amb_mini} proposed a variational method for minimizing the spread of the ambiguity function. As opposed to our approach, that studies the ambiguity function for general wavelet transforms, these papers are restricted to the STFT.  
The spread of the STFT ambiguity function plays
an important role in many applications, for example in RADAR
and coding applications \cite{Amp_app3}\cite{Amp_app4}\cite{Amp_app2}, and in operator
approximation by Gabor multipliers \cite{Amp_app1}.}

{
In the following subsection we show, in the general case,} that the more localized the ambiguity function of the window is, the more sparsifying the wavelet transform is in some sense. For that we discuss sparse signals in the context of generalized  wavelet transforms.

\subsection{Sparse signals and separation preservation}
Let us consider the following model for sparse signals. Let $f$ be a window with $\norm{Af}=1$, where $A$ is the Duflo-Moore operator (see 7 of Assumption \ref{ass_voice}, and Remark \ref{ass_voice_square}). Let $\d_{g_0}(g)=\d(g_0^{-1}\bullet g)$ be a translated delta functional in phase space. We define a \textbf{sparse phase function} as a finite combination of translated delta functionals, namely
\[F = \sum_{n=1}^N c_n\delta_{g_n}\]
For some $N\in\NN$ and coefficients $c_n\in\CC$. 
We define the synthesis of $F$ to be the signal
\[s=V_f^*F=\sum_{n=1}^N c_n\pi(g_n)f \ \ \in \cH.\]
This is consistent with the wavelet inversion formula (\ref{eq:wave_inversion1}). We call such an $s$ a \textbf{sparse signal}. It is easy to see that the wavelet transform of a sparse signal is given by
\[[V_fV_f^*F](g) = \sum_{n=1}^N c_nV_f[f](g_n^{-1}\bullet g),\]
which is consistent with Proposition \ref{Amb_is_kernel}.

Note that the wavelet transform of a sparse signal is a blurring of the sparse phase function with the ambiguity function. Thus, the better the ambiguity function is localized, the better $V_fV_f^*F$ preserves the separation of $F$. Preserving separation is an important property for greedy sparse algorithms e.g. matching pursuit \cite{Matching_pursuit}, explained next. Given a sparse signal $s$, based on the points $\{g_n\}_{n=1}^N$ we want to calculate it's sparse representation in phase space $F$, basing our calculation on the wavelet transform. In matching pursuit, we initialize $s_0=s$, and at each step $k$  pick the largest wavelet coefficient $c_k=V_f[s_k](g'_k)$ and its position in phase space $g'_k\in G$. Then, we define the remainder $s_{k+1}=s_k-c_k\pi(g'_k)f$, and continue the process until $c_k$ is sufficiently small. Now, the better localized $V_f[f]$ is, the better $V_f[s]$ retains the separation of $F$, keeping the peaks of $V_f[s_k]$ as close as possible to $\{g_n\}_{n=k}^N$. As a result, we expect matching pursuit  to perform better the more $V_f[f]$ is localized.

\subsection{Localization of ambiguity functions}

{ In this subsection we use the global localization framework to relate} the localization of $V_f[f]$ with the variances $\s^{{\bf W}_m}_f({\bf \bT}_m)$. Namely, we relate the variances  $\s^{{\bf W}_m}_f({\bf \bT}_m)$ with decay estimates of $V_f[f]$. 
Note that by (\ref{eq:no_Z_localization}), no localization of $V_f[f]$ is possible in  the direction of the center $Z$. Thus, in this section we study the localization of $V_f[f]$ in the cross-section $G_z$, defined in (\ref{eq:center_cross}).

Let us start with a toy example. Consider the group $G=\{\RR,+\}$, the space $L^2(\RR)$, the left translation representation $\pi(g)=L(g)$, and the observable $\bQ f(x)=xf(x)$. This representation is not a SPWT on its own, as it is reducible, but we can think of $\pi(g)$ as a restriction of a representation of a ``bigger'' group to a subgroup.
The ambiguity function of $f\in L^2(\RR)$ satisfies $V_f[f](g) = f*\tilde{f}(g)$, where $\tilde{f}(x) = \overline{f}(-x)$. It is intuitive that the more localized $f$ is, the more localized $V_f[f]$ is. One way to see this is by a corollary of the Chebyshev inequality \cite{Adjoint}. Namely, for $f,h$ with $e_f(\bQ)=e_1$, $e_h(\bQ)=e_2$, $\s_f(\bQ)=\s_1$ and $\s_h(\bQ)=\s_2$, we have
\begin{equation}
\abs{\ip{f}{h}}\leq \frac{2 \sqrt{\s_1}}{\abs{e_1-e_2}} + \frac{2 \sqrt{\s_2}}{\abs{e_1-e_2}}+\frac{4 \sqrt{\s_1\s_2}}{\abs{e_1-e_2}^2}.
\label{eq:chebyshev}
\end{equation}
As a result of (\ref{eq:chebyshev}), we have
\begin{equation}
\abs{V_f[f]}(g)=\ip{f}{\pi(g)f}\leq \frac{4 \sqrt{\s_f(\bQ)}}{\abs{g}} + \frac{4 \s_f(\bQ)}{\abs{g}^2}.
\label{eq:decay_R}
\end{equation}
The bound in (\ref{eq:decay_R}) involves two parts. One part are the decay terms $\frac{4 }{\abs{g}}$ and $\frac{4}{\abs{g}^2}$, independent of the choice of $f$. The other part are the constants $\s_f(\bQ)$ and $\s_f(\bQ)^2$, which we control.
Thus, the smaller the variance of $f$ is, the more localized the bound of $V_f[f]$ is.

{ The global localization framework of SPWTs lends itself to a generalization of the above decay estimation approach. Indeed, the group representation $\pi(g)$ translates the expected values in a structured way, corresponding to the group product. Moreover, the manifold structure of $G$ as a direct product of physical quantities, together with the $quantity_m$ transforms, allow the use of Chebyshev inequality.}
Consider the construction in Subsection \ref{Solving the global canonical commutation relation}. Using Theorem \ref{Theorem:pull_trans2}, we are able to pull forward the discussion on the decay of the ambiguity function to the $quantity_m$ spaces $L^2(N_m\times{\cal Y})$, $m=1,\ldots,M$. In these spaces, $\pi_m$ are mapped to left translations, so we can use standard versions of multidimensional Chebyshev inequalities. In the following we show how to reduce the analysis to the case where ${\cal Y}=\{1\}$.

Let $\Psi_m$ be the $quantity_m$ transform guaranteed by Theorem \ref{Theorem:pull_trans2}, and $L^2(N_m\times{\cal Y})$ the $quantity_m$ domain. We have ${\bf\bT}_m=\Psi_m^*{\bf\bQ}_{X_m}\Psi_m$, where ${\bf\bQ}_{X_m}=(\bQ_{X_m}^1,\ldots,\bQ_{X_m}^{K_m})$ is a tuple of multiplicative operators along the $N_m$ axis of $N_m\times{\cal Y}_m$. Let $\bQ_m^{{\bf w}_m}=\sum_{k=1}^{K_m}w_m^k\bQ_{X_m}^k$ be a linear combination of the observables in ${\bf\bQ}_{X_m}$.
 Let us denote, by abuse of notation, the functions $N_m\times{\cal Y}\rightarrow \CC$, defined to be $({\bf g}_m,y)\mapsto \sum_{k=1}^{K_m}w_m^kg_m^k$,  by 
\[\bQ_m^{{\bf w}_m}({\bf g}_m,y)=\bQ_m^{{\bf w}_m}({\bf g}_m)=\sum_{k=1}^{K_m}w_m^kg_m^k.\]
Both the expected value and the variance of $\bQ_m^{{\bf w}_m}$ are based on integrations of the form
\begin{equation}
\iint_{N_m\times{\cal Y}}R({\bf g}_m,y)\abs{f({\bf g}_m,y)}^2d{\bf g}_m dy,
\label{eq:ii8uu96}
\end{equation}
where the functions $R:N_m\times{\cal Y}\rightarrow\CC$ is $\bQ_m^{{\bf w}_m}$ for the expected value, and $\abs{\bQ_m^{{\bf w}_m}-e_f(\bQ_m^{{\bf w}_m})}^2$ for the variance.
Consider the function $F\in L^2(N_m)$, defined by
\begin{equation}
F({\bf g}_m)= \sqrt{\int_{\cal Y} \abs{f({\bf g}_m,y)}^2 dy }.
\label{eq:4kko9uygfd}
\end{equation}
 By Fubini's theorem on (\ref{eq:ii8uu96}), we may wright
\begin{equation}
e_f(\bQ_m^{{\bf w}_m})= \int_{N_m} \bQ_m^{{\bf w}_m}({\bf g}_m)\abs{F({\bf g}_m)}^2 d{\bf g}_m = e_F(\bQ_m^{{\bf w}_m})
\label{eq:yyyyy5rk}
\end{equation}
\begin{equation}
\s_f(\bQ_m^{{\bf w}_m})= \int_{N_m} \abs{\bQ_m^{{\bf w}_m}({\bf g}_m)-e_f(\bQ_m^{{\bf w}_m})}^2\abs{F({\bf g}_m)}^2 d{\bf g}_m = \s_F(\bQ_m^{{\bf w}_m}).
\label{eq:iyyo1q}
\end{equation}
This calculation shows that we may reduce the analysis in the following two subsections, to the case where ${\cal Y}=\{1\}$.
Thus, in the following we assume without loss of generality that $\Psi_m:\cH\rightarrow L^2(N_m)$, and denote by $L_m({\bf g}_m)$ the left translation in $L^2(N_m)$.

\subsubsection{The self-adjoint case}
\label{Decay:The self-adjoint case}

Let $\bf{e}$ denote the unit element in each of the groups $N_m$, $m=0,\ldots,M$. 
Fix and index $m\geq 1$, and assume that $G_m$ is $\RR$ or $\ZZ$.
let $g'\in G_z$ be a point in phase space with corresponding coordinate  ${\bf g}_m'=\bf{e}$. Let us denote, with abuse of notation, a generic point $g\in G$ 
having coordinates ${\bf g}_{m'}=\bf{e}$ for every $m'\neq m$, by $g_m$.
Our goal is to analyze the decay of $V_f[f](g_m \bullet g')$ as $g_m$ varies.

Denote $q=\Psi_m f$ the mapping of $f$ to the $quantity_m$ domain $L^2(N_m)$, and note that $\Psi_m$ intertwines $\pi(g)$ in $\cH$ with the representation $\rho_m(g)=\Psi_m\pi(g) \Psi_m^*$ in $L^2(N_m)$. We have $\rho_m({\bf g}_m)=L_m({\bf g}_m)$, where $L_m({\bf g}_m)$ is the left translation in $L^2(N_m)$. Moreover, by ${\bf\bT}_m=\Psi_m^*{\bf \bQ}_m\Psi_m$, and by Propositions \ref{gamma_dilation0} and \ref{Prop:scalar_var_translation1}, we have
\begin{align}
\label{eq:5asr}
{\bf e}_{\rho_m(g_m\bullet g')q }({\bf \bQ}_m) & =  {\bf e}_{\pi_m(g_m\bullet g')f }({\bf \bT}_m)=  {\bf g}_m\bullet {\bf A}_m({\bf h}'_m) {\bf e}_{f}({\bf \bT}_m) = {\bf g}_m\bullet {\bf A}_m({\bf h}'_m) {\bf e}_{q}({\bf \bQ}_m),\\
{\s}^{\bf W}_{\rho_m(g_m\bullet g')q }({\bf \bQ}_m)  & ={\s}^{\bf W}_{\pi(g_m\bullet g')f}({\bf \bT}_m)
={\s}^{{\bf A}_m({\bf h}'_m){\bf W}{\bf A}_m({\bf h}'_m)^*}_{f}({\bf \bT}_m) = {\s}^{{\bf A}_m({\bf h}'_m){\bf W}{\bf A}_m({\bf h}'_m)^*}_{q}({\bf \bQ}_m),\\
\abs{V_f[f](g_m \bullet g')} & =\ip{q}{\rho_m(g_m\bullet g')q},
\end{align}
for any weight matrix ${\bf W}$. 

To relate the variance of $f$ with the decay of $V_f[f]$, we start by considering a simple type of decay. We demand optimal decay of $V_f[f]$ along the directions of the axis of $N_m$. For any $k=1,\ldots,K_m$, consider the weight ${\bf W}_k={\bf w}_k{\bf w}_k^*$ corresponding to the standard direction ${\bf w}_k$ with entries $w_k^j=\delta_{k,j}$.
Now, by (\ref{eq:chebyshev}), (\ref{eq:5asr}), and (\ref{scalar_var_translation1}), for any $k=1,\ldots,K_m$ we have
\begin{align}
\label{eq:LAST34}
\abs{V_f[f](g_m \bullet g')}= \abs{\ip{q}{\rho_m(g_m)\rho_m(g')q}} \leq &\\
\frac{2 \sqrt{\s_f(\bT_m^k)}}{\abs{{ e}_f(\bT_m^l)-g^k_m\bullet[ {\bf A}_m({\bf h}'_m) {\bf e}_{f}({\bf \bT}_m)]_k}} \ +\ & \frac{2 \sqrt{\s^{A_m({\bf h}'_m){\bf w}_k}_{f}({\bf \bT}_m)}}{\abs{{ e}_f(\bT_m^l)-g^k_m\bullet [{\bf A}_m({\bf h}'_m) {\bf e}_{f}({\bf \bT}_m)]_k}}\\
& +\frac{4 \sqrt{\s_f(\bT_m^k)\s^{A_m({\bf h}'_m){\bf w}_k}_{f}({\bf \bT}_m)}}{\abs{e_f(\bT_m^l)-g^k_m\bullet [{\bf A}_m({\bf h}'_m) {\bf e}_{f}({\bf \bT}_m)]_k}^2}
\end{align}
where $[ A_m({\bf h}'_m) {\bf e}_{f}({\bf \bT}_m)]_k$ is the $k$-th entry of the vector ${\bf A}_m({\bf h}'_m) {\bf e}_{f}({\bf \bT}_m)$. In case ${\bf e}_{f}({\bf \bT}_m)={\bf 0}$, (\ref{eq:LAST34}) reduces to
\begin{align}
\label{eq:LAST342}
\abs{V_f[f](g_m \bullet g')} \leq &\\
\frac{2 \sqrt{\s_f(\bT_m^k)}}{\abs{g^k_m}} \ +\ & \frac{2 \sqrt{\s^{A_m({\bf h}'_m){\bf w}_k}_{f}({\bf \bT}_m)}}{\abs{g^k_m}}
 +\frac{4 \sqrt{\s_f(\bT_m^k)\s^{A_m({\bf h}'_m){\bf w}_k}_{f}({\bf \bT}_m)}}{\abs{g^k_m}^2}
\end{align}
Let us interpret (\ref{eq:LAST34}) or (\ref{eq:LAST342}).
For $f$ with fixed expected values, the denominators are decay terms independent of $f$. The numerators are variances of $f$, which we can control.
Thus the decay of the bound (\ref{eq:LAST34}) and (\ref{eq:LAST342}) is faster the smaller the variances of $f$ are. 

The following Proposition extends (\ref{eq:LAST342}) to directional variances.
\begin{proposition}
\label{prop:direc_var_decay1}
Consider a SPWT, and a canonical multi-observable ${\bf \bT}$. Let $m\geq 1$ be an index such that $G_m$ is $\RR$ or $\ZZ$. Let $f$ be a window with $e_f({\bf\bT}_m)={\bf 0}$, and let ${\bf w}^k\in\CC^{K_m}$ be a direction. Then
\begin{align}
\label{eq:LAST34255}
\abs{V_f[f](g_m \bullet g')} \leq &\\
\frac{2 \sqrt{\s^{{\bf w}_m}_f({\bf\bT_m})}}{\abs{{\bf w}_m\cdot {\bf g}_m}} \ +\ & \frac{2 \sqrt{\s^{{\bf A}_m({\bf h}'_m){\bf w}_m}_{f}({\bf \bT}_m)}}{\abs{{\bf w}_m\cdot {\bf g}_m}}
 +\frac{4 \sqrt{\s^{{\bf w}_m}_f({\bf\bT}_m)\s^{{\bf A}_m({\bf h}'_m){\bf w}_m}_{f}({\bf \bT}_m)}}{\abs{{\bf w}_m\cdot {\bf g}_m}^2}
\end{align}
\end{proposition}
In Proposition \ref{prop:direc_var_decay1}, the vector ${\bf w}_m$ is the direction in which we apply the corollary of Chebyshev inequality  (\ref{eq:chebyshev}), and ${\bf g}_m$ is the direction in which we bound the decay.
As a first application of Proposition \ref{prop:direc_var_decay1}, we are interested in fast decay in the directions of the axis of $N_m$. 
Note that the variances of elements in the orbit of $f$ appear in (\ref{eq:LAST34255}). Therefore, we want to choose a window with a minimal global variance $\Sigma^{{\bf I}_m}_f({\bf \bT}_m)$. 
As a second example, let us demand fast isotropic decay in $N_m$. We decompose the coordinate vector ${\bf g}_m=\abs{{\bf g}_m}\hat{{\bf g}}_m$, where $\hat{{\bf g}}_m$ is a unit vector. Assuming that ${\bf e}_{f}({\bf \bT}_m)={\bf 0}$, we get by Proposition \ref{prop:direc_var_decay1},
\begin{align}
\abs{V_f[f](g_m \bullet g')} \leq &\\
\frac{2 \sqrt{\s^{\hat{{\bf g}}_m}_f({\bf \bT}_m)}}{\abs{{\bf g}_m}} \ +\ & \frac{2 \sqrt{\s^{{\bf A}_m({\bf h}'_m)\hat{{\bf g}}_m}_{f}({\bf \bT}_m)}}{\abs{{\bf g}_m}}
 +\frac{4 \sqrt{\s_f^{\hat{{\bf g}}_m}({\bf \bT}_m)\s^{{\bf A}_m({\bf h}'_m)\hat{{\bf g}}_m}_{f}({\bf \bT}_m)}}{\abs{{\bf g}_m}^2}.
\label{eq:LAST343}
\end{align}
Since we are interested in fast isotropic decay, we decrease the numerators of (\ref{eq:LAST343}) by minimizing $\Sigma^{{\bf 1}_m}_f({\bf \bT}_m)$ where ${\bf 1}_m$ is the isotropic weight function with all entries equal to a single positive constant.

\subsubsection{The unitary case}

In the unitary case, where $G_m$ is equal to $e^{i\RR}$ or $e^{2\pi i\NN/N}$ for some $m\geq 1$, we can derive a corresponding version of Chebyshev inequality. Since the discussion can be pulled forward to $L^2(N_m)$ as described above, we assume without loss of generality that $\cH=L^2(N_m)$ and $\pi|_{N_m}({\bf g}_m)=L_{m}({\bf g}_m)$ is the left translation. We start with the case $G_m=e^{i\RR}$. 
In the following we show an equivalence of standard localization notions in $\RR^2$, and the localization notions $e_f(\bQ_m^k)$ and $\s_f(\bQ_m^k)=1-\abs{e_f(\bQ_m^k)}^2$ based on observables.

First we consider a procedure for mapping $f\in L^2(e^{i\RR})$ to a function in $L^2(\RR^2)$. 
Consider the standard embedding of $e^{i\RR}$ to the unit circle in $\RR$, 
\[e^{i\RR}\ni x+iy\mapsto \nu(x+iy)= (x,y)\in \RR^2.\]
 Let $f\in L^2(e^{i\RR})$, and $\e>0$. Denote by $f_{\e}\in L^2(\RR^2)$ the function, defined in polar coordinates, by
\[f_{\e}\big(r\nu(e^{i\theta})\big)=\left\{
\begin{array}{ccc}
	f(e^{i\theta}) & , & \abs{r-1}\leq\e \\
	 0            &  ,  & \abs{r-1}>\e
\end{array}
\right.\]

Next we show how to relate the localization of $f$ to the localization of $f_{\e}$. Denote generic points in $\RR^2$ by ${\bf x}=(x,y)$. Define the standard expected values of $h\in L^2(\RR^2)$ by
\[X_{h}= \iint_{\RR^2}x\abs{h({\bf x})}^2 d{\bf x}\]
\[Y_{h}= \iint_{\RR^2}y\abs{h({\bf x})}^2 d{\bf x}\]
and denote ${\bf X}_{h}=(X_{h},Y_{h})$, where $h\in L^2(\RR^2)$.
Define the isotropic variance
\[D_{h}=\iint_{\RR^2}\abs{{\bf x}-{\bf X}_{h}}^2\abs{h({\bf x})}^2 d{\bf x}.\]
It is easy to see that
\begin{equation}
{\bf X}_{f_{\e}}=\nu\big(e_f(\bQ)\big) +o_{\e}(1) \quad , \quad D_{f_{\e}}=\s_f(\bQ) +o_{\e}(1),
\label{eq:X_ft2e_f}
\end{equation}
where $o_{\e}(1)$ converges to zero as $\e\rightarrow 0$. Note that by $\s_f(\bQ)=1-\abs{e_f(\bQ)}^2$, (\ref{eq:X_ft2e_f}) also relates the variance of $f$ to the expected value of $f_{\e}$.

Since our goal is to derive a version (\ref{eq:chebyshev}) for unitary observables, and since localization in $L^2(e^{i\RR})$ relates to localization in $L^2(\RR^2)$, our next goal is to derive a version of (\ref{eq:chebyshev}) to some Chebyshev inequality in $L^2(\RR^2)$.
Let $B_r({\bf X}_{h})$ be a disc of radius $r$ about ${\bf X}_{h}$. The standard isotropic Chebyshev inequality in $L^2(\RR^2)$ reads
\begin{equation}
\iint_{B_r({\bf X}_{h})^c}\abs{h({\bf x})}^2d{\bf x} \leq \frac{D_{h}}{r^2}.
\label{eq:isi_cheby1}
\end{equation}
\begin{lemma}
\label{iso_cheby_lemma}
 Let $h_1,h_2\in L^2(\RR^2)$, and denote $r=\abs{{\bf X}_{h_1}-{\bf X}_{h_2}}/2$. Then
\[\ip{h_1}{h_2}\leq \frac{\sqrt{D_{h_1}}}{r}+\frac{\sqrt{D_{h_2}}}{r}+\frac{\sqrt{D_{h_1}}\sqrt{D_{h_2}}}{r^2}.\]
\end{lemma}
\begin{proof}
We have
\[\iint_{\RR^2} h_1({\bf x})\overline{h_2({\bf x})}d{\bf x} = \]
\[\iint_{B_r({\bf X}_{h_1})} h_1({\bf x})\overline{h_2({\bf x})}d{\bf x}
+\iint_{B_r({\bf X}_{h_2})} h_1({\bf x})\overline{h_2({\bf x})}d{\bf x}
+\iint_{B_r({\bf X}_{h_1})^c\cap B_r({\bf X}_{h_2})^c} h_1({\bf x})\overline{h_2({\bf x})}d{\bf x}.\]
Therefore, by the Cauchy Schwarz inequality, by the monotonicity of integrals of nonnegative functions, and by the isotropic Chebyshev inequality (\ref{eq:isi_cheby1}),
\begin{equation}
\begin{split}
\iint_{\RR^2} h_1({\bf x})\overline{h_2({\bf x})}d{\bf x} \leq &
		\sqrt{\iint_{\RR^2} \abs{h_1({\bf x})}^2d{\bf x}}\sqrt{\iint_{B_r({\bf X}_{h_2})^c}\abs{h_2({\bf x})}d{\bf x}} \\
+&\sqrt{\iint_{B_r({\bf X}_{h_1})^c} \abs{h_1({\bf x})}^2d{\bf x}}\sqrt{\iint_{\RR^2}\abs{h_2({\bf x})}d{\bf x}} \\
+&\sqrt{\iint_{B_r({\bf X}_{h_1})^c} \abs{h_1({\bf x})}^2d{\bf x}}\sqrt{\iint_{B_r({\bf X}_{h_2})^c} \abs{h_2({\bf x})}^2d{\bf x}} \\
\leq & \frac{\sqrt{D_{h_1}}}{r}+\frac{\sqrt{D_{h_2}}}{r}+\frac{\sqrt{D_{h_1}}\sqrt{D_{h_2}}}{r^2}
\end{split}
\label{eq:yyy666}
\end{equation}
\end{proof}

We can now relate the decay of $V_f[f]$ to variances, using Lemma \ref{iso_cheby_lemma} on $h_1=f_{\e}$ and $h_2=[\pi(g)f]_{\e}$. 
Recall that we assume that for a specific $m\geq 1$, $N_m=G_m=e^{i\RR}$ is one dimensional. Let $e$ denote the unit element in $e^{i\RR}$. As in Subsection \ref{Decay:The self-adjoint case}, 
let $g'\in G$ be a point in phase space with corresponding coordinate  ${ g}_m'={e}$, and denote a generic point $g\in G$ 
having coordinates ${\bf g}_{m'}={\bf e}$ for every $m'\neq m$, by $g_m$.
Our goal is to analyze the decay of $V_f[f](g_m \bullet g')$ as $g_m$ varies.
Similarly to the proof of Proposition \ref{Prop:scalar_var_translation2}, using the linearity of the expected values with respect to ${\bQ}_m$, and by the multi-canonical commutation relation (\ref{global_observe}),  
\[{ e}_{\pi(g_m\bullet g')f}({ \bQ}_m)= { g}_m \bullet { e}_{f}\big({\bf A}_m({\bf h}_m'){ \bQ}_m\big).\]
Moreover, by ${ \s}_{\pi(g_m\bullet g')f}({ \bQ}_m)=1-\abs{{ e}_{\pi(g_m\bullet g')f}({ \bQ}_m)}^2$, we have
\[{ \s}_{\pi(g_m\bullet g')f}({ \bQ}_m)= { \s}_{f}\big({\bf A}_m({\bf h}_m'){ \bQ}_m\big).\]
Denote ${\Delta}_{g_m}({ \bQ}_m)={ e}_{\pi(g_m\bullet g')f}({ \bQ}_m)-{ e}_{f}({ \bQ}_m)$. By (\ref{eq:X_ft2e_f}), we have
\[ {r}_{g_m}=\frac{1}{2}\abs{\Delta_g({ \bQ}_m)}  = \frac{1}{2}\abs{{\bf X}_{f_{\e}}-{\bf X}_{[\pi(g_m\bullet g')f]_{\e}}}+ o_{\e}(1).\]
By Lemma \ref{iso_cheby_lemma} on $f_{\e}$ and $[\pi(g_m\bullet g')f]_{\e}$,
and by (\ref{eq:X_ft2e_f}), we have
\begin{equation}
\begin{split}
V_f[f](g_m\bullet g') = & \ip{f}{\pi(g_m\bullet g')f}\\
\leq &\frac{\sqrt{\s_f({ \bQ}_m)} }{r_{g_m}} + \frac{\sqrt{\s_{f}({\bf A}_m({\bf h}_m'){ \bQ}_m)} }{r_{g_m}} + \ \frac{\sqrt{\s_f({ \bQ}_m)} \sqrt{\s_{f}({\bf A}_m({\bf h}_m'){\bQ}_m)}  }{r_{g_m}^2} +o_{\e}(1). 
\end{split}
\label{eq:fffgdfg54gq0}
\end{equation}
This is true for every $\e>0$, so we must have
\begin{equation}
V_f[f](g_m\bullet g') \leq \frac{\sqrt{\s_f({ \bQ}_m)} }{r_{g_m}} + \frac{\sqrt{\s_{f}({\bf A}_m({\bf h}_m'){ \bQ}_m)} }{r_{g_m}} + \ \frac{\sqrt{\s_f({ \bQ}_m)} \sqrt{\s_{f}({\bf A}_m({\bf h}_m'){\bQ}_m)}  }{r_{g_m}^2}. 
\label{eq:fffgdfg54gq}
\end{equation}
To control the decay rate of (\ref{eq:fffgdfg54gq}), we want to minimize the variances of the windows in the orbit of $f$. Therefore, we minimize the global variance $\Sigma_f({\bf \bQ}_m)$.

Next we explain the way to extend (\ref{eq:fffgdfg54gq}) to the $K_m$-dimensional case, where $N_m=[e^{i\RR}]^{K_m}$.
In the $K_m$-dimensional case,  $f\in L^2([e^{i\RR}]^{K_m})$. It is easy to extend the above results to a correspondence between $L^2([e^{i\RR}]^{K_m})$ and $L^2(\RR^{2K_m})$, based on the mapping $\nu:[e^{i\RR}]^{K_m}\rightarrow\mathbb{T}^{K_m}$, where $\mathbb{T}^{K_m}$ is the unit torus in $\RR^{2K_m}$. 
A Chebyshev inequality with a ``torus symmetry'', extending Lemma \ref{iso_cheby_lemma} to $L^2(\RR^{2K_m})$, can be derived. Here, the Chebyshev inequality is isotropic in 2D subspaces of $L^2(\RR^{2K_m})$ spanned by pairs of axis.
The resulting decay estimate is given in the following Proposition.
\begin{proposition}
\label{Prop:uni_decay1}
Consider a SPWT, and a canonical multi-observable ${\bf\bT}$. Let $m\geq 1$ be an index such that $G_m=e^{i\RR}$, let $f$ be a window with ${\bf e}_f({\bf\bQ}_m)={\bf e}$, and let ${\bf w}_k$ be a standard direction, with entries $w_k^j=\delta_{k,j}$. Let $g_m\bullet g'\in G$ as before, let $\boldsymbol{\Delta}_{{\bf g}_m}({\bf \bQ}_m)={\bf e}_{\pi(g_m\bullet g')f}({\bf \bQ}_m)-{\bf e}_{f}({\bf \bQ}_m)\in \CC^{K_m}$, and let
\[ {\bf r}_{{\bf g}_m}=\frac{1}{2}\abs{\boldsymbol{\Delta}_{{\bf g}_m}({\bf \bQ}_m)} \in \RR^{K_m}.\]
Then
\begin{equation}
V_f[f](g_m\bullet g') \leq \frac{\sqrt{\s_f({ \bQ}^k_m)} }{\abs{{\bf r}_{{\bf g}_m}}} + \frac{\sqrt{\s^{{\bf w}_k}_{f}({\bf A}_m({\bf h}_m'){\bf \bQ}_m)} }{\abs{{\bf r}_{{\bf g}_m}}} + \ \frac{\sqrt{\s_f({\bQ}^k_m)} \sqrt{\s^{{\bf w}_k}_{f}({\bf A}_m({\bf h}_m'){\bf \bQ}_m)}  }{\abs{{\bf r}_{{\bf g}_m}}^2}. 
\label{eq:fffgdfg54gq45}
\end{equation}
\end{proposition}

As before, Proposition \ref{eq:fffgdfg54gq45} leads us to minimize the global variance $\Sigma_f({\bf \bQ}_m)$.

\begin{example}
Let us consider the special case where ${\bf A}_m({\bf h}'_m)={\bf I}$ and $N_m=G_m=e^{i\RR}$. Here, $\pi(g_m\bullet g')$ is translation in $e^{i\RR}$ by $g_m$. As a result, by the law of cosines, $r_g= \abs{\Delta_g({\bf \bT}_m)}/2$ satisfies
\[r_{g_m}=\frac{1}{2}\sqrt{\big(1-\s_f({\bf \bQ}_m)\big)\big(2-2\cos(g_m)\big)}.\]
Now, equation (\ref{eq:fffgdfg54gq}) takes the form
\[V_f[f](g_m\bullet g') = \ip{f}{\pi(g_m\bullet g')f}\]
\begin{equation}
\leq 2\frac{\sqrt{\s_f({\bf \bQ}_m)} }{\sqrt{\big(1-\s_f({\bf \bQ}_m)\big)\big(2-2\cos(g_m)\big)}} + 4\frac{\s_f({\bf \bQ}_m) }{\big(1-\s_f({\bf \bQ}_m)\big)\big(2-2\cos(g_m)\big)}. 
\label{eq:f7k3j}
\end{equation}
The only controllable part in the decay (\ref{eq:f7k3j}) is $\s_f({\bf \bQ}_m)$, which shows that minimal variance corresponds to optimally decaying ambiguity function.
\end{example}

The case where $G_m=e^{2\pi i \ZZ/N}$ is treated similarly, by embedding each $e^{2\pi i n/N}$ in $\nu(e^{2\pi i n/N})$ in the unit circle in $\RR^2$, and extending $f(e^{2\pi i n/N})$ in a small disc about $\nu(e^{2\pi i n/N})$ to get a function $f_{\e}\in L^2(\RR^2)$.

\section*{Acknowledgments}

This research was supported in part by the EU FET Open grant UNLocX: Uncertainty
principles versus localization properties, function systems for efficient coding schemes (Grant agreement
no 255931).

\bibliographystyle{plain}	
\bibliography{Ref_uncertainty2}

\appendix

\section{Direct integrals}

A direct integral of Hilbert spaces is a generalization of a direct product. The idea is that instead of using a finite set for the carrier space (the index set of the Hilbert spaces), we use a measure space.
We introduce the theory in a very restricted case which is of importance to us. 
\begin{definition}
\label{direct_space}
Let $\cY$ be a measure space, and consider the Hilbert space $\cH=L^2(\cY)$. Let $\cX$ be another measure space.
The \textbf{direct integral} of $\cH$, over the carrier space $\cX$, is denoted by $\int_{\cX}^{\oplus}\cH d\mu(x)$, and defined to be
	\[\int_{\cX}^{\oplus}\cH d\mu(x) \cong L^2(\cX \times \cY).\]
	For a vector $f\in \int_{\cX}^{\oplus}\cH d\mu(x)$ and $x\in\cX$, we denote in short $f(x) = f(x,\cdot)$.
\end{definition}
Note that Definition \ref{direct_space} extends the notion of direct product. Indeed, a vector $f\in\cH^N=L^2(\cY)^N$ can be thought of as a function that maps each index $1\leq n \leq N$ to a vector $f_n\in L^2(\cY)$. In direct integrals the index set is $\cX$, and $f\in \int_{\cX}^{\oplus}\cH d\mu(x)$ is the function that maps indices $x\in\cX$ to vectors $f(x,\cdot)\in L^2(\cY)$.

One of the main endeavors of representation theory is to describe any arbitrary representation as a combination of explicit ``simple'' representations, served as building blocks.
For this end, there is a way to decompose certain classes of representations of $G$ to a direct integral of irreducible representations. To formulate this statement, the index running over the different representations in the decomposition is in a measure carrier space, and the Hilbert space on which the decomposed representation acts is a direct integral. 
We assume in our analysis that $G$ is a physical quantity.
In this case, the carrier space is the space of irreducible unitary representations of $G$, namely the characters $\chi(G)$. There is a way to define a measure on this carrier space, called a Plancherel measure, that admits the desired decomposition. 
\begin{definition}
Let $G$ be a physical quantity.
The \textbf{Plancherel measure} of $\chi(G)$ is the standard Lebesgue measure of $\chi(G)$, considered as the physical quantity $\hat{G}$.
\end{definition}

Given a character $\chi$ of a physical quantity $G$, and $m\in \NN \cup\{\infty\}$, we denote by $\chi^{[m]}$ the representation in $\CC^m$ defined by
\[\chi^{[m]}(g) (z_1,\ldots,z_{m}) =(\chi(g) z_1,\ldots, \chi(g) z_{m}). \]
In case $m=\infty$, the notation $\CC^{\infty}$ means $l^{\infty}$.
\begin{definition}
Let $G$ be a physical quantity.
Let $m\in \NN \cup\{\infty\}$. 
Consider the Hilbert space
\[\cH=\int_{\chi(G)'}^{\oplus}\CC^{m} d\mu(\chi).\]
The representation
\begin{equation}
\rho = \int_{\chi(G)}^{\oplus} \chi^{[m]}\  d\mu(\chi)
\label{eq:direct_rep}
\end{equation} 
in the space $\cH$, 
is defined by
\[\left[\rho(g)f\right](\chi) = \chi^{[m]}(g) f(\chi) \]
for any $f\in\cH$ and $\chi\in \chi(G)$ (almost everywhere). 

Let $\rho'$ be a representation of $G$ on $\cH_{\rho'}$, unitarily equivalent to $\rho$. Then
 $\rho$ is called a direct integral decomposition of $\rho'$, and $m$ is called the \textbf{multiplicity} of the decomposition.
\end{definition}

The informal idea in this limited definition, is that some representations $\rho'$ of physical quantities contain each character of $G$ as an irreducible subrepresentation, with a constant multiplicity $m$ over all irreducible subrepresentations. It is not accurate to say that the characters $\chi$ are subrepresentations of $\rho'$ in the sense that they are unitarily equivalent to the restrictions of $\rho'$ to invariant subspace of $\cH_{\rho'}$. However, in the language of direct integrals, we are able to say that each character $\chi$ of $G$ appears in $\rho'$ with multiplicity $m$ in the sense of (\ref{eq:direct_rep}).

The following uniqueness theorem can be found in its general form in Theorem 3.25 of \cite{Fuhr_wavelet}.

\begin{proposition}
Let
\[\int_{\chi(G)}^{\oplus} \chi^{[m]} d\mu(\chi) \cong \int_{\chi(G)}^{\oplus} \chi^{[m']} d\mu(\chi)\]
be two unitarily equivalent representations of a physical quantity $G$. Then $m=m'$.
\label{unique_direct_int}
\end{proposition}

Last, we give a direct integral decomposition of a useful representation.
Let $G$ be a physical quantity, and consider the left translation $L(g)$ in $L^2(G)$. 
Consider an isomorphism $\hat{G}\mapsto \chi(G)$, $\hat{g}\mapsto \chi_{\hat{g}}$. 
Let $\cF_G$ be the Fourier transform between $L^2(G)$ and $L^2(\hat{G})$. Namely, for $f\in L^1(G)\cap L^2(G)$ we have 
\[[\cF_G f](\hat{g})=\int_G f(g)\overline{\chi_{\hat{g}}(g)} dg = \int_G f(g)\chi_{\hat{g}^{-1}}(g) dg,\]
 and $\cF_G f$ is defined by a density argument for $f\in L^2(G)$. Here, $\hat{g}$ are interpreted as frequencies.
The Fourier transform $\cF_G$ transforms translations $L(g)$ to a modulation operator, namely
\[[\cF_G L(g)f](\hat{g}) = \chi_{\hat{g}^{-1}}(g)[\cF_G f](\hat{g}).\]
 As a result, we have the representation equivalence
\begin{equation}
L \cong \int_{\chi(G)}^{\oplus} \chi^{[1]} d\mu(\chi)
\label{eq:direct_decomposition_of_translation}
\end{equation}
To see this, we interpret the carrier space $\chi(G)$ as the reflected frequency domain. The isometric isomorphism $U$ between the representation spaces of $L$ and $\int_{\chi(G)}^{\oplus}\chi^{[1]} d\mu(\chi)$ is given by $[Uf](\chi_{\hat{g}}) = [\cF_G f](\hat{g}^{-1})$, which gives 
\[[U L(g) f](\chi_{\hat{g}}) = \chi_{\hat{g}}(g)[\cF_G f](\hat{g}^{-1})=\chi_{\hat{g}}(g)[Uf](\chi_{\hat{g}}).\]

\end{document}